\documentclass[12pt,fleqn]{amsart}
\usepackage{epsf,amscd,amsmath,amssymb,amsfonts,mathtools}
\usepackage{graphicx}
\usepackage{epstopdf}
\usepackage{fullpage}
\newcommand{\be}{\begin{equation}}
\newcommand{\ee}{\end{equation}}
\newcommand{\bea}{\begin{eqnarray}}
\newcommand{\eea}{\end{eqnarray}}
\newcommand{\beas}{\begin{eqnarray*}}
\newcommand{\eeas}{\end{eqnarray*}}

\theoremstyle{plain}
\newtheorem{thm}{Theorem}
\newtheorem{lem}[thm]{Lemma}
\newtheorem{rem}[thm]{Remark}
\newtheorem{cor}[thm]{Corollary}
\newtheorem{prop}[thm]{Proposition}

\theoremstyle{definition}
\newtheorem{defn}[thm]{Definition}

\newtheorem{ex}[thm]{Example}

\numberwithin{thm}{section}
\numberwithin{equation}{section}

\newcommand{\A}{{\mathbb A}}

\def\One{\mathbb{I}}

\usepackage{float}
\floatstyle{boxed} 
\restylefloat{figure}
\usepackage{hyperref}

\title{Algebraic Interplay between Renormalization and Monodromy}
\author{Dirk Kreimer and Karen Yeats}
\address{Humboldt U.\ Berlin and University of Waterloo}
\begin{document}
\maketitle
\begin{abstract}
We investigate combinatorial and algebraic aspects of the interplay between renormalization and monodromies for Feynman amplitudes. We clarify how extraction of subgraphs from a Feynman graph interacts with putting edges onshell or with contracting them to obtain reduced graphs. Graph by graph this leads to a study of cointeracting bialgebras.  One bialgebra comes from extraction of subgraphs and hence is needed for renormalization. The other bialgebra is an incidence bialgebra for edges put either on- or
offshell. It is hence related to the monodromies of the multivalued function to which a renormalized graph evaluates. Summing over infinite series of graphs, consequences for Green functions are derived using combinatorial Dyson--Schwinger equations. 
\end{abstract}
\tableofcontents
\section{Introduction}
In perturbative quantum field theory Green functions are regarded as formal sums of Feynman graphs each of them contributing to a chosen renormalized Green function $G_R=\sum_{\Gamma}\Phi_R(\Gamma)$ from which amplitudes $\mathsf{A}$ are derived.
To be more precise any Feynman graph $\Gamma$  in such a sum over graphs is evaluated by renormalized Feynman rules $\Phi_R$. The evaluation $\Phi_R(\Gamma)$ leads, for each graph, to a multi-valued function $\Phi_R(\Gamma)(\{m_e\},\{q_l\})$ which depends on the set of   masses $\{m_e\}$
assigned to internal edges of the Feynman graph and also on the set of momenta $\{q_l\}$ assigned to 
external edges. 
The latter are represented as half-edges $l$ of the graph
labeled by momentum vectors $q_l$ describing a momentum of a particle incoming at any such  half-edge. Each graph $\Gamma$ contributing to a given amplitude provides the same set $L$ of labeled half-edges representing the external particles.\footnote{A vertex which has no such half-edge can equivalently be regarded as a vertex where such an external half-edge $l$ is attached with zero external momentum $q_l=0$.}

We assume the evalution gives the function $\Phi_R(\Gamma)$ as a scalar under the Lorentz group.\footnote{Theories involving spin can be treated similarly but are left to future work.} 
Hence, after the evaluation we get a multi-valued function of all Lorentz scalars $q_e\cdot q_f$, $e,f\in L$ which we can form from scalar products of external momenta.\footnote{For a vector $r=(r_0,r_1,r_2,r_3)^T\in \mathbb{M}^4(\mathbb{C})$
  we have the Lorentz scalar $r\cdot r= r^2=r_0^2-r_1^2-r_2^2-r_3^2$, $r_j^2=r_j\bar{r}_j$, $r_j=\Re(r_j)+\imath \Im(r_j)\in \mathbb{C}$, $\bar{r}_j=\Re(r_j)-\imath \Im(r_j)$, $j\in \{0,1,2,3\}$.}
We regard masses as fixed, given parameters in this context. We let $Q^L(\mathbb{R})$ be the real vector space spanned by the independent scalar
products $q_e\cdot q_f$ and $Q^L(\mathbb{C})$ its complexification (see Section~\ref{QL}). See Figure~\ref{triangleKin} for an example.
\begin{figure}[h]
\includegraphics[width=6cm]{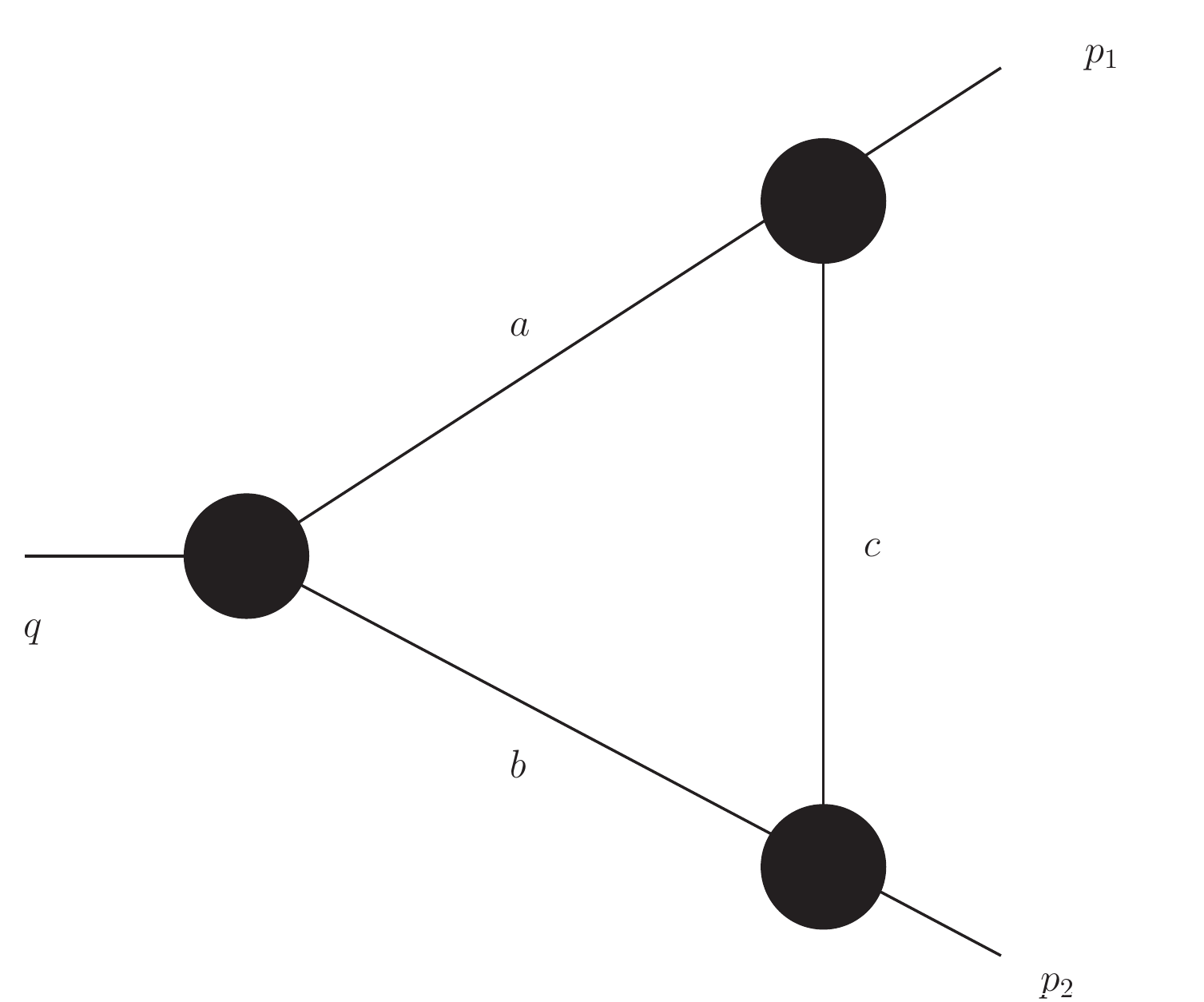}.
\caption{The one-loop triangle graph. The set $L$ is provided by three half-edges labeled by $q,p_1,p_2$. These denote three  vectors in four-dimensional complex Minkowski space $\mathbb{M}^4(\mathbb{C})$, subject to the side-constraint $q+p_1+p_2=0$. Feynman rules assign to such a graph $\Gamma$ a function $\Phi_R(\Gamma)(q^2, p_1^2, p_2^2,m_a^2,m_b^2,m_c^2)$.
The space $Q^L(\mathbb{R})$ is a $\mathbb{R}$-vectorspace spanned by $z_1:=q^2, z_2:=p_1^2, z_3:=p_2^2$ where we use that $2p_1\cdot p_2=q^2-p_1^2-p_2^2$. The mass squares $m_i^2$ are kept fixed with 
$\Re(m_i^2)\geq 0$, $0<-\Im(m_i^2)\ll 1$, $i\in \{a,b,c\}$. 
One is particularly interested in the variation of $\Phi_R(\Gamma)$ when 
$q,p_1,p_2$ and hence the $z_i$ vary in or near the real locus provided by $\mathbb{M}^4(\mathbb{R})$ and $Q^L(\mathbb{R})$.
}
\label{triangleKin}
\end{figure}

Typically, we are interested in the behaviour of such a multi-valued function when we vary  a chosen variable $s$ which defines a vector in $Q^L$ whilst keeping the other variables fixed and real.  The most typical form of the variable $s$ is described briefly below.

Each Feynman graph contributing to the same amplitude or Green function  provides such a multi-valued function  of $s$.
The Green functions of the theory can be expanded to be sums over such graphs.   Each Green function is multi-valued as a function of $s$ and is a solution of a fixed-point equation, a Dyson-Schwinger equation.
To understand the monodromy of such functions one studies the behaviour when internal edges are on the mass shell. Following $S$-matrix theory \cite{Smatrixbook}, we partition the set $L$ into two sets
$L_{in},L_{out}$ and define
\[
s=\left( \sum_{e\in L_{in}} q(e)\right)^2=\left( \sum_{e\in L_{out}} q(e)\right)^2,
\]
using momentum conservation $\sum_{e\in L_{in}} q(e)=-\sum_{e\in L_{out}} q(e)$.  The choices of $s$ that we will be interested in will be of this form.  More details on $s$ can be found in Appendix~\ref{appC}.

In the spirit of the idea outlined above for any graph $\Gamma$ contributing to a chosen amplitude $\mathsf{A}$,
consider a set $I$ of internal edges of $G$ with the property that the removal of the edges in $I$ separates the graph in two parts $\Gamma_{in}$ and $\Gamma_{out}$ such that the external particles coupled to one part form the set $L_{in}$ and the others $L_{out}$.
\[
\Gamma-I=\Gamma_{in}\cup \Gamma_{out},  L_{\Gamma_{in}}=L_{in}, L_{\Gamma_{out}}=L_{out}, 
\]
and such that no $\Gamma-H$ with $H\subsetneq I$ decomposes $\Gamma$ with the same $L_{in}$ and $L_{out}$, that is $I$ is minimal.

Such a set $I$ we call a cut for the partition into $L_{in},L_{out}$. The sum over all cuts
for a given partition contributes to the monodromy of that partition
as a function of 
\[
s=\left(\sum_{e\in L{in}}q_e \right)^2.
\]

This is not enough to understand the monodromy of $\Phi_R(\Gamma)(s)$. We have to refine the partition further. This came as a surprise in the early days of quantum field theory. The monodromy of an amplitude regarded as a function of $s$ can not be understood by the two-partion into $in$ and $out$ states as suggested by $S$-matrix theory. 

Anomalous thresholds, first discovered in particle physics experiments, appear. 
See for example Sec.(3) in the review by Amati and Fubini \cite{AmatiFubini}.  These anomalous thresholds are related to monodromies corresponding to finer partitions of $L$ \cite{BlochKreimerOS}. Below we will study the generic situation refining partitions until $L$ is separated into $|L|$ elements
and consider the fixed point equations for Green functions for any prescribed partition of $L$.  These fixed point equations are the cut analogues of Dyson-Schwinger equations.  We will describe their structure.

A partition of $L$ into two parts has a normal physical threshold $s_0\in\mathbb{R}$ as its physical
observable corresponding to a discontinuity in the function $\Phi_R(\Gamma)(s)$ at 
$s=s_0$, and finer partitions into $k\gneq 2$ parts  give rise to anomalous thresholds  $s_i\in\mathbb{R}$, $i\geq 1$ for this function $\Phi_R(\Gamma)(s)$. Such a partition of $L$ into two or more parts generates a partition of a graph $\Gamma$ into subgraphs $\Gamma_j$ for each part $j$. Each $\Gamma_j$ must be renormalized to render $\Phi_R(\Gamma)(s)$ well-defined.

The quest to understand the interplay between renormalization and the monodromies generated by $\Phi_R(\Gamma)$ through such partitions motivates this paper. We use that partitions of $L$ are realized by removing edges while the  monodromies are stable when shrinking remaining edges \cite{Smatrixbook,MarkoDirk}.

For renormalization we know that the sum of all graphs contributing to a Green function has a distinguished structure: the coproduct $\Delta_{core}$  closes when acting on (combinatorial) Green functions \cite{Karenbook}. 

Here we show that this coexists with the analytic structure of graphs.  In view of the above considerations, the analytic structure of the thresholds as discussed in \cite{MarkoDirk} can be studied through an algebraic avatar: the combinatorics of the interplay of reduced or removed edges. These basic operations of contracting an edge or removing an edge assign to a graph $\Gamma$ a lower triangular matrix $M$, and similarly for any
sum of graphs.

There is a corresponding incidence coalgebra based on reduction or removal of edges \cite{coaction} which is represented as
\[
\rho\left[(M)_{ij}\right]=\sum_k (M)_{ki}\otimes (M)_{jk}. 
\]
The coaction of $\Delta_{core}$ on Cutkosky graphs and the coproduct $\rho$ form cointeracting bialgebras giving an algebraic formulation of the interplay between renormalization and monodromy.
The cut Dyson-Schwinger equations also interact well with the coproduct.  This is a second algebraic manifestation of the underlying compatibility of renormalization and the monodromies.

Over all, we clarify the interplay between renormalization and monodromy: we can renormalize by local counterterms in Green functions $G_R(q)$, $q\in Q^L(\mathbb{R})$, which are multi-valued as functions of physical observables.

\subsection*{Organization of the paper}

Section~\ref{sec graph summary} briefly introduces graphs and cuts of graphs as we wish to formulate them.  More details are given in Appendix~\ref{graph appendix}.  Section~\ref{Hopf} proceeds to define the relevant Hopf algebras for cut and uncut graphs.  This is how renormalization is brought in algebraically.  To understand the interplay of renormalization and monodromy we need to understand the interplay of the coproducts and the cuts.  This is done in Section~\ref{coactions} with the structure of coactions. Some information on Feynman rules which is helpful to understand the significance of these coactions is collected in Appendix~\ref{appC}. Notably, in Subsection~\ref{cointbi} we use that we have cointeracting bialgebras encapsulating this interplay.  The formal definitions of the two participatng bialgebras and their cointeraction are treated in Appendix~\ref{appB}. Mathematically, this is closely related to the interplay between motivic and deRham classes and physically, this is closely related to sector decomposition.  Arriving at the cut Dyson--Schwinger equations in Section~\ref{DSE} we develop the set up and prove how the coproducts and coactions act on the Green functions.  Mathematically, this is closely related to the notion of assembly maps.  Finally, in Section~\ref{sec concl} we conclude. References can be found after the appendices.

\subsection*{Acknowledgments}
DK is supported by grant KR1401/5-2 of the DFG.  
KY is supported by an NSERC Discovery grant and the Canada Research Chairs program. She was supported by the Humboldt Foundation as a Humboldt fellow during the development of this work.  KY would like to thank DK and Humboldt University  for  hosting  her  visit  to  Berlin  as  a  Humboldt  Fellow.
Both authors would like to thank Marko Berghoff, Spencer Bloch, Michael Borinsky, David Broadhurst, Lo\"ic Foissy, John Gracey, Ralph Kaufmann, Nick Olson-Harris, Henry Ki\ss ler and Karen Vogtmann for  helpful discussions.

\section{Graphs and their spanning forests}\label{sec graph summary}

\subsection{Graphs} We want to study graphs with cuts in two different ways, as graphs with certain edges marked as cut and potentially certain vertices partitioned into pieces, and as pairs $(\Gamma, F)$ of a graph $\Gamma$ and a spanning forest $F$.  The first formulation is best suited to cut versions of Dyson-Schwinger equations as we will discuss in Section~\ref{DSE}, while the second is best suited to understanding the cointeraction of renormalization and cut structures as we will discuss in Section~\ref{coactions}.
In short, we want to determine the equations for Green functions which describe scattering.

To get there we need formal notions of all these different kinds of graphs.  In this section we will overview these ideas at an intuitive level.  Appendix~\ref{graph appendix} gives a formalization of these notions that is well suited to our purposes, and includes details.

Graphs for us are based on half edges.  That is, a graph is a set of \emph{half edges} and the information of how these half edges are paired to make \emph{internal edges} of the Feynman graph, and are collected together into \emph{corollas} of at least 3 half edges to make the vertices of the Feynman graph.  Half edges that are not paired into internal edges are \emph{external edges}.  Both the groupings into vertices and the groupings into edges can be represented as partitions of the set of half edges, and so we will write our graphs as $\Gamma=(H_\Gamma, \mathcal{V}_\Gamma, \mathcal{E}_\Gamma)$ where $H_\Gamma$ is the set of half edges, $\mathcal{V}_\Gamma$ is the partition of $H_\Gamma$ giving the vertices and $\mathcal{E}_\Gamma$ is the partition of $H_\Gamma$ giving the internal and external edges.  For details see Section~\ref{set graphs}.  Sometimes we may want to think of an element of $\mathcal{V}_\Gamma$ simply as a vertex and sometimes we will want to think it as a corolla, that is the set of half edges which define the vertex.  When it is useful to emphasize the corolla as opposed to the vertex, we will write $\mathbf{c}_v$ for the corolla of the vertex $v$.

It is worth emphasizing that our graphs have no vertices of degree less than 3, but they may have multiple edges and self-loops.
We define $L_\Gamma$ to be the set of external edges and let $e_\Gamma$, $l_\Gamma$, and $v_\Gamma$ be the number of internal edges, external edges, and vertices of $\Gamma$ respectively.

For a connected graph $\Gamma$ we write $|\Gamma|=|H^1(\Gamma)| = e_\Gamma-v_\Gamma+1$, the number of independent loops, or the dimension of the cycle space of $\Gamma$.  For a disjoint union of graphs $\Gamma_1$, $\Gamma_2$ we define $|\Gamma_1\dot{\cup} \Gamma_2| = |\Gamma_1| + |\Gamma_2|$.

We write $\Gamma-e$ for the graph $\Gamma$ with the edge $e$ cut, but unlike in usual graph theory, when we cut an edge we do not remove the half edges forming it, we simply disconnect those two half edges so they no longer form an edge.  An edge is a \emph{bridge} if removing it increases the number of connected components and a graph is \emph{bridgeless} if it has no bridges.

We write $\Gamma/e$ for the graph $\Gamma$ with the edge $e$ contracted. We think of edge contraction intuitively as shrinking the edge to length $0$.  For a formal definition in our set up see Section~\ref{sec graphs}.

We will use the same notation for cutting or contracting sets of edges: $\Gamma-X, \Gamma/X$ for $X$ a set of edges.  Furthermore if $X$ is a subgraph we use the same notation for cutting or contracting the edges of $X$.

Given a spanning tree $T$ (see Section~\ref{subsec span} for the definition of spanning tree) of a graph $\Gamma$ and an edge $e$ of $\Gamma$ which is not in $T$, the graph $T\cup e$ has a unique cycle known as the \emph{fundamental cycle} given by $e$ and $T$.  We will denote this fundamental cycle $l(T,e)$.  Fundamental cycles will be quite important in Section~\ref{cointbi}.

\subsection{Cuts}

We are interested in cutting edges so as to disconnect graphs.  From a physicist's viewpoint the cut edges can also be regarded as marked edges which are put on-shell when we apply Feynman rules.
For the Dyson-Schwinger equations we will also be interested in splitting (partitioning)  corollas.  This is a different kind of cut.

We will introduce the vector space $H_C$ generated by Cutkosky graphs, which are graphs which have cuts generated by a removal of edges.  The base graph $\Gamma$ is also allowed to vary.

In particular, we will study series on such graphs $\Gamma\in H_C$ which have cuts all corresponding to a chosen partition of a given common set of external edges $L$. Such series can be obtained as solutions to fixed point equations formulated using pre-Cutkosky graphs $H_{pC}$. Pre-Cutkosky graphs are graphs which may have cuts both by cutting edges and by also splitting internal vertices $v$  by partitions of the corollas $\mathbf{c}_v$, as will be described in more detail below.

Fortunately, both the cut edges and split corollas can be represented nicely in our formulation of graphs because they can be given by refinements of the partitions $\mathcal{E}$ and $\mathcal{V}$.  That is given a graph $\Gamma=(H_\Gamma, \mathcal{V}_\Gamma, \mathcal{E}_\Gamma)$ we represent cuts of both types by giving a second graph with the same half edges $H = (H_\Gamma, \mathcal{V}_H, \mathcal{E}_H)$, where $\mathcal{V}_H$ refines $\mathcal{V}_\Gamma$ and $\mathcal{E}_H$ refines $\mathcal{E}_\Gamma$.  See Section~\ref{subsec cuts} for details.

We define a \emph{pre-cut graph} to be such a pair $\Gamma=(\Gamma,H)$ and by abuse of notation we call both the pre-cut graph and the underlying graph (before any cutting) $\Gamma$.  The reason for this is that we think of a pre-cut graph $\Gamma$ as being the ordinary graph $\Gamma$ with the extra information of the cut.  We continue to use $|\Gamma|$ on a pre-cut graph as before on the underlying uncut graph.  For $\Gamma=(\Gamma,H)$, we write $\|\Gamma\|$ for $|H|$.

A pre-cut graph is a \emph{cut graph} if no vertices are split, so the only cuts are to edges.

For quantum field theory, as is common in graph theory, we do not want our edge cuts to include edges with both ends in the same component after the cut.  To this end we define a \emph{pre-Cutkosky graph} to be a pre-cut graph where each cut edge has the property that the two ends are in different components of the graph after the cut and we define a \emph{Cutkosky graph} to be a cut graph that is pre-Cutkosky.  See Section~\ref{sec span} for futher details.  $H_{pC}$ and $H_{C}$ are respectively the $\mathbb{Q}$ vector spaces spanned by bridgeless pre-Cutkosky graphs and bridgeless Cutkosky graphs.  We will write $H_{core}$ for the $\mathbb{Q}$ vector space spanned by bridgeless graphs in the original sense.  All three of these can be upgraded to free commutative $\mathbb{Q}$ algebras generated by the connected graphs and where the commutative product is disjoint union.

As is explored further in Section~\ref{sec span} we are also interested in spanning forests of cut graphs and we say a spanning forest is compatible with the cut if its components induce the same cut.  Sometimes it is better to think of a cut in terms of a compatible forest rather than as the cut itself, and so, in some sections, instead of cut graphs we work with pairs $(\Gamma,F)$ of a graph in the original sense and a spanning forest of the graph.

\subsection{Sub- and co-graphs}

For the Hopf algebras in the next section it will be important to have approprite notions of subgraphs and co-graphs in each of these contexts.

In the case of graphs in the original sense, we require only that our subgraphs are bridgeless and are full at each vertex in the sense that if a vertex appears in the subgraph then its whole corolla must appear.  The co-graph is then simply the graph arising from contracting the edges of the subgraph, while also removing two-valent vertices by un-subdividing any pair of edges joined by a two-valent vertex, that is, replacing them by a single edge.

In the case of pre-cut graphs, as well as needing to be bridgeless and full at the vertices, we require that our subgraphs only have the cuts they inherit, no additional edges are cut, and no vertices are cut more than before.  The co-graphs are then well defined with the understanding that if the subgraph itself is cut then the vertex or edge it forms in the co-graph is cut in the same way.

For a pair $(\Gamma,F)$ of a graph and a forest, given a subgraph $\gamma\subseteq \Gamma$ then we have a graph forest pair $(\gamma, F\cap \gamma)$.  However, we will only count this as a subgraph if $F/(F\cap \gamma)$ is a forest of $\Gamma/\gamma$.  In this case we get the graph forest pair $(\Gamma/\gamma, F/(\Gamma\cap \gamma))$ as the co-graph and so we can form the coproduct that we desire.  For more details including examples see Section~\ref{subco}.

Finally, dually to taking sub- and co-graphs we can insert one graph into another.  This reverses the above operations, see Section~\ref{composing} for details.

\section{Hopf algebras}\label{Hopf}
Hopf algebras play an important role in perturbative QFT because they allow us to organize the recursive structure of the expansion in terms of Feynman graphs. This has been particularly well-studied in the context of renormalization theory and the associated forest formula of Zimmermann \cite{K,CK,RHI,RHII,anatomy,Michibook}. 
Also, the Dyson--Schwinger equations
were identified as fixed-point equations in the corresponding Hochschild cohomology of the Hopf algebra of renormalization for any renormalizable field theory \cite{BergKr,Karenbook,FoissyDSE}.

Here, we generalize such a setup to Cutkosky graphs and graphs with forests, starting from the core Hopf algebra \cite{core,BV}. We can also obtain quotient Hopf algebras by restricting the allowable vertices. The renormalization Hopf algebras 
appear as special cases of such quotient Hopf algebras.
This supports a future study of towers of renormalizable theories \cite{John1,John2} which provide interactions corresponding to vertices going far beyond interactions studied in renormalizable field theories so far.

We concentrate on Hopf algebras dedicated to describing scattering.
Hence, we work in the arena of graphs which have cuts as 
studied by Cutkosky \cite{Cutkosky}.
We need to study (pre-) Cutkosky graphs, their Hopf algebra structure and corresponding series over such graphs, providing the combinatorial backbone for the study of variations of Green functions. 
We also extend our study to the study of pairs $(\Gamma,F)$ of graphs $\Gamma$ and spanning forests $F$. The associated graph complexes relate naturally to the study of Cutkosky rules \cite{BlochKreimerOS,MarkoDirk,BerghM}.

As we saw in the previous section, the Cutkosky situation and the graph-and-forest situation are quite similar.  Both should be viewed as defining graphs with cuts, but the forest provides the additional information of a spanning tree in each piece after cutting.  If we gather together all the $(\Gamma,F)$ pairs where $F$ gives the same cut in $\Gamma$, then this equivalence class of $(\Gamma,F)$ pairs contains the same information as the Cutkosky graph for this cut.

\medskip
There are three $\mathbb{Q}$-Hopf algebras, $H_{core}$ of core graphs, $H_{pC}$ of pre-Cutkosky graphs, and $H_{GF}$ of graphs-forest pairs which we will consider. Furthermore we will define $\mathbb{Q}$-vector spaces $H_C$ of Cutkosky graphs and $H_{nC}$ of non-Cutkosky graphs.

As variants it would be possible to also consider $H_{V}$, $H_{pC,V}$, $H_{C,V}$ and $H_{nC,V}$ obtained by restricting to quotient algebras defined by restricting to graphs with vertices of a given valence prescribed by a set $V$.  We will not elaborate on this, but will provide the framework for these quotient algebras.

The Hopf algebras induce various coactions which we will discuss below 
in Section~\ref{coactions} in particular with regards to their interpretation in physics.
\subsection{The core Hopf algebra $H_{core}$}
The core Hopf algebra $H_{core}$ \cite{core,BV} is based on the $\mathbb{Q}$-vector space generated by connected bridgeless Feynman graphs.
It is graded by the loop number $|\Gamma|$.

We define  a commutative product
\[
m: H_{core}\otimes H_{core}\to H_{core},\, m(\Gamma_1,\Gamma_2)=\Gamma_1\dot{\cup} \Gamma_2,
\]
by disjoint union. The unit $\One$ is provided by the empty set, $|\One|=0$, so that we get a connected free commutative $\mathbb{Q}$-algebra with bridgeless graphs as generators.  So that the algebra is connected we take a quotient by identifying all graphs consisting of a single vertex with external edges with $\One$.\footnote{We do not regard the dot $\cdot$ as a core graph of grade zero. In fact isolated vertices (without external edges) are not even permissible in our definition of graph, as our vertices are parts in a vertex partition and so they are nonempty subsets of half edges.  If we instead began with a more conventional definition of graph, then for the core Hopf algebra we would need to take a larger quotient given  by the ideal generated by graphs containing any non-vanishing number of isolated vertices.  These structures can also be studied in the non-connected context  \cite{Michibook}.}

We define a coproduct by
\be\label{coreco}
\Delta_{core}(\Gamma)=\Gamma\otimes\One+\One\otimes \Gamma+\sum_{\gamma}\gamma\otimes \Gamma/\gamma,
\ee
where the sum is over all $\gamma\in H_{core}$ such that $\gamma\subsetneq \Gamma$, where these are subgraphs in the sense of Section~\ref{subco} with $\Gamma$ viewed as a trivial pre-cut graph $(\Gamma,\Gamma)$.  In particular subgraphs are bridgeless and full at the vertices.
In other words there are connected bridgeless graphs $\gamma_i$ such that $\gamma=\dot{\cup}_i \gamma_i$, and $\Gamma/\gamma$ denotes the co-graph in which all internal edges of all $\gamma_i$ shrink to zero length in $\Gamma$ and 2-valent vertices are replaced by edges, see Section~\ref{subco} and  Equations~\ref{edgecont} and \ref{lengthcont}. 

We have a counit $\hat{\One}:H_{core}\to\mathbb{Q}$ which annihilates any non-empty graph and $\hat{\One}(\One)=1$ and we have the antipode
$S:H_{core}\to H_{core}$, $S(\One)=\One$
\[
S(\Gamma)=-\Gamma-\sum_{\gamma\subsetneq \Gamma}S(\gamma) \Gamma/\gamma.
\]
Furthermore our Hopf algebras are graded by the loop order,
\[
H_{core}=\oplus_{j=0}^\infty H_{core}^{(j)},\,H_{core}^{(0)}\cong \mathbb{Q}\One\,\text{and}\, \mathbf{Aug}_{core}=\oplus_{j=1}^\infty H_{core}^{(j)}.
\]
and $h\in H_{core}^{(j)}\Leftrightarrow |h|=j$.
Such Hopf algebras are the dual of a universal enveloping algebra of a Lie algebra which
originates from a pre-Lie algebra. See \cite{MarkoDirk} for a discussion of 
such pre-Lie algebras.

\begin{rem}\label{necklace}
Necklaces in the following sense are the primitives for the core Hopf algebra.

In combinatorics a \emph{necklace} of length $n$ over an alphabet $A$ is an equivalence class of length $n$ strings over $A$, where the equivalence is under rotations.  For example, over the alphabet $\{a,b\}$ there are 16 words of length $4$ but only 6 necklaces $\{aaaa,aaab,aabb,abab,abbb,bbbb\}$.  The words $abbaab$ and $baabba$ are reflections of each other but not rotations of each other; they are different necklaces.\footnote{Equivalence classes of words under both rotation and reflection are often called \emph{bracelets}.}

For the case of uncut graphs and $|L_g|=n$, one loop graphs can be encoded by necklaces as follows.  Let $g$ be a one loop graph on $1\leq k\leq n$
vertices $v_j$ with $\sum_{j=1}^k (\mathbf{val}(v_j)-2)=n$, and $\mathbf{val}(v_j)\geq 3$, $\forall v_j$. 
Such a graph $g$ is a string of $k$ vertices and $k$ edges in between, and $|L_g|=n$ half edges distributed over the vertices. If two such graphs $g,h$ can be transformed into each other by a cyclic permutation of the string of vertices we consider them as equivalent, $g\sim h$,
and call this equivalence class a necklace $\omega$  Notate $\omega$ as a string of integers $n_1, \ldots, n_k$ up to cyclic permutation where $n_i$ indicates a vertex with $n_i$ external edges, that is, a vertex of degree $n_i+2$.
\end{rem}

\subsection{Quotient Hopf algebras}\label{quotient}
There are many quotient Hopf algebras originating from $H_{core}$.  The primary use we will make of the quotient Hopf algebras is in order to bootstrap easy proofs of the various cut Hopf algebras of future sections.  Renormalization Hopf algebras of Feynman graphs form another important class of examples.

The following lemma is handy.  It concerns when functions constructed with a projection after a coproduct can themselves be coproducts.
\begin{lem}\label{lem Hopf proj}
  Let $H$ be a Hopf algebra with coproduct $\Delta$.  Let $K\subseteq H$ be a subspace of $H$ and let $P:H\rightarrow K$ be a projection.  If $(P\otimes P)\Delta P = (P\otimes P)\Delta$ then  $\Delta_K: (P\otimes P)\Delta : K\rightarrow K\otimes K$ is coassociative.
\end{lem}

We will be applying this in the case, as above, where we have a distinguished basis and the projection sends certain basis elements to $0$.  In that case the condition in the lemma says that there is no way for a basis element in $H\setminus K$ to have a term in its coproduct that is in $K\otimes K$.

\begin{proof}
  \begin{align*}
    (\Delta_K\otimes \text{id})\Delta_K
    & = (P\otimes P\otimes \text{id})(\Delta \otimes \text{id})(P\otimes P) \Delta \\
    & = (P\otimes P \otimes P)(\Delta \otimes \text{id})\Delta  \quad \text{by hypothesis}\\
    & =(P\otimes P \otimes P)(\text{id} \otimes \Delta)\Delta \\
    & = (\text{id}\otimes \Delta_k)\Delta_K \quad \text{by the analogous argument.}
  \end{align*}
\end{proof}

Let us pursue the renormalization Hopf algebras of Feynman graphs in more detail as an example of how to work with quotient Hopf algebras.  The coproduct structure for these Hopf algebras is inherited from the core Hopf algebra by setting graphs with vertices of undesired valence to zero \cite{BV}.

In general let $N$ be a finite set of integers $n_i\geq 1$ and $V_N$ a corresponding set of vertices
$v_i,\,\mathbf{val}(v_i)=n_i+2$. The choices $N=\{1\}$, $N=\{2\}$, $N=\{4\}$
correspond to Hopf-algebras on $3-,4-,6-$regular graphs renormalizable in $6-,4-,3-$ dimensional spacetime respectively.

Let now $H_N$ be the sub-vectorspace of $H_{core}$ regarded as a vectorspace where we remove all graphs which have vertices 
$v$ with $(\mathbf{val}(v)-2)\not\in N$. Let $P_N$ be the projector $H_{core}\to H_N$.

Define $\Delta_N:H_N\to H_N\otimes H_N,\,\Delta_VN:=(P_N\otimes P_N)\Delta_{core}$.  

\begin{lem}
$H_N(\One,\hat{\One},m,\Delta_N,S_N)$ is a quotient Hopf algebra of $H_{core}$.
\end{lem}

Here we understand that unit $\One$, counit $\hat{\One}$ and multiplication $m$ are taken from $H_{core}$, and the antipode is similarly defined using $\Delta_N$ by the iteration 
$S_N:=-m_H\circ (S_N\otimes P_{aug})\circ \Delta_N$ where $P_{aug}$ (not to be confused with $P_N$) is the projector into the augmentation ideal 
\[
\mathbf{Aug}(H_N):=\oplus_{j=1}^\infty  H_N^{(j)},
\] 
of $H_N=\oplus_{j=0}^\infty H_N^{(j)}$.

\begin{proof}
  The previous lemma gives that $\Delta_N$ is coassociative since for any graph $\Gamma$ and subgraph $\gamma$ of $\Gamma$, each vertex of $\Gamma$ appears in either $\gamma$ or $\Gamma/\gamma$ with degree unchanged, so if $\Gamma$ has a vertex not in $N$ then so does either $\gamma$ or $\Gamma/\gamma$.

  The remaining properties are immediate from the properties of the core Hopf algebra.
\end{proof}

\subsection{The Hopf algebra $H_{pC}$}
We will define the Hopf algebra $H_{pC}(\One,\hat{\One},m,\Delta_{pC},S_{pC})$. As  a vector space it is the span of pre-Cutkosky graphs $G$.
Given a pre-Cutkosky graph $\Gamma$ we will use the notation $\hat{\Gamma}$ for the underlying uncut graph and $\tilde{\Gamma}$ for the graph after the cuts.  See Section~\ref{sec precut} for details.

$H_{pC}$ can be graded by either $|\Gamma|:=|\hat{\Gamma}|$ or $\|\Gamma\|:=|\tilde{\Gamma}|$, the first Betti number of $\tilde{G}$.  When viewing $H_{pC}$ as a Hopf algebra in its own right, as we are doing in this section, the $|\Gamma|$ grading is more useful.  Later when we work with the coaction the $\|\Gamma\|$ grading is more useful.

The Hopf algebra structure is inherited from the core coproduct $\Delta_{core}$, using that any subgraph $\gamma$ when contracting its internal edges $e\in E_\gamma$ shrinks, $\gamma\to \gamma/E_\gamma$, to a $h_0(\gamma)$-partition of its external half-edges $L_\gamma$ and forms a cut corolla.
Concretely the cut corollas are formed as follows. If $|L_\gamma|=2$, in $\Gamma/\gamma$ the two half-edges $h_1(\gamma),h_2(\gamma)$ associated to $L_\gamma$ form an edge $\in C_{\Gamma/\gamma}$ (assuming $h_0(\gamma)=2$, else the edge is uncut in $\Gamma/\gamma$).  
If $|L_\gamma|>2$, $\gamma$ shrinks to a corolla $\mathbf{c}_v$ with a new vertex $v\in V_{\Gamma/\gamma}$ with a corresponding $h_0(g)$-partition of $\mathbf{c}_v$.

To go into this in more detail, a subgraph of $\hat{\Gamma}$ which is bridgeless and full at the vertices determines a pre-cut subgraph of $\Gamma$ (see Lemma~\ref{lem precut subs}), let $f_\Gamma$ be the map from such subgraphs of $\hat{\Gamma}$ to pre-cut subgraphs of $\Gamma$.  All the subgraphs appearing in the core coproduct are bridegless and full at the vertices.

Taking $\Delta_{core}(\hat{\Gamma})$ but replacing subgraphs $\gamma$ of $\hat{\Gamma}$ with $f_G(\gamma)$ and co-graphs $\hat{\Gamma}/\gamma$ with $\Gamma/f_G(\gamma)$ 
gives us the coproduct we want.

The rest of the properties of a Hopf algebra are directly inherited from $H_{core}$.
Figures~\ref{dPc} and \ref{dPcO} give some examples.
\begin{figure}[H]
\includegraphics[width=14cm]{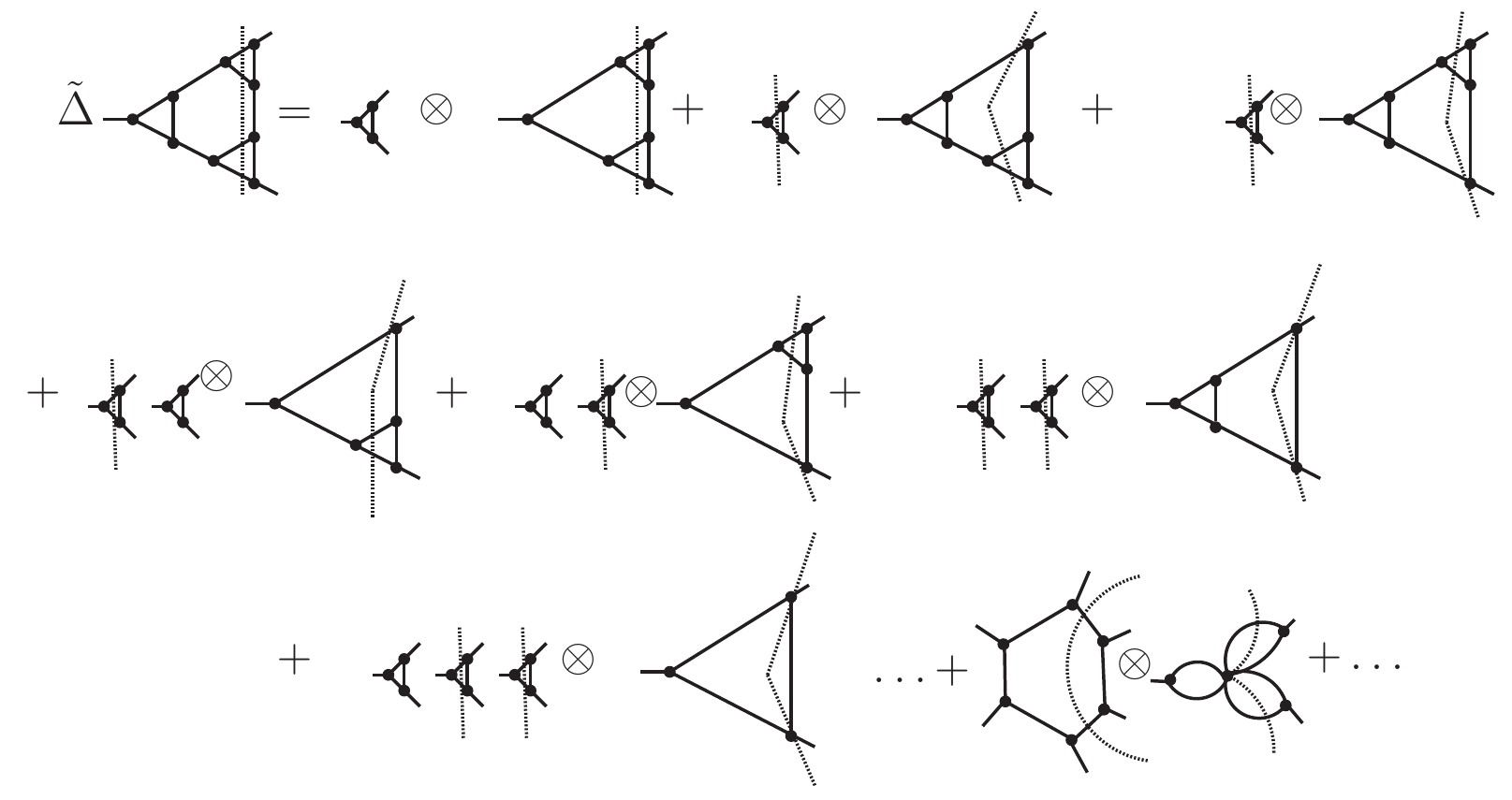}
\caption{The reduced coproduct for $\Gamma\in H_{pC}$. Even the coproduct of a graph $\Gamma$ which has no partitioned vertices has co-graphs with such partitioned vertices. Note that we omit --for space reasons-- to give all of the terms where co-graphs are generated with vertices of valence higher than three. Just one is indicated.
Similarly in the next Fig.(\ref{dPcO}).}
\label{dPc}
\end{figure}
\begin{rem}
  If we wish to impose normality conditions on the vertex cuts, that is we wish to view vertex cuts as successive refinements from a first cut into two pieces that has internal edges on both sides (see Remark~\ref{rem normal} and Section~\ref{sec co}), then at this point we need to take extra care as when we shrink a subgraph we may create new vertex cuts and we need to enforce that these are normal.

  First note that no other problems with normality can occur.  Specifically, if $\Gamma/\gamma$ satisfies the normality condition on the chains giving the refinement, then the subgraph and co-graph are both pre-cut graphs.  Furthermore since $\Gamma$ is pre-Cutkosky, every cut edge has its two ends in different components of $\tilde{\Gamma}$, so the same holds for $\gamma$ immediately, and for $\Gamma/\gamma$ it holds for all edges that are edges of $\Gamma$, and also for the edge coming from $\gamma$ since taking a compatible spanning forest of $\Gamma$ we see that the two ends of $\gamma$ are in different trees of this forest, and so the edge corresponding to them in $\Gamma/\gamma$ also has ends in two different components.

Now we will enforce the normality condition on the newly created co-graph vertices.  Let $P_{pC}$ be the projection which takes pairs $(\Gamma,H)$ with at least one non-normal vertex to $0$.  This projection is defined at the level of the preferred basis elements.  For each such $(\Gamma,H)$ and subgraph $(\gamma,h)$, each vertex of $\Gamma$ is either a vertex of $\gamma$ or of $\Gamma/\gamma$ with the same refinement.  Thus at least one of $\gamma$ or $\Gamma/\gamma$ also has a non-normal vertex.  Therefore by Lemma~\ref{lem Hopf proj}, taking $\Delta_{core}(\hat{\Gamma})$ but replacing subgraphs $\gamma$ of $\hat{\Gamma}$ with $P_{pC}(f_G(\gamma))$ and co-graphs $\hat{\Gamma}/\gamma$ with $P_{pC}(\Gamma/f_G(\gamma))$ gives us a coassociative map.
\end{rem}

\begin{figure}[H]
\includegraphics[width=14cm]{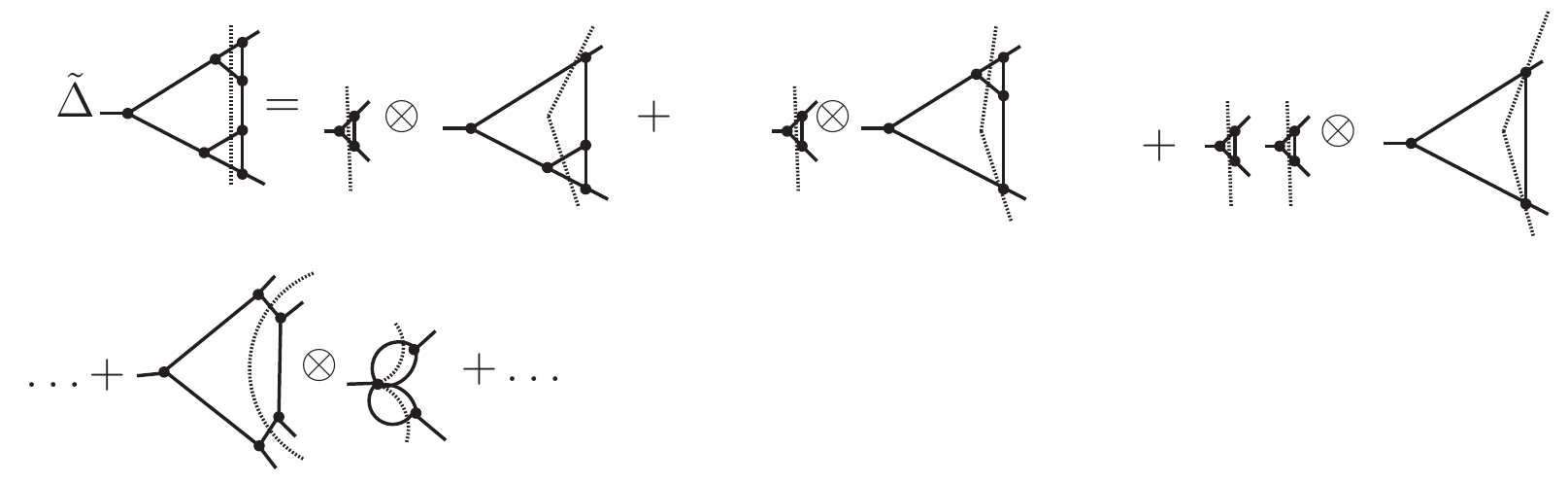}
\caption{The reduced coproduct for $\Gamma\in H_{pC}^0$, $\tilde{\Delta}_{pC}:H_{pC}^0\to H_{pC}^0\otimes H_{pC}^0$ with mostly only the terms involving vertices of degree 3  included, as before.}
\label{dPcO}
\end{figure}

Note that we also have the $\mathbb{Q}$-vector space $H_C$ generated by Cutkosky graphs.  We do not give $H_C$ a Hopf algebra structure as a co-graph of a Cutkosky graph by a Cutkosky subgraph may only be pre-Cutkosky.  This idea returns in Section~\ref{sec more coactions}.

\subsection{Pre-Cutkosky Necklaces}\label{pC necklaces}
Similarly to how core primitives can be represented by necklaces (see Remark~\ref{necklace}), pre-Cutkosky necklaces $\omega$  are the generating primitives for the Hopf algebra $H_{pC}$.

Consider a pre-Cutkosky graph $\gamma$ with $|\gamma|=1$. Such a $\gamma$ has has a unique spanning forest $F$ 
and defines a $h_0(F)$-partition $P_\omega$ of $L_\gamma$ of size $|L_\gamma|$.  Once again we can encode the features that matter for $\gamma$ using a necklace.  The only difference from the core case is that we need to have a bigger alphabet to keep track of the cuts of edges and vertices and how the external legs are partitioned across them.

There are multiple ways we could set this up with only inconsequential differences.  We will choose an alphabet with a letter for cut edges, a letter for uncut edges, and nonnegative integers as letters, indicating the number of external legs in each part of a vertex cut with the parts containing the internal edges coming first and last.  Note that this approach loses the order of the half edges at the cut vertex, but this will not be an important loss for us and it is consistent with the information we will need in order to index our Green functions in Section~\ref{DSE}.  If it were important to keep the order of the half edges at each vertex, then we could number the half edges incident to a vertex starting from the internal edge corresponding to the previous letter in the word, and then to represent the vertices in the word we'd need letters so that we could represent each arbitrary set partitions of $\{1,2,\ldots, m\}$.  This would be heavy but not a fundamental difficulty.

For example, using our convention, the graph on the right hand side of the last line of Figure~\ref{dPc} is represented by the necklace $1u01u10u$ where $u$ is the letter for uncut edges.

Each such necklace $\omega$ is a sequence alternating between $v_\gamma$ cut- or un-cut edges and of $v_\gamma$ cut- or un-cut vertices, $1\leq v_\gamma\leq |L_\gamma|$, where the cut vertices are given by a sequence of integers, one per part of the cut.

We can also view core necklaces (see Remark~\ref{necklace}) as pre-Cutkosky necklaces simply by adding the letter for the uncut edge between each $n_i$.

\subsection{Extension to a core coproduct for pairs $(\Gamma,F)$}
It is also possible to extend the Hopf algebra $H_{core}$ of graphs to a Hopf algebra $H_{GF}$ of pairs $(\Gamma,F)$ given by a graph $\Gamma$ and a spanning forest $F$ of $\Gamma$
\cite{MarkoDirk}.
The resulting coproduct is not directly analogous to the coproduct we defined on the pre-Cutkosky graphs, as the cut vertices were crucial for that definition.  The coproduct on $H_{GF}$ is more closely related to the coaction of the next section, while still being an honest coproduct.

 Let $\mathcal{F}_\Gamma$ be the set of all
spanning forests of $\Gamma$. The empty graph $\One$ has an empty spanning forest also denoted by $\One$.

 We define a $\mathbb{Q}$-Hopf algebra $H_{GF}$ for such pairs $(\Gamma,F)$ by setting 
\bea\label{HopfPairs}
\Delta_{GF}(\Gamma,F) & = & (\Gamma,F)\otimes (\One,\One)+(\One,\One)\otimes (\Gamma,F)+\nonumber\\
 &  & +\sum_{{\gamma\subsetneq \Gamma \atop F-(F\cap \gamma)\in \mathcal{F}_{\Gamma/\gamma}} \atop
 F\sim F-(F\cap \gamma)} (\gamma,\gamma\cap F)\otimes (\Gamma/\gamma, F-(F\cap \gamma)),
\eea
 where $\mathcal{F_\Gamma}$ is the set of all forests of $\Gamma$.  Additionally, by $F-(F\cap \gamma)\in \mathcal{F}_{\Gamma/\gamma}$ we mean to interpret the edges of $F-(F\cap \gamma)$ as a subgraph of $\Gamma/\gamma$ and then check if that subgraph is an element of $\mathcal{F}_{\Gamma/\gamma}$.  This ensures that 
only terms contribute such that $\Gamma/\gamma$ has a valid spanning forest.  Finally, 
by $F\sim F-(F\cap \gamma)$ we mean that the partition of external legs 
of $(\Gamma,F)$ and $(\Gamma/\gamma,F-(F\cap \gamma))$ are identical.

See Figure~\ref{deltaGF}  for an example.

\begin{figure}[H]
\includegraphics[width=14cm]{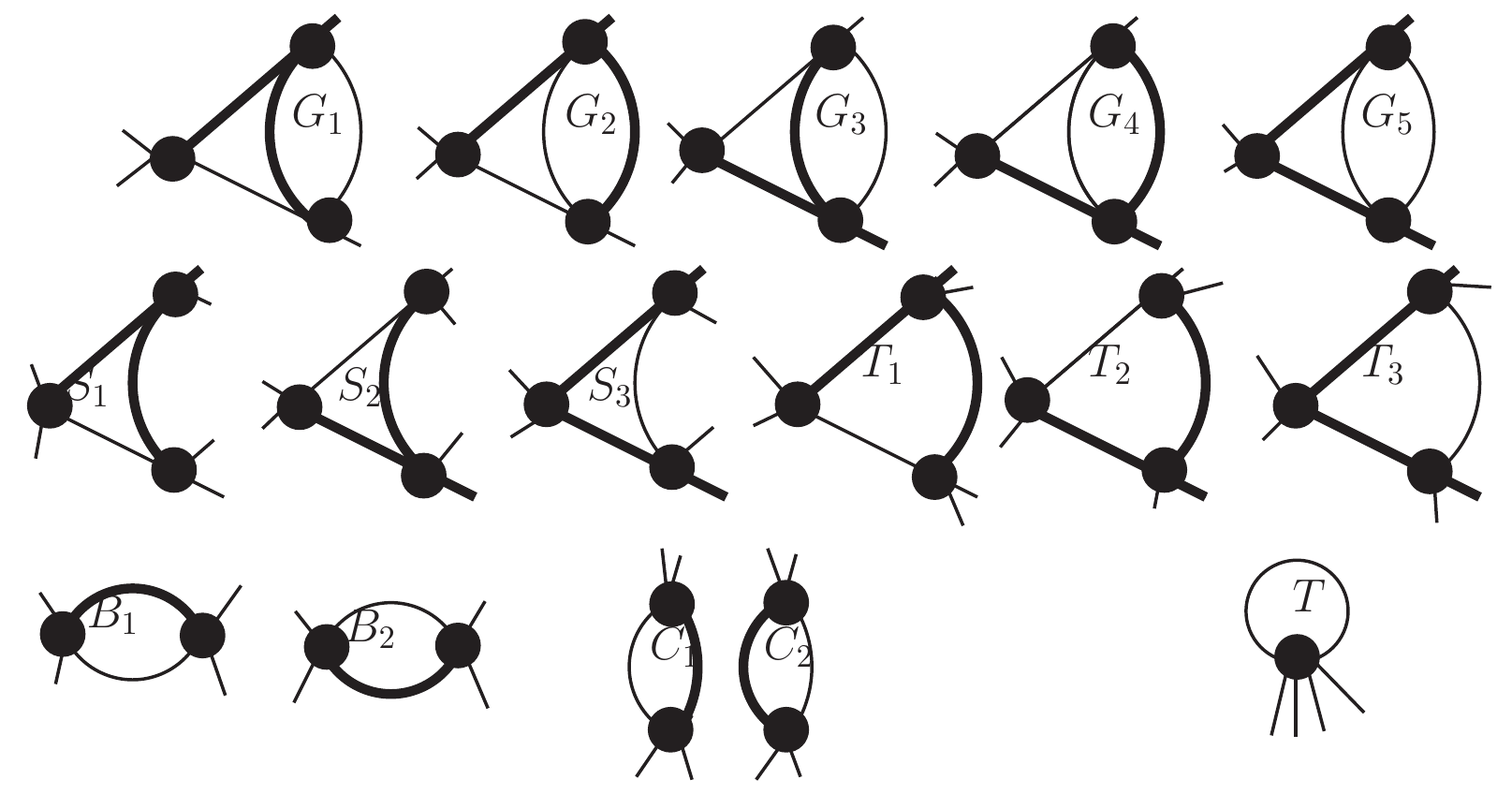}
\caption{The coproduct on $(G,F)$ for all five spanning trees of the Dunce's cap $G$.
We call the five pairs $G_i:=(G,F_i)$, where the spanning trees $F_i$ are indicated by bold lines.
We have for reduced coproducts $\tilde{\Delta}_{GF}$: 
$\tilde{\Delta}_{GF}(G_1)=S_1\otimes T+C_2\otimes B_1$, $\tilde{\Delta}_{GF}(G_2)=T_1\otimes T+C_1\otimes B_1$,
$\tilde{\Delta}_{GF}(G_3)=S_2\otimes T+C_2\otimes B_2$, $\tilde{\Delta}_{GF}(G_4)=T_2\otimes T+C_1\otimes B_2$ and
$\tilde{\Delta}_{GF}(G_5)=S_3\otimes T+ T_3\otimes T$. Note that in this graph, $T-(T\cap g)\not\in F_{G_5/g}$, where $g$ is the subgraph formed by the two edges not in $T$.
Also, note that this calculation agrees with Proposition~\ref{csptr}: $3\times 1+3\times 1+2\times 2=10=5\times 2$.
}
\label{deltaGF}
\end{figure}

We also define the commutative product to be
\[
m_{GF}((\Gamma_1,F_1),(\Gamma_2,F_2))=(\Gamma_1\dot{\cup} \Gamma_2,F_1\dot{\cup} F_2),
\]
whilst $\One_{GF}=(\One,\One)$ serves as the obvious unit which induces a counit 
through $\hat{\One}_{GF}(\One_{GF})=1$. 

\begin{thm}(Theorem 3.1 of \cite{MarkoDirk})
This is a graded commutative bi-algebra graded by $|\Gamma|$
and therefore a Hopf algebra $H_{GF}(\One_{GF},\hat{\One}_{GF},m_{GF},\Delta_{GF},S_{GF})$.
\end{thm}

This is essentially inherited from the core Hopf algebra, see Theorem 3.1 of \cite{MarkoDirk} for the proof.

\subsection{Counting spanning trees}\label{sec count sp tr}
It is often useful to count the number of spanning trees of a graph to control, for example, the number of Hodge matrices describing the analytic structure of an evaluated Feynman graph.  This was used in \cite{coaction} to determine all the Hodge matrices generated from variations of external momenta, in \cite{MarkoDirk}
to determine the number of terms generated from integrating out energy integrals, and 
is used below in Section~\ref{SecDec} to determine the number of sectors in a sector decomposition of physics amplitudes.

So we let  $\mathit{spt}(\Gamma)=|\mathcal{T}(\Gamma)|$ be the number of spanning trees of $\Gamma$,
$\mathit{spt}:H_{core}\to \mathbf{N}$, 
and define $\mathbf{spt}:H_{core}\to \mathbf{N}$, $\mathbf{spt}(\Gamma):=\mathit{spt}(\Gamma)|\Gamma|!$.

\begin{prop}\label{csptr}
  \mbox{}
  \begin{enumerate}
    \item 
\[
\mathbf{spt}(\Gamma)=\sum_{|\Gamma^\prime|=1} \mathbf{spt}(\Gamma^\prime)\mathbf{spt}(\Gamma/\Gamma'),
\]
and
\item
if $|\Gamma|=1$ and $\Gamma$ is bridgeless we have $\mathbf{spt}(\Gamma)=spt(\Gamma)=e_\Gamma$ while for $|\Gamma|>1$
\[
\mathbf{spt}(\Gamma)=m^{|\Gamma|-1}\mathbf{spt}^{|\Gamma|}\tilde{\Delta}_{core}^{|\Gamma|-1}(\Gamma)=m^{|\Gamma|-1}\mathit{spt}^{|\Gamma|}\tilde{\Delta}_{core}^{|\Gamma|-1}(\Gamma).
\]
\end{enumerate}
\end{prop}
\begin{proof}
  \begin{enumerate}
  \item 
    Recall the notion of fundamental cycle from Section~\ref{sec graph summary}.
    
  Let us count $(T,e)$ pairs with $T\in \mathcal{T}(\Gamma)$ and $e\in E_\Gamma\setminus E_T$ in two different ways.  Counting directly, there are $\mathit{spt}(\Gamma) |\Gamma|$ such pairs.  Now we will count $(T,e)$ pairs based on the fundamental cycles.  Each cycle $C$ can appear as a fundamental cycle for any edge $e$ in $C$ and any spanning tree formed from a spanning tree of $\Gamma/C$ along with the edges of $C\setminus e$.  So $C$ is the fundamental cycle for $|C|\textit{spt}(\Gamma/C)$ $(T,e)$ pairs. So there are $\sum_{|\Gamma^\prime|=1} \mathit{spt}(\Gamma^\prime)\mathit{spt}(\Gamma/\Gamma')$ $(T,e)$ pairs in all.  Thus we have
  \[
  \mathit{spt}(\Gamma) |\Gamma| = \sum_{|\Gamma^\prime|=1} \mathit{spt}(\Gamma^\prime)\mathit{spt}(\Gamma/\Gamma').
  \]
  Multiplying both sides by $(|\Gamma|-1)!$ gives the result.
  
  \item The $|\Gamma|=1$ case is immediate as a bridgeless graph with $|\Gamma|=1$ is simply a cycle.  The first equality follows from iterating part $i)$.  To see the same argument directly, note  For any $T\in\mathcal{T}(\Gamma)$ the basis of fundamental cycles $\{l_1, \ldots, l_{|\Gamma|}\}$ can be ordered in $|\Gamma|!$ ways corresponding exactly to the $|\Gamma|!$ flags generated by  
$\tilde{\Delta}^{|\Gamma|-1}(\Gamma)$
\[
l_{i_1}\otimes l_{i_2}/E_{l_{i_1}}\otimes\cdots\otimes l_{i_{|\Gamma|}}/(\cup_{j=1}^{|\Gamma|-1}E_{l_j}).
\]

Since the $\mathbf{spt}$ on the right of the first equality only acts on one loop graphs it can be replaced by $\mathit{spt}$.
\end{enumerate}
\end{proof}
See Figure~\ref{deltaGF} which provides an example of the notions introduced above.

\begin{rem}
  We can combine the various Hopf algebras with spanning forests.
  In particular we can extend the Hopf algebra $H_{pC}$ to $H_{pC,GF}$ for pairs $(G,F)$ with $\Gamma$ pre-Cutkosky and a compatible $F$.  This does not contain further information than $(\hat{\Gamma},F)$, but gives a different viewpoint since a smaller set of forests is compatible with a pre-Cutkosky graph $\Gamma$ than the set of all spanning forests of $\hat{\Gamma}$.  This is because there can be different pre-Cutkosky graphs $\Gamma_1$ and $\Gamma_2$ with $\hat{\Gamma}_1 = \hat{\Gamma}_2$.  Running over all pre-Cutkosky graphs does give all spanning forests of each uncut graph and so $H_{pC,GF}$ is isomorphic to $H_{GF}$, but by writing $G_{pC,GF}$ we are emphasizing collecting forests by which cut they give.
  
All Hopf algebras $H_N$ also extend to Hopf algebras $H_{N,GF}$ for pairs of graphs and spanning forests using Equation~\ref{deltaGF} and projecting to graphs in $H_N$ on both sides of the coproduct.
\end{rem}

\section{Coactions}\label{coactions}
There are various coactions of physical relevance. In particular, the core Hopf algebra $H_{core}$ (as well as each of its quotient Hopf algebras of interest to us) coacts on the Hopf algebra of cut graphs $H_{pC}$,
an algebraic manifestation of basic assumptions as locality of counterterms and the existence of operator product expansions on which QFT is based.

Furthermore, the Hopf algebra structure of graphs also implies that $H_{pC}$ coacts on 
Cutkosky graphs in $H_C$ which gives an iterative structure to dispersion relations. 
$H_{pC}$ also coacts on graphs in $H_{nC}$ whose variations appear 
on non-principal sheets. The latter two coactions will be investigated in greater detail in future work.

In Appendix~\ref{appC} we collect properties of Feynman rules which clarify how the coactions discussed here connect to physics.
\subsection{$H_{pC}\to H_{core}\otimes H_{pC}$}\label{hcorehpc}
This coaction ensures that we can renormalize uncut graphs as usual and that cluster separation is respected in the scattering asymptotics  (see Remark~\ref{remclustersep}). The coaction is based on $\Delta_{core}$ lifted to graphs in $H_{pC}$ by taking only core (that is uncut) subgraphs as allowable subgraphs in the coproduct.  We will now describe this coaction in more detail.

A small technical lemma will be helpful, as all the variants of this coaction that we might use will have the same form.
\begin{lem}
  Let $H$ be a Hopf algebra of graphs and $B$ a vector subspace of $H$ and additionally suppose that the coproduct of $H$ is of the form 
  \[
  \Delta(h) = \sum_{\substack{j\subseteq h \\ j, h/j \in B}}j\otimes h/j
  \]
   on basis elements.  
  Let $J$ be a vector space of graphs with $B$ as a subspace.

  Then $\bar{\Delta}(J) \rightarrow H\otimes J$ given by
  \[
  \bar{\Delta}(j) = \sum_{\substack{k\subseteq j \\ k\in B, j/k\in J}} k\otimes j/k 
  \]
  is a coaction provided that whenever $k\subseteq \ell \subseteq j$ with $j\in J, j\in B, \ell/k \in B$ and $j/\ell \in J$ then $\ell\in B \Leftrightarrow j/k\in J$.
\end{lem}

\begin{proof}
  The counital property is straightforward.  To show coassociativity, consider  $k\subseteq \ell \subseteq j$.  In order to obtain $k\otimes \ell/k \otimes j/\ell$ by $(\text{id}\otimes \bar{\Delta})\bar{\Delta}$ we must have $k\in B$, $j/k\in J$, $\ell/k \in B$, and $j/\ell \in J$.  In order to obtain $k\otimes \ell/k \otimes j/\ell$ by $(\Delta \otimes \text{id})\bar{\Delta}$ we must have $\ell\in B$, $j/\ell\in J$, $k \in B$, and $\ell/k\in B$.  So to obtain the same terms in both directions we must have that if $j\in J, j\in B, \ell/k \in B$ and $j/\ell \in J$ then $\ell\in B \Leftrightarrow j/k\in J$.
\end{proof}

Returning to the coaction of $H_{core}$ on $H_{pC}$,
$H_{core}$, regarded as a sub-vectorspace of $H_{pC}$, is embedded in $H_{pC}$ by viewing $\Gamma\in H_{core}$ as $(\Gamma, \Gamma) \in H_{pC}$, which we will do below without further comment. In fact 
\[
\Gamma \in H_{pC}\cap H_{core}\Leftrightarrow |\Gamma|=||\Gamma||.
\]
Of particular interest for this coaction are the graphs where $|\Gamma|\gneq ||\Gamma||$. Consider the vectorspace $H_{pC}^>$ generated by such graphs.
Note that $H_{pC}^>\subsetneq \mathbf{Aug}_{pC}$, where $\mathbf{Aug}_{pC}$ is the augmentation ideal of $H_{pC}$.

The coproduct $\Delta_{core}$ of $H_{core}$ has then a natural extension $\bar{\Delta}_{core}$ which coacts on $H_{pC}^>$:
\[
\bar{\Delta}_{core}: H_{pC}^>\to H_{core}\otimes H_{pC}^>,
\]
\[
\bar{\Delta}_{core}(\Gamma)=\One\otimes \Gamma+\sum_{\gamma\in H_{core}} \gamma\otimes \Gamma/\gamma,
\]
where the sum is over all $\gamma\in H_{core}$ such that $H_{core}\ni \gamma\subsetneq \Gamma\in H_{pC}^>$.  Note that $(\hat{\One}\otimes \mathrm{id})\bar{\Delta}_{core}=\mathrm{id}$ as it must.
 
Note that $\Gamma/\gamma\in H_{pC}^>$ by construction. One has
\[
(\Delta_{core}\otimes\mathrm{id}) \bar{\Delta}_{core}=(\mathrm{id}\otimes \bar{\Delta}_{core})\bar{\Delta}_{core}: H_{pC}\to H_{core}\otimes H_{core}\otimes H_{pC},
\]
since the condition of the lemma is satisfied because if $\ell/k$ and $k$ are core then $\ell$ is core, and contraction of a core subgraph in a pre-Cutkosky graph always gives a pre-Cutkosky graph.
Consequently, this is a coaction.

\begin{rem}\label{hccoact}
A similar coaction exists for Cutkosky graphs $\Gamma\in H_C^>\subsetneq H_{pC}^>$:
\[
\bar{\Delta}_{core}: H_C^>\to H_{core}\otimes H_C^>.
\]
This is a coaction since only contracting un-cut subgraphs cannot create cut vertices and so the condition of the lemma holds.  As above we can also view this coaction on all of $H_C$, $\bar{\Delta}_{core}:H_C \to H_{core}\otimes H_C$.

Additionally, letting $H_C^{\|j\|}$ be the subspace of $H_C$ spanned by graphs with $\|G\|=j$, we have $H_C^>\subsetneq H_C$ and $H_C$ can be written as a direct sum $H_C=\bigoplus_{j=0}^\infty H_C^{||j||}$.  A graph $\Gamma\in H_C^{||0||}$ has Cutkosky cuts such that no loop is left intact.  This is the vectorspace of almost leading singularities. 
Note that $H_C^{||0||}={H_C^>}^{||0||}$ whilst for $j\gneq 0$ we have
${H_C^>}^{||j||}\subsetneq {H_C}^{||j||}$, the difference 
${H_C}^{||j||}\setminus {H_C^>}^{||j||}=H_{core}^{(j)}$, the space of $j$-loop core graphs. 

Here we call a Cutkosky graph leading if all its edges are cut so that they are evaluated on-shell, and almost leading if the cut edges ensure that no loop is left intact.  This parallels the language of leading and almost or weakly leading singularities.
\end{rem}
\begin{rem}\label{coactthm}
There is an obvious result that renormalization is compatible with these coactions.  Specifically,
\[
\bar{\Phi}(\Gamma):=m_{\mathbb{C}}(S_R^\Phi\otimes\Phi)\bar{\Delta}_{core}(\Gamma),
\]
exists for any Cutkosky graph $\Gamma\in H_C^>$ and for any map
$S_R^\Phi:H_{core}\to\mathbb{C}$
which is a counterterm  $S_R^\Phi:=-R(S_R^\Phi\star \Phi\circ P)$. See Sections~\ref{parametricFR} and \ref{momspaceFR} for information on Feynman rules and counterterms.

From a physicist's viewpoint this is straightforward. Indeed,
$\Gamma\in H_C^>$ is overall convergent so that a renormalization of its subgraphs with loops left intact suffices \cite{CK}.  

Overall convergence follows from the fact that graphs $\Gamma\in H_C^{||0||}$ provide an integrand which is to be integrated over a compact domain only \cite{DirkEll,MarkoDirk}. For the definition of the renormalization scheme $R$ we refer the reader to \cite{MarkoDirk,BrownKreimer} and Appendix~\ref{appC}. Note that renormalization as needed here uses the core Hopf algebra to provide counterterms for graphs with vertices of any valence which emerge from shrinking edges.\footnote{Shrinking edges increases the overall degree of divergence and hence alters the action of $R$ accordingly. See Remark~\ref{higherdegreediv}.}
\end{rem}

Furthermore,
\begin{rem}
  Let $\tilde{\bar{\Delta}}_{core} = \bar{\Delta}_{core} - \One\otimes \mathrm{id}$ applied to $H_{pC}^>$.  Then
\[
\tilde{\bar{\Delta}}_{core}^{||\Gamma||}(\Gamma)
\] 
provides a flag of $||\Gamma||+1$
graphs, where the rightmost tensor factor is an element of $H_C^{||0||}$.
\end{rem}
Figure~\ref{bardelta} shows how $\bar{\Delta}_{core}$ acts in an example.
\begin{rem}\label{remclustersep}
We can let $P_0:H_{pC}\to H_{pC}^{||0||}$ be the projector to pre-Cutkosky graphs with no loop left intact, and $\bar{\Delta}_0:=(\mathrm{id}\otimes P_0)\bar{\Delta}_{core}$. Then for $\Gamma\in H_{pC}^>$, 
$\bar{\Delta}_0(\Gamma)=x\otimes y$, $x\in H_{core}$, $y\in H_{pC}^{||0||}$.
This defines an iterated integral iterating $\Phi_R(x)$ into $\Phi(y)$
to obtain $\Phi_R(\Gamma)$, see Appendices \ref{parametricFR} and \ref{momspaceFR} as well as \cite{MarkoDirk}. 

Furthermore the map
\begin{equation}\label{eq disp eq}
\tilde{\Delta}_{pC}^{(h_0(\tilde{\Gamma})-2)}(y),
\end{equation}
decomposes $y$ into clusters related by dispersion relations, a subject of future work. $S_R^\Phi(x)$ provides all counterterms to render $\Phi_R(\Gamma)$ finite.  In brief, subdivergences separated by cuts can be renormalized independently, and this is captured algebraically by the relation between applying $\tilde{\Delta}_{pC}$ and then $\bar{\Delta}_{core}$ to each component and applying $\bar{\Delta}_{core}$ followed by $\tilde{\Delta}_{pC}$.

If the pre-Cutkosky graph $\Gamma$ gives a 2-partition of its set of external edges $L_\Gamma$,
$h_0(\tilde{\Gamma})=2$, then the map of \eqref{eq disp eq} is the identity and the dispersion is with respect to the momentum flow attributed to a normal cut.
\end{rem}
\begin{rem}\label{corevspc}
Accordingly we can generate all graphs in $H_{pC}$ by dressing all uncut edges and all uncut vertices of $H_{pC}^{||0||}$ by graphs from $H_{core}$.

The dressing of graphs in $H_{pC}^{||0||}$ by graphs from 
$H_{ren}\subsetneq H_{core}$ also corresponds to the identification and separation of hard diagrams from the infrared singularities of diagrams in $H_{pC}^{||0||}$,
see for example Chapter 12 in \cite{Sterman}.
This continues to hold for full Green functions as Equation~\ref{primitivesgpc} below exhibits.
\end{rem}

\begin{figure}[h]
\includegraphics[width=14cm]{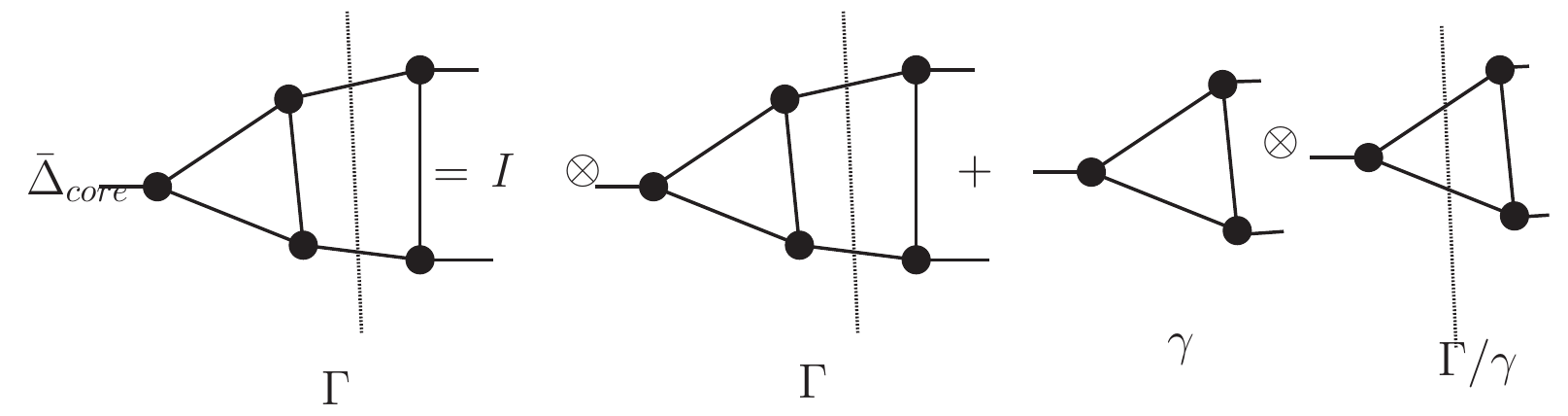}
\caption{The coaction $\bar{\Delta}_{core}$, $\bar{\Delta}(\Gamma)=\One\otimes \Gamma+\gamma\otimes \Gamma/\gamma$. Evaluating the rhs by $m\circ (S_R^\Phi\otimes\Phi)$ one gets a finite result $\Phi(\Gamma)-\Phi_0(\gamma)\Phi(\Gamma/\gamma)$ with $-\Phi_0(\gamma)$ the counterterm for the triangle subgraph $\gamma\in H_{core}$.
This uses that $\Gamma\in H_{pC}$ is overall convergent thanks to the presence of a cut.}
\label{bardelta}
\end{figure}
On the same graph as in Figure~\ref{bardelta}, $\Delta_{pC}$ acts as in Figure~\ref{tildedeltapC}.
\begin{figure}[h]
\includegraphics[width=14cm]{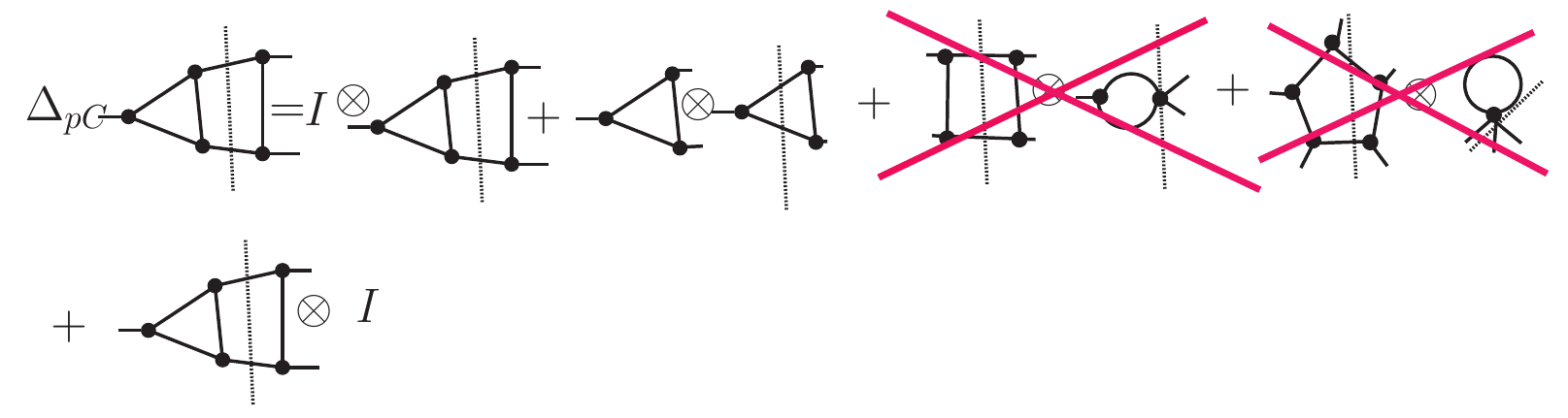}
\caption{The coproduct ${\Delta}_{pC}(\Gamma)=\One\otimes \Gamma+\Gamma\otimes \One +\gamma\otimes \Gamma/\gamma$.  If we chose to impose normality at each vertex cut, then we obtain only only one term
apart from $\Gamma\otimes\One$ or $\One\otimes \Gamma$.  This is because, 
out of the three terms obtained by shrinking cycles, two generate cut corollas which do not cut internal momentum flow so the cut at the corolla in the co-graph  is not normal.  These terms are crossed out in the figure and do not appear if we impose normality of the cuts at each vertex.}
\label{tildedeltapC}
\end{figure}

We note that the sub-vectorspace of graphs $H_{pC}^{||0||}\subsetneq H_{pC}$ such that no loop is left intact, $||\Gamma||=0$, forms a sub Hopf algebra $H_{pC}^{||0||}$ of $H_{pC}$.
\[
\Delta_{pC}H^{||0||}_{pC}\subset H^{||0||}_{pC}\otimes H^{||0||}_{pC}.
\]
Also we have 
\[
\bar{\Delta}_{core}(\Gamma)=\One\otimes \Gamma,\,\forall \Gamma\in H_{pC}^{||0||}.
\]
Corresponding to $H_{pC}^{||0||}$
there is a similar sub vectorspace $H_C^{||0||}\subsetneq H_C$ of graphs such that no loop is left intact.
\subsection{Cointeracting bialgebras}\label{cointbi}
There is an interesting interplay between the Hopf algebra $H_{GF}$ and the cubical chain complex associated to pairs $(\Gamma,T)$ \cite{rational}. In particular the chain complex gives rise to an incidence bialgebra $B_I$ based on the set of edges $E_T$ 
in the spanning tree $T$ of a pair $(\Gamma,T)$. It was studied in \cite{coaction}
and established that there is a coaction which assigns Galois conjugates to any pair $(\Gamma,T)$.
$H_{GF}$ and $B_I$ cointeract as bialgebras. This notion of cointeracting bialgebras has been recently investigated by Foissy \cite{cointeractiontalk,cointpapers} and others \cite{Manchon}.

In fact for any pair $(\Gamma,T)$ the set $E_L:=E_\Gamma\setminus E_T$ of edges $e\not\in E_T$ decomposes
under $\Delta_{GF}$ in the sense that each term in $\Delta_{GF}((\Gamma,T))$ partitions $E_T$ between the two sides and in fact the partition determines the graph-tree pairs on the two sides, as described in more detail in Section~\ref{subsec global set up}.  To say this in another way let us use the notation $E_{(\Gamma,T)}:= E_L$.
Then if we write
\[
\Delta_{GF}((\Gamma,T))=\sum_i (\Gamma,T)_i^\prime\otimes (\Gamma,T)_i^{\prime\prime},
\]
for any $i$ we have $E_{(\Gamma,T)_i^\prime}\cap E_{(\Gamma,T)_i^{\prime\prime}}=\emptyset$
and  $E_{(\Gamma,T)_i^\prime}\cup E_{(\Gamma,T)_i^{\prime\prime}}=E_L$ and the $E_{(\Gamma,T)_i^\prime}$ and $E_{(\Gamma,T)_i^{\prime\prime}}$ determine $(\Gamma,T)_i^\prime$ and $(\Gamma,T)_i^{\prime\prime}$. 

This allows us to work with a commutative and cocommutative bialgebra $(\mathcal{A}_p,m,\Delta_c)$ defined in terms of the unordered set $E_L$
with coproduct denoted by $\Delta_c$ and another map, $\rho$ which is a slightly modified incidence coproduct on intervals in $E_T$.  On graphs with tadpoles\footnote{i.e. self-loops} allowed (see Appendix~\ref{appB}), $\rho$ is a coproduct and with the same product $m$ we get a second bialgebra.
These two bialgebras are in cointeraction.

In Appendix~\ref{appB} we set this up in two different ways, first using a direct approach using the sets $E_L$ and $E_T$ followed by an approach via generators $x_{e,[a,b]}$ provided by single edges $e\in E_L$ and intervals $[a,b] \in E_T$. 

In particular in  Appendix~\ref{appB} we show, as part of Theorem~\ref{thmgenerators},
\be\label{cointeq}
m_{1,3,24}\circ(\rho\otimes\rho)\circ \Delta_c=(\Delta_c\otimes\mathrm{id})\circ\rho,
\ee
where $m_{1,3,24}$ is the map which multiplies together the arguments in the second and fourth slots.  This identity holds quite generally: we are working with cut edges via $(\Gamma,F)$ pairs and additionally the graph edges can be marked as to whether or not they are allowed as tadpoles.
The identity \eqref{cointeq} explains how renormalization and monodromies interfere.
As is discussed in more detail in Sections~\ref{shuffle} and \ref{subsec global set up}, in the $(\Gamma,F)$ context, $\Delta_c$ is closely related to $\Delta_{GF}$, differing in that $\Delta_c$ requires that the left hand side of the tensor product is always uncut.  The map $\rho$, as mentioned above, comes from the incidence coproduct on intervals in the power set of $E_T$.  An interval is interpreted as specifying edges to cut and edges to contract: $[a,b]$ represents cutting the set of edges $a$ and contracting the set of edges $E_T\setminus b$.  The map $\rho$ also differs from the incidence coproduct in that we can mark edges for which we forbid intervals which would yield these edges as tadpoles when $E_T\setminus b$ is contracted.  This means that $\rho$ is a coproduct on graphs with edges that are allowable as tadpoles and a coaction more generally.
The forbidding of certain tadpoles lines up with the fact that tadpoles vanish in kinematic renormalization schemes.  
\begin{rem}
In fact whenever renormalization is achieved by using the forest sum to subtract at a chosen $q_0\in Q^L$ tadpole integrals vanish (kinematic renormalization schemes)  and $\rho$ becomes a coaction on core
graphs and remains a coproduct on proper Cutkosky graphs in $H_C$. The situation is essentially the same in minimal subtraction schemes for massless particles. With massive particles the situation is slightly more subtle and the full set-up of Appendix~\ref{sec global co} is needed. See also the discussion  in \cite{coaction}.
\end{rem}

To explore these notions the next step is to define (combinatorial) Galois conjugates. 
\subsubsection{Galois conjugates as fundamental cycles and cuts}\label{subsec galois}
The notion of Galois conjugates for Feynman graphs owes its existence to ideas by Francis Brown \cite{BrownCNTP,Brownrecent}. Below it inspires a combinatorial study 
of pairs $(\Gamma,T)$ and $(\Gamma,f)$ through removing or shrinking edges.

Let a pair $(\Gamma,T)$ of a graph and spanning tree be given. Consider two mutually disjoint not necessarily non-empty subsets $p,q\in E_T$,
$p\cap q=\emptyset$.  Set $\gamma:=\Gamma/p$, $f:=T/p\setminus q$.
We call such a pair $(\gamma,f)$ a Galois conjugate of $(\Gamma,T)$.

So $(\Gamma,T)$ is a Galois conjugate of itself ($p=q=\emptyset$) as is any pair $(\Gamma,F)$
for $p=\emptyset,q=E_T\setminus E_F$. 

The set of all Galois conjugates of a pair $(\Gamma,T)$ is denoted by 
\be\label{galconj}
\mathsf{Gal}_{\Gamma,T}:=\{(\Gamma/p,T/p\setminus q)\,|\, p,q\in E_T, p\cap q =\emptyset \}.
\ee
We consider Galois conjugates relative to their fixed $(\Gamma,T)$ and so we define the set $ E_L=\{e\in E_{\Gamma/p}| e\not\in E_{T/p}\}$ to be the same for all Galois conjugates.

For any $e\in E_L$, $l(T/p,e)$ is a fundamental cycle for $\gamma$ and we define the path $t_e:=l(T/p,e)\cap T/p$ and also define 
\[
f_e:=t_e\cap f, C_e:=t_e\setminus f_e= t_e\cap q.
\] 

For $C_e\not=\emptyset$ we call 
$l(f_e,e):=l(T/p,e)\setminus C_e$ a cut fundamental cycle.  

A fundamental cycle $l(T/p,e)$ for any Galois conjugate
defines a graph $\gamma_e=(H_{\gamma_e},\mathcal{V}_{\gamma_e},\mathcal{E}_{\gamma_e})$
with 
\begin{itemize}
\item $H_{\gamma_e}=\dot{\cup}_{v\in l(T/p,e)}\mathbf{c}_v$ which also determines $\mathcal{V}_{\gamma_e}$ whilst 
\item $\mathcal{E}_{\gamma_e}$ is determined by letting the edges in $l(T/p,e)$  define the parts of cardinality two in $\mathcal{E}_{\gamma_e}$. 
\end{itemize}
This is the fundamental cycle with external edges since the full corollas at each vertex are included.

Also, $l(T/p,e)$ defines a pair $(\gamma_e,t_e)\in H_{GF}$ with 
$t_e=\gamma_e\cap T/p$. 
Similarly a cut fundamental cycle $l(f_e,e)$ defines a Cutkosky graph $(\gamma_e,h_e)\in H_C$ and a pair $(\gamma_e,f_e)\in H_{GF}$
where now
\[
(\gamma_e,h_e)=((H_{\gamma_e},\mathcal{V}_{\gamma_e},\mathcal{E}_{\gamma_e}),(H_{\gamma_e},\mathcal{V}_{\gamma_e},\mathcal{E}_{h_e})).
\] 
Here, $H_{\gamma_e},\mathcal{V}_{\gamma_e},\mathcal{E}_{\gamma_e}$ are as before and 
$\mathcal{E}_{h_e}$ is determined by letting the edges in $l(T/p,e)\setminus C_e$
determine the parts of cardinality two.  This is the cut fundamental cycle with full corollas and hence with external edges.

\begin{figure}[h]
\includegraphics[width=10cm]{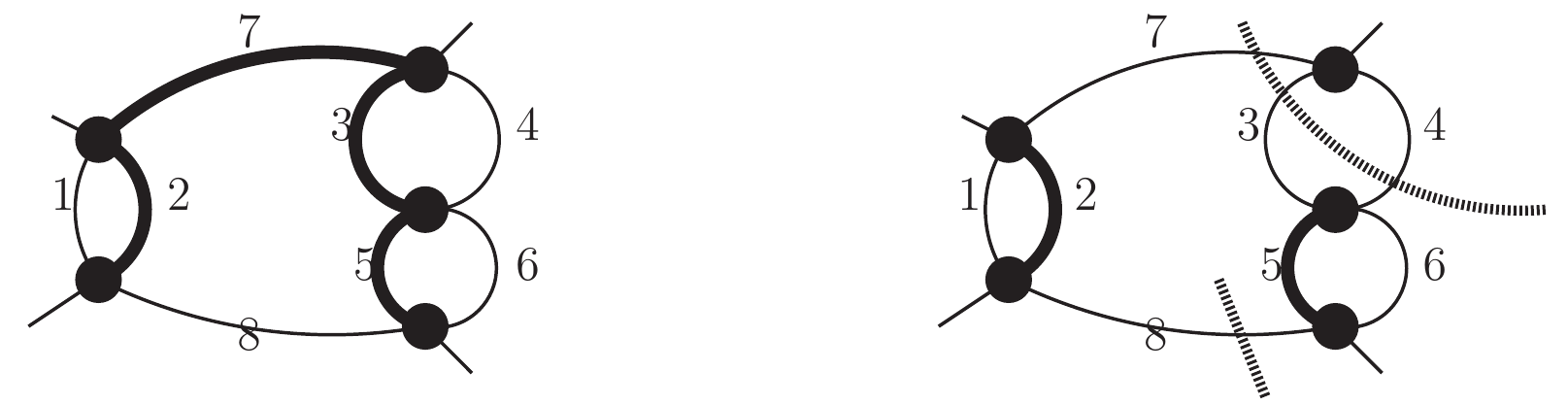}.
\caption{On the left we see a pair $(\Gamma,T)$ with $|\Gamma|=4$, edges in $T$ are in bold lines. The four fundamental cycles are $l_1=l(T,e_1)$, $t_{e_1}=e_2$, $l_2=l(T,e_4)$, $t_{e_4}=e_3$, $l_3=l(T,e_5)$, $t_{e_5}=e_6$, $l_4=(T,e_8)$, $t_{e_8}=e_2,e_7,e_3,e_5$.  On the right, $l_1$ and $l_3$ remain fundamental cycles while we now have cut fundamental cycles $l(T,e_4)\setminus e_3$ ($C_{e_4}=e_3$) and $l(T,e_8)\setminus\{e_7,e_3\}$ ($C_{e_8}=\{e_7,e_3\}$).}
\label{coactef}
\end{figure}

\subsubsection{Comparison of $\Delta_c$ and $\Delta_{GF}$}\label{shuffle}
Equation (\ref{cointeq}) answers a question which has not been satisfactorially answered yet in the physics literature:  how are the algebraic structures of renormalization and the analytic structure of physics amplitudes compatible? 

Here, we consider the removal
of edges from a spanning tree as synonymous with an investigation of the analytic structure of amplitudes in view of the results in \cite{MarkoDirk}.
For the purposes of this section, we will view a spanning forest $F$ as being obtained by removing edges from a spanning tree $T$, and so implicitly the information of $T$ is carried along with $(\Gamma,F)$.
For a pair $(\Gamma,T)$ define the set of fundamental cycles to be
\[
\mathsf{L}_T:=\{l(T,e): e\in E_L\}
\]
Note that $C_e=\emptyset$ for all $l(T,e)$ in a pair $(\Gamma,T)$.
For a pair $(\Gamma,F)$ define
\[
\mathsf{L}_F:=\{l(T,e): e\in E_L, C_e=\emptyset\}, 
\]
to be the set of fundamental cycles corresponding to loops left intact, hence for which $t_e$ is contained in $F$.

Consider the map
\be\label{coactgf}
\Delta_{GT}: H_{GF}^>\to H_{GT}\otimes H_{GF}^>,\, (\Gamma,F)\mapsto \One\otimes (\Gamma,F)+\sum_{p\subseteq \mathsf{L}_F} (p,t_p) \otimes (\Gamma/p,F/t_p),
\ee
where
\[
t_p:=\cup_{l(T,e)\in p\subseteq \mathsf{L}_F}t_e,
\] 
is the union of the spanning trees $t_e$ in fundamental cycles of $p$. 
This map is the variant of $\Delta_{GF}$ which agrees with the coaction $\Delta_c$ in the context of pairs $(\Gamma,F)$, as will be proved in Lemma~\ref{lem c GT}.

We emphasize that  $\mathsf{L}_F=\emptyset$ is possible whilst $|\mathsf{L}_T|$ can be large since removing edges from $T$ might leave few or no loops intact.

Here $H_{GT}\subsetneq H_{GF}$ is the sub-Hopf algebra coming from sets of pairs $(\Gamma,T)$ with $T$ a spanning tree of $\Gamma$. It forms a sub-Hopf algebra of $H_{GF}$ by definition of $\Delta_{GF}$ which acts on such pairs as
$H_{GT}\to H_{GT} \otimes H_{GT}$. $H_{GF}^>$ is generated from pairs $(\Gamma,F)$ where $F$ is not a spanning tree, $e_F<e_T$.

Note that there is a surjective map $H_{GF}^>\to H_C^>$ by 
\[
(\Gamma,F)\to ((H_\Gamma,\mathcal{V}_\Gamma,\mathcal{E}_\Gamma), (H_\Gamma,\mathcal{V}_\Gamma,\mathcal{E}_H)),
\] 
with $\mathcal{E}_H$ determined by  $\emptyset\not= C_G=E_T\setminus E_F$.

For any $(\Gamma,F)$, we define $(\Gamma_0,F_0):=(\Gamma/\mathsf{L}_F,F/t_{\mathsf{L}_F})$. 
We have 
\[
\Delta_{GF}((\Gamma_0,F_0))=\One\otimes (\Gamma_0,F_0), 
\]
as it corresponds to a union of cut fundamental cycles.

Next, using the notation of Appendix~\ref{appB}, we consider the map $w:H_{GF}\to \mathcal{A}$: 
\[
w((\Gamma,F))=E_L^{[E_T\setminus E_F, E_T]} = \prod_{l(T,e)\in \mathsf{L}_F}x_{e,[E_T\setminus E_F,E_T]}.
\]

\begin{lem}\label{lem c GT}
\[
\Delta_c \circ w((\Gamma,F))=(w\otimes w)\circ \Delta_{GT}((\Gamma,F)).
\]
\end{lem}
Note that while the maps $\Delta_c$ and $\Delta_{GT}$ agree, the corresponding Hopf algebras are not the same because the products are different -- with $\Delta_{GT}$ we use the usual disjoint union product of graphs, while with $\Delta_c$ we merge subgraphs of the fixed parent graph.

\begin{proof}
We can invert the map $w:H_{GF}\to \mathcal{A}$ for any given pair $(\Gamma,F)$
(which hides the cocommutativity of $\Delta_c (w((\Gamma,T)))$). See Figure~\ref{GFvsShuffle}.

For the coproduct terms this gives
\[
w(p)\otimes w((\mathsf{L}_F\setminus p)) = p\otimes (\Gamma/p,F/t_p).
\]
In particular for $p=\One$ we obtain 
\[
1_F \otimes w((\Gamma,F)) = \One \otimes (\Gamma,F),
\]
and for $p=\mathsf{L}_F$, 
\[
w((\mathsf{L}_F,t_{\mathsf{L}_F}))\otimes w((\Gamma_0,F_0)) = (\mathsf{L}_F,t_{\mathsf{L}_F})\otimes
(\Gamma_0,F_0),
\]
 and we use $w(\One)=1_F$.
\end{proof}

Note that under $w$ the expanded flag $\Delta_p^{|\Gamma|-1}w((\Gamma,T))$ maps to
the flag $w^{\otimes|\Gamma|}\Delta^{|\Gamma|-1}_{GF}((\Gamma,T))$ using the parlance of \cite{MarkoDirk}. 
\begin{figure}[H]
\includegraphics[width=14cm]{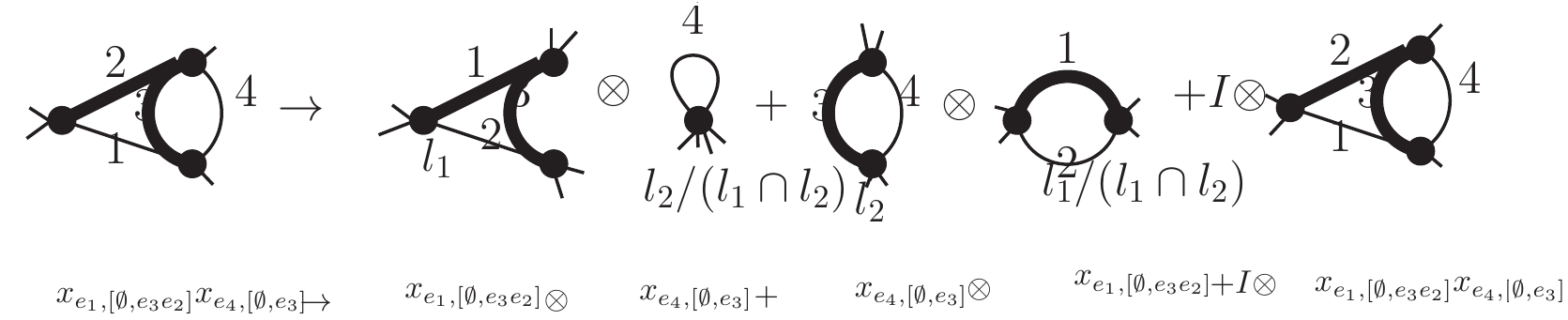}
\caption{We compare $\Delta_{GF}$ (above) and $\Delta_c$ (below).
$l_1$ is the cycle $e_2e_3e_1$ corresponding to $x_{e_1,[\emptyset,e_3e_2]}$ and $l_2$ is the cycle $l_4l_3$ corresponding to $x_{e_4,[\emptyset,e_3]}$. $l_1/(l_1\cap l_2)$ is the cycle $e_1e_2$ corresponding to $x_{e_1,[\emptyset,e_2]}$ and $l_2/l_1$ is the cycle $e_4$, a tadpole corresponding to $x_{e_4,[\emptyset,\emptyset]}$. $e_2,e_3$ make up the spanning tree.}
\label{GFvsShuffle}
\end{figure}

\subsubsection{The Galois coaction}\label{gfcoaction}
Switching from $H_{core}$ to $H_{GT}$ and from $H_C$ as in Remark~\ref{hccoact} to $H_{GF}$ we have the coaction \eqref{coactgf} as a coaction corresponding to the one above Section~\ref{hcorehpc}. 

We can use it for a coaction for Galois conjugates of $\Phi_R((\Gamma,T))$ \cite{coaction} which we define as the set of renormalized evaluations (see \eqref{galconj} and \cite{MarkoDirk} for the evaluations of graphs with Cutkosky cuts)
\[
\mathsf{Gal}_{\Gamma,T}^{\Phi_R}:=\{\Phi_R((\Gamma/p,T/p\setminus q))\,|\, p,q\in E_T, p\cap q =\emptyset \}.
\]
In fact the set $\mathsf{Gal}_{\Gamma,T}^{\Phi_R}$ decomposes into two mutually disjoint sets:
\[
\mathsf{Gal}_{\Gamma,T}^{\Phi_R}=\mathsf{Gal}_{\Gamma,T}^{\mathfrak{m}}\dot{\cup} \mathsf{Gal}_{\Gamma,T}^{\mathfrak{dr}},
\]
\[
\mathsf{Gal}_{\Gamma,T}^{\mathfrak{m}}=\{\Phi_R((\Gamma/p,T/p))\,|\, p\in E_T, (q =\emptyset) \},
\]
\[
\mathsf{Gal}_{\Gamma,T}^{\mathfrak{dr}}=\{\Phi_R((\Gamma/p,T/p\setminus q))\,|\, p,q\in E_T, p\cap q =\emptyset,\,q\not=\emptyset \}.
\]

Following the notation of \cite{Rella}, each graph $\Gamma/p$ in a pair $(\Gamma/p,T/p)\in \mathsf{Gal}_{\Gamma,T}^{\mathfrak{m}}$ defines an integrand corresponding to a class $[\hat{\omega}_{\Gamma/p}]$ ({\it de Rham framing}) and has an associated domain of integration which defines a class $[\hat{\sigma}_{\Gamma/p}]$ ({\it Betti framing}). 

The two framings pair to a motivic Feynman integral: 
\[
I_{\Gamma/p}^{\mathfrak{m}}=[H,[\hat{\omega}_{\Gamma/p}],[\hat{\sigma}_{\Gamma/p}]]^{\mathfrak{m}},
\]
 as a pairing of Betti and de Rham classes yields periods 
 \[
 I_{\Gamma/p}=\int_{\hat{\sigma}_{\Gamma/p}}\hat{\omega}_{\Gamma/p}^R=\int_{\sigma_{\Gamma/p}}\omega_{\Gamma/p}^R
 =\Phi_R(\Gamma/p),
 \] 
 with $H$ the associated Hodge structure (see \cite{Rella} for notation)
 \[
 H=H^{e_{\Gamma/p}-1}(P \setminus Y_{\Gamma/p} , B\setminus (B \cap Y_{\Gamma/p} )),
 \]
and $\omega_{\Gamma/p}^R$ the (class of) the renormalized form for $\Gamma/p$. 

Similarly the cointeraction on graphs above suggests that on the de Rham side the pairing is: 
\[
I_{\Gamma/p\setminus q}^{\mathfrak{dr}}=[H,[\hat{\omega}_{\Gamma/p\setminus q}],[\hat{\sigma}_{\Gamma/p\setminus q}]]^{\mathfrak{dr}},
\]
with
 \[
 I_{\Gamma/p\setminus q}=\int_{\hat{\sigma}_{\Gamma/p\setminus q}}\hat{\omega}_{\Gamma/p\setminus q}^R=\int_{\sigma_{\Gamma/p\setminus q}}\omega_{\Gamma/p\setminus q}^R
 =\Phi_R(\Gamma/p\setminus q).
 \] 
Here on the Betti and de Rham sides we have the classes $\hat{\sigma}_{\Gamma/p\setminus q},\hat{\omega}_{\Gamma/p\setminus q}$ being determined by  localizing $\hat{\sigma}_{\Gamma/p}$ and $\hat{\omega}_{\Gamma/p}$ to the corresponding threshold divisor accordingly.

Then there is a coaction suggested by the incidence coalgebra structure above: 
\[
\mathsf{Gal}_{\Gamma,T}^{\mathfrak{m}}\to \mathsf{Gal}_{\Gamma,T}^{\mathfrak{m}}\otimes
\mathsf{Gal}_{\Gamma,T}^{\mathfrak{dr}}.
\]
We will discuss this coaction in more detail below.
\begin{rem}
 Here on the rhs for $\mathsf{Gal}_{\Gamma,T}^{\mathfrak{dr}}$ we compute modulo $2\pi\i$.
 This is evident as any element in $I^{\mathfrak{dr}}\in\mathsf{Gal}_{\Gamma,T}^{\mathfrak{dr}}$ corresponds to a threshold divisor defined by setting edges $e\in E_{on}=E_\Gamma\setminus E_T$ onshell. 
$I^{\mathfrak{dr}}$ is a physical observable hence real for fixed chosen internal masses and external momenta. This real observable is a function itself of those internal masses and external momenta and hence can have an imaginary part when we vary those parameters.
We are ignoring this imaginary part as it reflects variations from putting a larger 
set of edges $e\in E_\Gamma\setminus E_{\tilde{F}}$ onshell, with $E_{\tilde{F}}\subsetneq E_F$ and this variation is the captured by a different element in 
 $\mathsf{Gal}_{\Gamma,T}^{\mathfrak{dr}}$. That this can be done consistently for any chosen refinement of $L_\Gamma$ reflects the Steinmann relations \cite{Steinmann,coaction}.
\end{rem}   
 Note that this coaction is a coproduct on $\mathsf{Gal}_{\Gamma,T}^{\mathfrak{dr}}$
 reflecting the fact that $\rho$ is a coproduct for intervals which do not correspond to contracting massless edges, see Appendix~\ref{incidence}.

\medskip
We can describe this coaction through the cointeracting bialgebras defined in Appendix~\ref{appB}.
First, let $\Phi_R: \mathcal{A}_p\to I^{\mathfrak{m}}$ assign to a monomial in $q\in \mathcal{A}_p$ 
the corresponding renormalized motivic integral $I^{\mathfrak{m}}_q\in \mathsf{Gal}_{G,T}^{\mathfrak{m}}$
associated  with the corresponding Feynman integral $\Phi_R(q)$ assigned to $q$.

Now consider 
\[
\rho_\Phi: \mathcal{A}\to \mathsf{Gal}_{\Gamma,T}^{\mathfrak{m}}\otimes \mathsf{Gal}_{\Gamma,T}^{\mathfrak{dr}},
\]
which we define via 
\be\label{rhoaction}
\rho_\Phi:=(m_{\mathbb{C}}\otimes_{\mathbb{Q}}\mathrm{id})(S_R^\Phi\otimes\Phi\otimes {\bar{\Phi}})\circ (w^{-1})^{\otimes 3}\circ(\Delta_c\otimes \mathrm{id})\circ \rho= (\Phi_R\otimes_{\mathbb{Q}} {\bar{\Phi}})\circ (w^{-1})^{\otimes 2}\circ\rho.
\ee
Here on the right $\rho$ is a coaction when it acts on $J_1$  and is a coproduct in general. In particular for pairs $(\Gamma,T)$
\[
\rho_\Phi\circ w: \mathsf{Gal}_{\Gamma,T}^{\mathfrak{m}}\to \mathsf{Gal}_{\Gamma,T}^{\mathfrak{m}}\otimes
\mathsf{Gal}_{\Gamma,T}^{\mathfrak{dr}},
\]
coacts, while for pairs $(\Gamma,F)$
\[
\rho_\Phi\circ w: \mathsf{Gal}_{\Gamma,T}^{\mathfrak{dr}}\to \mathsf{Gal}_{\Gamma,T}^{\mathfrak{dr}}\otimes
\mathsf{Gal}_{\Gamma,T}^{\mathfrak{dr}},
\]
is a coproduct.

Finally we want to sum over $T\in\mathcal{T}(\Gamma)$ and consider 
\[
\mathfrak{C}: H_C\to \mathsf{Gal}_{\Gamma}^{\mathfrak{m}}\otimes \mathsf{Gal}_{\Gamma}^{\mathfrak{dr}},\,\mathsf{Gal}_{\Gamma}^{\mathfrak{m}}=\sum_{T\in\mathcal{T}(\Gamma)}
\mathsf{Gal}_{\Gamma,T}^{\mathfrak{m}},
\mathsf{Gal}_{\Gamma}^{\mathfrak{dr}}=\sum_{T\in\mathcal{T}(\Gamma)} \mathsf{Gal}_{\Gamma,T}^{\mathfrak{dr}},
\]
with
\[
\mathfrak{C}(\Gamma):=\sum_{T\in\mathcal{T}(\Gamma)}\rho_\Phi\circ w_{(\Gamma,T)}.
\]
\begin{lem}\label{lemmafrakC}
\[
\mathfrak{C}(\Gamma)=\sum_{p\in \mathcal{P}(L_\Gamma)}\sum_{F\sim p}\Phi_R(\Gamma/E_F)\otimes \bar{\Phi}(\tilde{\Gamma}_F),
\]
where we sum over all partitions $p$ of $L_\Gamma$ and over all forests $F$ compatible with $p$, and where $\tilde{\Gamma}_F$ is the graph $\Gamma$ with the cut corresponding to $F$ done (analogously to the associated graph of Section~\ref{sec precut}). We use $\Phi_R(\gamma)=\sum_{T\in\mathcal{T}(\gamma)}\Phi_R((\gamma,T))$, $\forall \gamma\in H_{core}$, a result of  \cite{MarkoDirk}.
\end{lem}
\begin{proof}
  For $\tilde{\Gamma}_F\in H_C^{(0)}$, the assertion is obvious
as there are no loops left intact and the spanning forest is then unique. 
If $\tilde{\Gamma}_F\in H_C^{(j)},j\gneq 0$, note that $\tilde{\Gamma}_F$ contains $j$ loops so that the spanning forests are not unique. On the left hand side 
$\Phi_R(\Gamma/E_F)$ contains then $j$ 1-vertex reducible petals $p(e_i)$ on $j$ edges $e_i$ which are part of these $j$ loops $l_i$  (one edge for each loop) but not part of the $j$ spanning trees $t_i=F\cap l_i$, $1\leq i\leq j$. Different choices of $F$, $F\sim p$
vary the $t_j$ and hence vary the edges forming the petals. As petals are 1-vertex reducible the petals $p(e)$ provide factors $\Phi_R(p(e))$ for any edge $e$.
The lemma follows if we have  $\Phi_R(p(e))=\Phi_R(p(f))$ for any distinct edges $e,f$ as then on the rhs we can factorize the sum over all $F\sim p$, 
\[
\sum_{F\sim p}\Phi((\Gamma,F))=\Phi(\tilde{\Gamma}_F),
\] 
by \cite{MarkoDirk}. This is true in kinematic renormalization schemes $R$ for which $\Phi_R(p(e))=0,\forall e$ but also when we enforce $\Phi_R(p(e))=1,\forall e$
as suggested by normalizing each graph against its leading singularity \cite{coaction}. If we use $\Phi_R(p(e))=1$ on the lhs the petals correspond to loops left intact on the rhs. Then $\Phi$ on the rhs has to be replaced by $\bar{\Phi}$,
see Eq.(\ref{coactthm}), as indicated.
Massless petals are forbidden in $\mathfrak{C}(\Gamma)$ by the action of $\rho$ and do not contribute on the right even for more general choices of renormalization schemes. 

For $\Phi_R(\Gamma/E_F)$ we use use kinematic renormaliztion conditions which define a renormalization point $S_0$ such that 
the first $\omega_{\Gamma/E_F}+1$ Taylor coefficients of  $\Phi_R(\Gamma/E_F)=\Phi_R(\Gamma/E_F)(S,S_0)$ vanish when expanding in the scale $S$, see Appendix~\ref{appC} and \cite{BrownKreimer}.
\end{proof}
\begin{rem}
Note that the above set-up is sufficiently flexible to allow the treatment of minimal subtraction (MS) renormalization schemes which set massless tadpoles to zero but not massive ones. Kinematic renormalization schemes would set them to zero in all cases, so that $\rho$ is a coaction on core graphs. Note that they form the leftmost column in the matrix in Fig.(\ref{finalexample}) below. Monodromies are then generated solely 
by variying variables $z\in Q^L$. In MS schemes massive tadpoles lead to monodromies 
as functions of $m_e^2$.
\end{rem}
\begin{rem}\label{mdrdualrem}
As we have $E_{on}=E_G\setminus E_F$ and $E_{\mathit{off}}=E_F$ there is an involution
$\iota,\iota^2=\mathrm{id}$, $\iota(\Gamma\setminus E_{on})=\Gamma/E_{\mathit{off}}$, $\iota(\Gamma/E_{\mathit{off}})=
\Gamma\setminus E_{on}$. We interpret it as implementing a combinatorial reflection duality between motivic and de Rham Feynman periods following Brown \cite{BrownCNTP}. 
\end{rem}
\begin{rem}
Below we derive formulae for the coproduct $\Delta_{\bullet}$ on full Green functions. In future work this will allow us to study $(\Phi_R\otimes_{\mathbb{Q}} {\Phi})\circ\rho$ acting on such Green functions as $\rho$ maps a Green function to Green functions of 
all Galois conjugates. 
\end{rem}
\subsubsection{Example}
\begin{ex}
We consider the one-loop triangle graph $t$  on edges $e_1,e_2,e_3$ with corrsponding vertices $v_{12},v_{23},v_{31}$.

\begin{figure}[H]
\includegraphics[width=14cm]{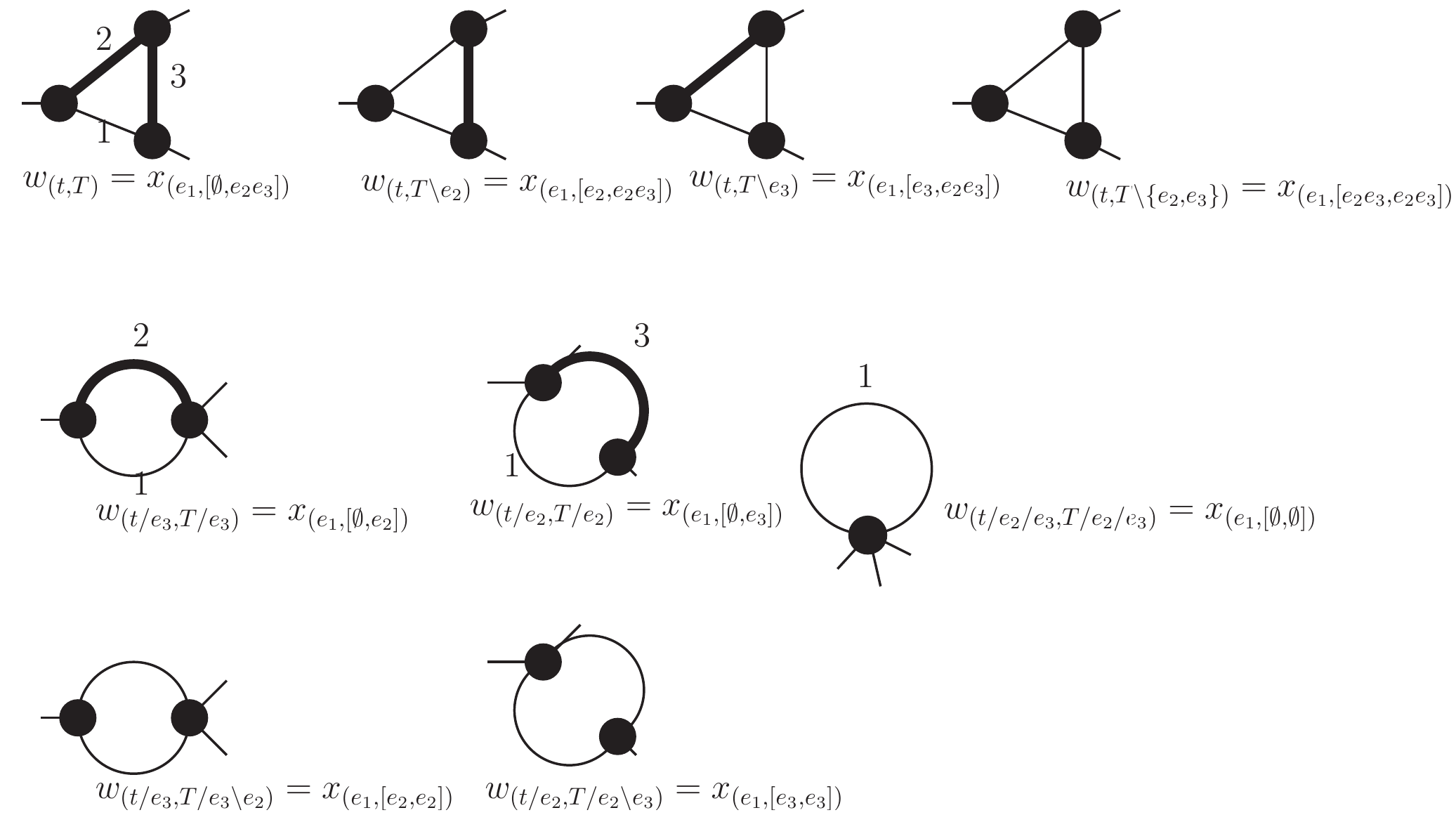}
\caption{The triangle graph $t$ with for example a spanning tree $T$ on edges $e_2,e_3$ and the associated Galois conjugates and generators.}
\label{examplet}
\end{figure}

It gives rise to reduced graphs $t_1:=t/e_1$, $t_2:=t/e_2$, $t_3:=t/e_3$ and tadpoles 
$t_{ij}= t_{ji}:=t/e_i/e_j$, $i\not= j$.
The graph $t$ has three spanning trees given by any pair of its edges.

Let us start with the pair $(t,T)=(t,\{e_2,e_3\})$, $T=\{e_2,e_3\}$, with fundamental cycle $l(\{e_2,e_3\},e_1)$, see Figure~\ref{examplet}. Let $v_{12}$ be the vertex between edges $1$ and $2$ and similarly for $v_{13}$, $v_{23}$.  The Galois conjugates are
\beas
(t, T) = (t,(\{v_{12}, v_{13}, v_{23}\}, \{e_2,e_3\})) & \sim &  x_{e_1,[\emptyset,e_2e_3]},\\
(t/e_2, T/e_2) = (t/e_2,(\{v_{13}, v_{12}\cup v_{13}\}, \{e_3\})) & \sim &  x_{e_1,[\emptyset,e_3]},\\
(t/e_3, T/e_3) = (t/e_3,(\{v_{12}, v_{13}\cup v_{23}\}, \{e_2\})) & \sim &  x_{e_1,[\emptyset,e_2]},\\
(t_{23}, T/e_2/e_3) = (t_{23},(\{v_{12}\cup v_{23}\cup v_{31}\}, \emptyset)) & \sim &  x_{e_1,[\emptyset,\emptyset]}\\
(t, T\setminus e_2) = (t,(\{v_{12}, v_{13}, v_{23}\}, \{e_3\})) & \sim &  x_{e_1,[e_2,e_2e_3]},\\
(t, T\setminus e_3) = (t,(\{v_{12}, v_{13}, v_{23}\}, \{e_2\})) & \sim &  x_{e_1,[e_3,e_2e_3]},\\
(t, T\setminus \{e_2, e_3\}) = (t,(\{v_{12}, v_{23}, v_{31}\}, \emptyset)) & \sim &  x_{e_1,[e_2e_3,e_2e_3]},\\
(t/e_2, T/e_2\setminus e_3) = (t/e_2,(\{v_{12}\cup v_{23}, v_{31}\}, \emptyset)) & \sim &  x_{e_1,[e_3,e_3]},\\
(t/e_3, T/e_3\setminus e_2) = (t/e_3,(\{v_{13}\cup v_{23}, v_{12}\}, \emptyset)) & \sim &  x_{e_1,[e_2,e_2]},
\eeas
where the spanning forests are indicates both in terms of $T$ and as a pair of a vertex set and an edge set.
The action of $\Delta_c$ is as follows
\beas
\Delta_c x_{e_1,[\emptyset,e_2e_3]} & = & x_{e_1,[\emptyset,e_2e_3]}\otimes 1_F+ 1_F\otimes x_{e_1,[\emptyset,e_2e_3]},\\
\Delta_c x_{e_1,[\emptyset,e_3]} & = & x_{e_1,[\emptyset,e_3]}\otimes 1_F+ 1_F\otimes x_{e_1,[\emptyset,e_3]},\\
\Delta_c x_{e_1,[\emptyset,e_2]} & = & x_{e_1,[\emptyset,e_2]}\otimes 1_F+ 1_F\otimes x_{e_1,[\emptyset,e_2]},\\
\Delta_c x_{e_1,[\emptyset,\emptyset]} & = &  x_{e_1,[\emptyset,\emptyset]}\otimes 1_F+1_F\otimes
 x_{e_1,[\emptyset,\emptyset]},\\
\Delta_c x_{e_1,[e_2,e_2e_3]} & = & 1_F\otimes x_{e_1,[e_2,e_2e_3]},\\
\Delta_c x_{e_1,[e_3,e_2e_3]} & = & 1_F\otimes x_{e_1,[e_3,e_2e_3]},\\
\Delta_c x_{e_1,[e_2e_3,e_2e_3]} & = & 1_F\otimes x_{e_1,[e_2e_3,e_2e_3]},\\
\Delta_c x_{e_1,[e_3,e_3]} & = & 1_F\otimes x_{e_1,[e_3,e_3]},\\
\Delta_c x_{e_1,[e_2,e_2]} & = & 1_F\otimes x_{e_1,[e_2,e_2]}.
\eeas
For the $\rho$ coaction, assuming all edges are massless and hence all tadpoles are forbidden, we find
\beas
\rho( x_{e_1,[\emptyset,e_2e_3]} ) & = &  x_{e_1,[\emptyset,e_2e_3]}\otimes 
x_{e_1,[e_2e_3,e_2e_3]}+x_{e_1,[\emptyset,e_3]}\otimes x_{e_1,[e_2,e_2e_3]}
+x_{e_1,[\emptyset,e_2]}\otimes x_{e_1,[e_3,e_2e_3]},\\
\rho( x_{e_1,[\emptyset,e_3]} ) & = &  x_{e_1,[\emptyset,e_3]}\otimes 
x_{e_1,[e_3,e_3]},\\
\rho( x_{e_1,[\emptyset,e_2]} ) & = &  x_{e_1,[\emptyset,e_2]}\otimes 
x_{e_1,[e_2,e_2]},\\
\rho( x_{e_1,[\emptyset,\emptyset]} ) & = &  x_{e_1,[\emptyset,\emptyset]}\otimes x_{e_1,[\emptyset,\emptyset]},\\
\rho( x_{e_1,[e_2,e_2e_3]} ) & = &  x_{e_1,[e_2,e_2e_3]}\otimes x_{e_1,[e_2e_3,e_2e_3]}
+ x_{e_1,[e_2,e_2]}\otimes x_{e_1,[e_2,e_2e_3]},\\
\rho( x_{e_1,[e_3,e_2e_3]} ) & = &  x_{e_1,[e_3,e_2e_3]}\otimes x_{e_1,[e_2e_3,e_2e_3]}
+ x_{e_1,[e_3,e_3]} \otimes x_{e_1,[e_3,e_2e_3]},\\
\rho( x_{e_1,[e_2e_3,e_2e_3]} ) & = &  x_{e_1,[e_2e_3,e_2e_3]} \otimes x_{e_1,[e_2e_3,e_2e_3]},\\
\rho( x_{e_1,[e_2,e_2]} ) & = &  x_{e_1,[e_2,e_2]} \otimes x_{e_1,[e_2,e_2]},\\
\rho( x_{e_1,[e_3,e_3]} ) & = &  x_{e_1,[e_3,e_3]} \otimes x_{e_1,[e_3,e_3]}.
\eeas
One immediately checks Theorem~\ref{thmgenerators}.

Consider now $\mathfrak{C}(t)$.
We have
\[
\mathfrak{C}(t)=\rho_\Phi\left(x_{e_1,[\emptyset,e_2e_3]}+x_{e_2,[\emptyset,e_3e_1]}+x_{e_3,[\emptyset,e_1e_2]}\right).
\]
We have for example
\beas
\rho_\Phi(x_{e_1,[\emptyset,e_2e_3]}) &  =  & \Phi_R\circ w^{-1}(x_{e_1,[\emptyset,e_2e_3]})\otimes \Phi\circ w^{-1}(x_{e_1,[e_2e_3,e_2e_3]})\\
 & & +
\Phi_R\circ w^{-1}(x_{e_1,[\emptyset,e_2]})\otimes \Phi\circ w^{-1}(x_{e_1,[e_2,e_2e_3]})\\
 & & +
\Phi_R\circ w^{-1}(x_{e_1,[\emptyset,e_3]})\otimes \Phi\circ w^{-1}(x_{e_1,[e_3,e_2e_3]})\\
 & = &
 \Phi_R((t,(\{v_{12}, v_{13}, v_{23}\}, \{e_2, e_3\})))\otimes \Phi((t,(\{v_{12}, v_{23}, v_{31}\}, \emptyset))\\
  & & +
 \Phi_R((t/e_3,(\{v_{12}, v_{13}\cup v_{23}\}, \{e_2\})))\otimes \Phi((t,(\{v_{12}, v_{13}, v_{23}\}, \{e_3\})))\\
 & & +
 \Phi_R((t/e_2,(\{v_{13}, v_{12}\cup v_{23}\}, \{e_3\})))\otimes \Phi((t,(\{v_{12}, v_{13}, v_{23}\}, \{e_2\}))).
\eeas
Indicating the trees by their edge sets for conciseness, We have 
\[
 \Phi_R((t,\{e_2,e_3\}))+ \Phi_R((t,\{e_3,e_1\}))+ \Phi_R((t,\{e_1,e_2\}))=\Phi_R(t),
\]
and
\beas
 \Phi_R((t/e_2,\{e_3\}))+ \Phi_R((t/e_2,\{e_1\})) & = & \Phi_R(t/e_2),\\
  \Phi_R((t/e_3,\{e_1\}))+ \Phi_R((t/e_3,\{e_2\})) & = & \Phi_R(t/e_3),\\
 \Phi_R((t/e_1,\{e_3\}))+ \Phi_R((t/e_1,\{e_2\})) & = & \Phi_R(t/e_1).
\eeas
Also $\Phi((t,(\{v_{12}, v_{13}, v_{23}\}, \emptyset)))=\Phi(t\setminus \{e_1,e_2,e_3\})$
and $\Phi((t,(\{v_{12}, v_{13}, v_{23}\}, \{e_i\})))=\Phi(t \setminus \{e_j,e_k\})$ for $\{i,j,k\}=\{1,2,3\}$.

We hence find
\beas
\rho_\Phi\left(x_{e_1,[\emptyset,e_2e_3]}+x_{e_2,[\emptyset,e_3e_1]}+x_{e_3,[\emptyset,e_1e_2]}\right) & = &
\Phi_R(t)\otimes \Phi(t\setminus \{e_1,e_2,e_3\})\\
 & + & \Phi_R(t/e_1) \otimes \Phi(t \setminus \{e_2,e_3\})\\
 & + & \Phi_R(t/e_2) \otimes \Phi(t \setminus \{e_3,e_1\})\\
 & + & \Phi_R(t/e_3) \otimes \Phi(t \setminus \{e_1,e_2\}).
\eeas
This confirms Lemma~\ref{lemmafrakC} and agrees, for example, with the study of the triangle on pp.(15,16) in \cite{coaction}.

\end{ex}

\subsection{Sector Decomposition}\label{SecDec}
One reason to consider Dyson--Schwinger equations for pairs $(\Gamma,F)$ lies in an accompanying sector decomposition. Sector decompositions are used in physics to allow
for stable evaluations of Feynman graphs in regions where the Feynman integrand 
suffers from infrared, ultraviolet, collinear or other kinematical singularities
\cite{Heinrich,HeinrichBinot}. For a related use of Hopf algebras in infrared singular situations see \cite{BorHerzog}. 

Sector decompositions have a rich mathematical structure in terms of generalized perutahedra for quantum field theories on Euclidean space \cite{SecEuc}, see also the discussion  in \cite{MichiSampling}. We hope that our approach allows to generalize such structures 
to the study of quantum field theories on Minkowski space in particular with regard to resulting infrared divergences.

Consider a non-negative weight on each edge of $\Gamma$ which we will think of as an edge length.  The approach to Feynman integration by sector decomposition \cite{Heinrich} involves breaking up the region of integration based on the relative order of the lengths of the edges.  As such, if there are $n$ edges then there are $n!$ sectors, but we can usefully gather sectors together based on certain features of the sector.

One very important feature of a sector, particularly for the parametric representation, is the minimum spanning tree.  An order on the edges determines a unique minimum spanning tree and this also determines the minimum monomial in the denominator in parametric representation.  For purposes such as the Hepp bound~\cite{Hepp}, only this mimimum monomial matters.

Given a fixed spanning tree, many sectors will have that tree as a minimum spanning tree.  In particular, given a fixed spanning tree, we can take any order on the edges of the spanning tree and any order on the edges not in the spanning tree, and then we can shuffle the tree and non-tree edges provided each edge not in the tree appears after all the other edges in its fundamental cycle.

Those sectors where all the edges of the minimum spanning tree are smaller than all edges not in the tree are special, as these sectors do not require renormalization.  Another way to look at this is that if we shrink the edges to length $0$ one at a time, following the order, then these are the orders where the tree shrinks first, leaving a rose, and then the petals are shrunk one by one.  Let $\mathbf{secJ}(\Gamma)$ be the number of such sectors.  Then 
\[
\mathbf{secJ}(\Gamma):=\mathbf{spt}(\Gamma) \times (e_\Gamma-|\Gamma|)!,
\]
where $(e_\Gamma-|\Gamma|)=e_T$ is the number of edges in any spanning tree $T$ of $\Gamma$ and $\mathbf{spt}(\Gamma)$ is as in Section~\ref{sec count sp tr}.  These sectors are important in \cite{MarkoDirk}.

The difference $n!-\mathbf{secJ}(\Gamma)$ comes from sectors where a loop would shrink before the full spanning tree has been retracted, that is the edges of a cycle would appear before any spanning tree does in the edge order.
These are the sectors which need a blow-up before they can be integrated out.  These sectors are taken into consideration if we work with renormalized Feynman rules. See the discussion in Appendix 3 of \cite{coaction}.

For either type of sector, if we build a graph by iterated insertion, as for example it would be generated by a Dyson-Schwinger equation, then the order of insertion determines an order on the fundamental cycles of the graph.  This does not completely determine the sector, but again we can gather together sectors with the same fundamental cycle order, as they arise from the same insertion of sectors of the sub- and co-graphs.

There is an easy combinatorial fact at play in the enumeration of the different types of sectors: the sectors covered by the cubical chain complex plus sectors corresponding to vanishing loops add to the total set of sectors. 

We can, hence, for renormalized amplitudes restrict the sector decomposition 
of physicists \cite{Heinrich} to the study of the bordification of OS untertaken by 
Vogtmann and collaborators \cite{jewel} in future work. In Section~\ref{DSE}
we will systematically build Feynman graphs from iterating one-loop necklaces which 
is then in accordance with the sector decomposition desired in physics.

An important part of sector decomposition as an integration technique is rewriting the integrand into a form suitable for each sector.  Again this becomes a question of collecting together sectors with the appropriate combinatorial properties to suit each form of the integrand, as we will explore in the examples below.

\subsubsection{Sector decomposition and iteration of graphs}
For any chosen spanning tree $T$, we get an accompanying basis of cycles $l(T,e)$.
Consider parametric Feynman rules defined in Section~\ref{parametricFR}. Choose $e\in E_L$ so that for every fundamental cycle $l(T,e)$, we can switch to variables 
\[
A_e,a_f=A_f/A_e,\,f\in E_{t_e}.
\] 
A choice of an order for the variables $A_e$, $e\in E_L$ then ensures that all poles generated from fundamental cycles (loops) are normal crossing which implies a proper sector decomposition when we later build Feynman graphs from Dyson--Schwinger equation in Section~\ref{DSE}. 
 \subsubsection{Two examples}
Let us study two examples. 
We start with the triangle graph $t$ in Figure~\ref{trianglecanvassecdec}.
\begin{figure}[H]
\includegraphics[width=14cm]{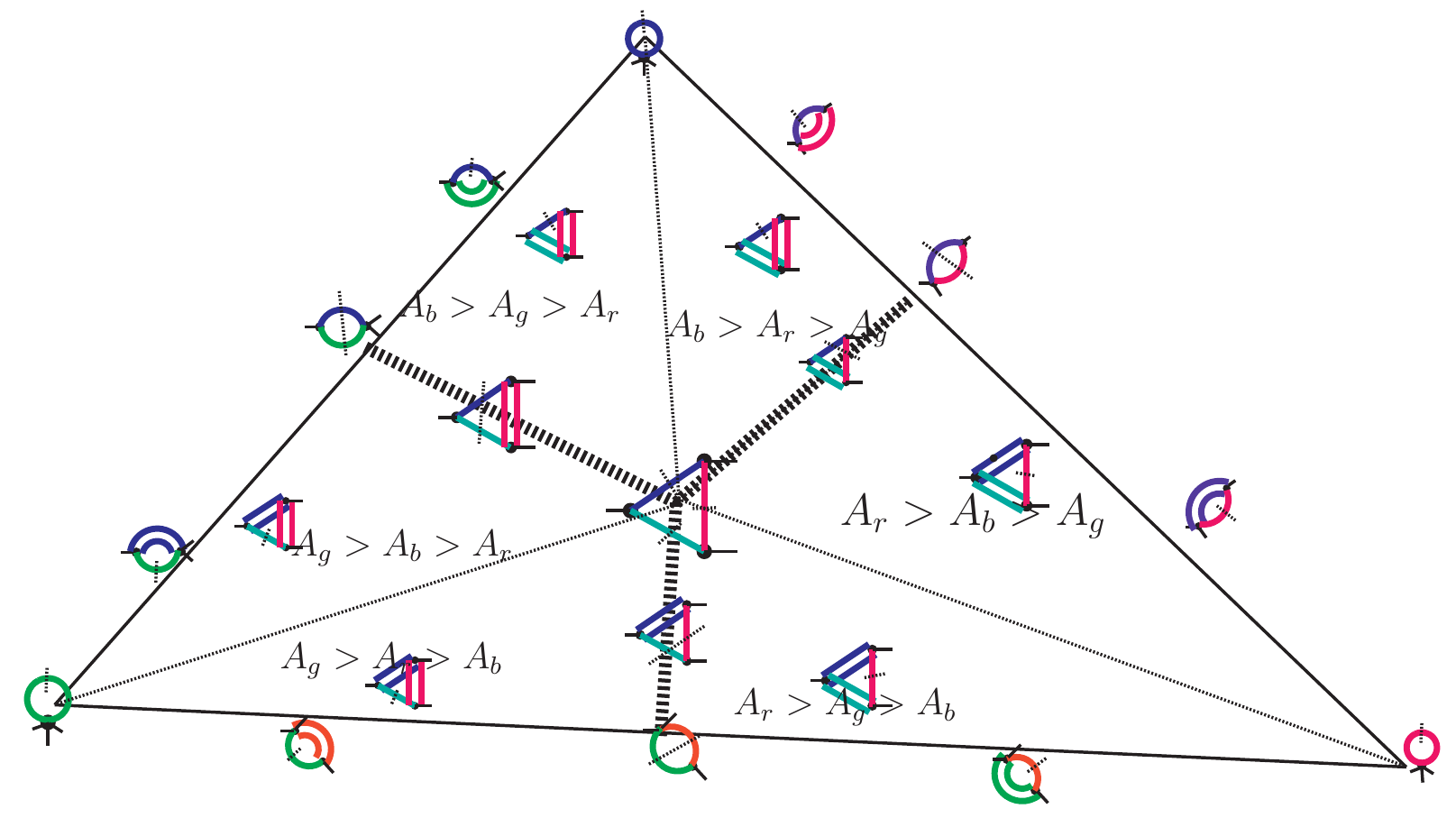}
\caption{The triangle graph $t$  on three different masses indicated by coloured edges. Edges in spanning trees $T$ or forests are indicated by doubled lines.
The other edges have a single line and are also marked by black cut across the edge. There is an associated cell for $t$ which is a $(e_t-1)$-simplex (itself a triangle). In the middle of the 2-simplex we see the triangle with all three edges onshell, corresponding to the pair $(t,T\setminus E_T)$. 
The codimension one surfaces are 1-simplices. In their middle we see bubble 
graphs obtained by shrinking an edge and putting the remaining edges onshell
corresponding to pairs $(t/e_i,T/e_i\setminus e_j)$. In the corners we have tadpoles $(t/T,\emptyset)$. The 2-simplex decomposes into six sectors corresponding to the $e_t!$ orders of possible edge length. The six sectors are grouped in three 2-cubes combining two triangles respectively which have boundaries which are four 1-simplices and four 0-simplices (corners) corresponding to Galois conjugates as indicated.}
\label{trianglecanvassecdec}
\end{figure}
Corresponding to the sector decomposition in Figure~\ref{trianglecanvassecdec}
is an integrand which in parametric variables (see Appendix~\ref{parametricFR}) reads (assuming we are, say, in a scalar field theory in $D=6$  dimensions of spacetime)
\[
\mathsf{Int}(t)(q,p)=\frac{\ln\frac{q_{rb}^2A_rA_b+q_{bg}^2A_bA_g+q_{gr}^2A_gA_r-(A_rm_r^2+A_bm_b^2+A_gm_g^2)(A_r+A_b+A_g)}{-\mu^2A_rA_b-\mu^2A_bA_g-\mu^2A_gA_r-(A_rm_r^2+A_bm_b^2+A_gm_g^2)(A_r+A_b+A_g)}}{(A_r+A_b+A_g)^3}\Omega_t,
\]
with $\mu^2>0$ a renormalization point and $\Omega_t$ the 2-form
\[
\Omega_t=A_g d\!A_b\wedge d\!A_r-A_b d\!A_r\wedge d\!A_g+A_r d\!A_g\wedge d\!A_b,
\]
and further where we have $p\in\mathbb{P}^2(\mathbb{R}_+)$ and $q\in Q^L(t)$, with $Q^L(t)$ a three dimensional real vectorspave generated by  Lorentz invariants
$q_{rb}^2,q_{bg}^2,q_{gr}^2$ as scalar products of external momenta at the vertices
whilst mass squares $m_r^2,m_b^2,m_g^2$ are kept fixed with positive real part and small negative imaginary part.

$\mathsf{Int}(t)(q,p)$ can be integrated against $\mathbb{P}^2(\mathbb{R}_+)$
and gives a function $\Phi_R(t)(q)$ whose monodromy we are after.

By construction it can be also integrated against the intersection of
$\mathbb{P}^2(\mathbb{R}_+)$ with the interior of any of the six sectors $A_i>A_j>A_k$ indicated in Fig.(\ref{trianglecanvassecdec}).\footnote{Along the 1-simplices defined by $A_i=0$ the integrand has to be modified according to the increased degree of divergence. As a result when gluing such simplices to describe graph complexes we have to follow the approach outlined by Berghoff in \cite{Marko}.}

In each of the six sectors we can implement the sector decomposition in accordance with the fundamental cycles $l(T,e)$ underlying the graph. For the triangle graph
we have three spanning trees and a single loop, which makes three fundamental cycles to consider.

Each corresponds to one of the 2-cubes in  Figure~\ref{trianglecanvassecdec}.
For example the fundamental cycle $l(\{e_b,e_r\},e_g)$ corresponds to the two sectors $A_g>A_b>A_r$ and $A_g>A_b>A_r$.

In the integrand we can rescale $a_b=A_b/A_g,a_r=A_r/A_g$, set $A_g=1$ and integrate 
in the two accompanying sectors to find
\[
\Phi_R(t)^{grb}:=\int_0^\infty \left(\int_0^{a_r}\frac{\ln\frac{q_{rb}^2a_ra_b+q_{bg}^2a_b+q_{gr}^2a_r-(a_rm_r^2+a_bm_b^2+m_g^2)(a_r+a_b+1)}{-\mu^2a_ra_b-\mu^2a_b-\mu^2a_r-(a_rm_r^2+a_bm_b^2+m_g^2)(a_r+a_b+1)}}{(a_r+a_b+1)^3}d\!a_b\right)d\!a_r,
\]
for $A_g>A_r>A_b$ ($a_r>a_b$) and
\[
\Phi_R(t)^{gbr}:=\int_0^\infty \left(\int_0^{a_b}\frac{\ln\frac{q_{rb}^2a_ra_b+q_{bg}^2a_b+q_{gr}^2a_r-(a_rm_r^2+a_bm_b^2+m_g^2)(a_r+a_b+1)}{-\mu^2a_ra_b-\mu^2a_b-\mu^2a_r-(a_rm_r^2+a_bm_b^2+m_g^2)(a_r+a_b+1)}}{(a_r+a_b+1)^3}d\!a_r\right)d\!a_b,
\]
for for $A_g>A_b>A_r$ ($a_b>a_r$). For the spanning tree $T_{rb}\in\mathcal{T}(t)$ on edges $e_r,e_b$ we set 
\[
\Phi_R((t;T_{rb}))=\Phi_R(t)^{grb}+\Phi_R(t)^{gbr},
\] 
and
\[
\Phi_R(t)=\sum_{T\in\mathcal{T}(t)}\Phi_R((t,T)),
\]
becomes a sum over three summands, one for each spanning tree, and each summand is a sum over all (two) sectors provided by the possible (two) orderings of edges in the spanning tree. This construction is generic as any spanning tree $T\in\mathcal{T}(G)$
of a graph $G$ gives rise to an $e_T$-cube in its cubical chain complex. Such a cube then has a decomposition in $e_T!$ simplices.

\medskip

Next consider the Dunce's cap graph $dc$ given in Fig.(\ref{dunceg}).
\begin{figure}[H]
\includegraphics[width=7cm]{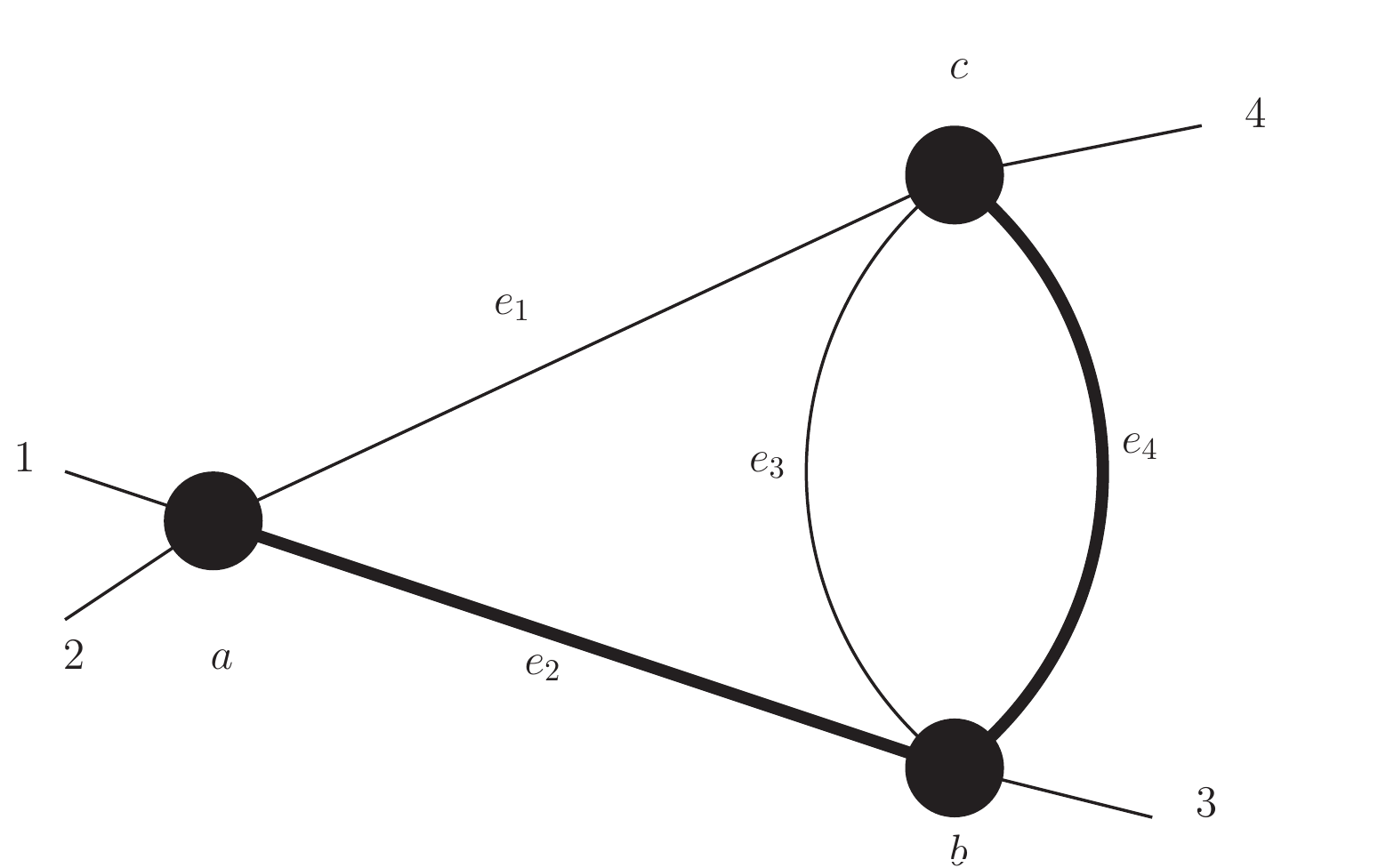}
\caption{The Dunce's cap graph $dc$. The spanning tree on edges $e_2,e_4$ is indicated by bold lines.}
\label{dunceg}
\end{figure}
Choose the spanning tree $T=(e_2,e_4)$ so that the loops are $g_1= l_1=(e_1,e_4,e_2)$ corresponding to the fundamental cycle $l(T,e_1)$
and $g_2= l_2=(e_3,e_4)$ corresponding to $l(T,e_3)$. This implies $A_3>A_4$ and $A_1>A_2,A_1>A_4$ as when we shrink edges, the edges in $t_e$ shrink before the edge $e$.

Consider the choice that $g_1$ is the co-loop and $g_2$ the subloop.
\[
dc=B_+^{(g_1/(l_1\cap l_2),T/e_4)}((g_2,e_4)).
\]
We have $A_1>A_3$ as $g_2$ is the subloop
and $g_1/(l_1\cap l_2)=g_1/e_4$ the co-loop.

Hence we rescale in thze first Symanzik polynomial of $dc$ ${\mathfrak{s}_\psi}_{dc}$
\[
((A_1+A_2)(A_3+A_4)+A_3A_4)\to A_1^2((1+a_2)(a_3+a_4)+a_3a_4)\to
A_1^2 a_3((1+a_2)(1+b_4)+a_3b_4). 
\] 
Following \eqref{FRInt} the corresponding integrand rescales to
\beas
\mathsf{Int}(dc)(q^2,p_1^2,p_2^2,\mu^2;a_2,a_3,b_4) & = & \left(\frac{\ln\frac{(q^2 a_2(1+b_4)+p_1^2 a_3b_4+p_2^2 a_2a_3b_4)}{(\mu^2a_2(1+b_4)+\mu^2a_3b_4+\mu^2a_2a_3b_4)}}{a_3((1+a_2)(1+b_4)+a_3b_4)^2}\right.\\
 &  & \left.  -  \frac{\ln\frac{q^2a_2(1+b_4)+\mu^2 a_3b_4(1+a_2)}{\mu^2a_2(1+b_4)+\mu^2a_3b_4(1+a_2)}}{a_3(1+a_2)^2(1+b_4)^2}
  \right) da_2da_3db_4,
\eeas
after integrating $A_1$ (the longest edge in the sectors under consideration) by the exponential integral. We set all masses to zero to zero for a succinct expression. This does not alter the argument.

Note that 
\[
\mathsf{Int}(dc)(q^2,p_1^2,p_2^2,\mu^2;a_2,0,b_4)=0,
\]
 as both terms on the right have a pole at $a_3=0$ with identical residue
\[
\frac{(\ln\frac{q^2}{\mu^2})da_2db_4}{(1+a_2)^2(1+b_4)^2}.
\]

The above integrand covers the three sectors $A_1>A_3>A_2>A_4,A_1>A_3>A_4>A_2$
and $A_1>A_2>A_3>A_4$ (the sector where renormalization, or bordification, is genuinely needed). 

A similar analysis for $g_1$ providing the subloop and $g_2$ the co-loop covers the sectors $A_3>A_1>A_2>A_4,A_3>A_1>A_4>A_2$.
This gives five sectors for the choice $T=e_2,e_4$. The choices $T=e_1,e_4$,
$T=e_2,e_3$, $T=e_1,e_3$ are similar and this covers twenty sectors altogether. 

The choice $T=e_1,e_2$ covers the remaining four sectors. See also \cite{coaction}.
A finer analysis to be given in future work exhibits that the boundary between sectors is generated by the coactions studied above in Section~\ref{cointbi}.
\subsection{Two more coactions}\label{sec more coactions}
For completeness and future use let us store two more coactions.
\subsubsection{$H_{C}\to H_{C}\otimes H_{pC}$}
$H_C\subsetneq H_{pC}$ is a vector subspace of $H_{pC}$. Hence 
\[
\Delta_{pC}:H_{C}^>\to H_{C}^>\otimes H_{pC}
\]
 coacts and similarly
\[
\Delta_{pC}:H_{C}^0\to H_{C}^0\otimes H_{pC}^0.
\]
Figure~\ref{hpccoact} gives a typical example.
\begin{figure}[H]
\includegraphics[width=14cm]{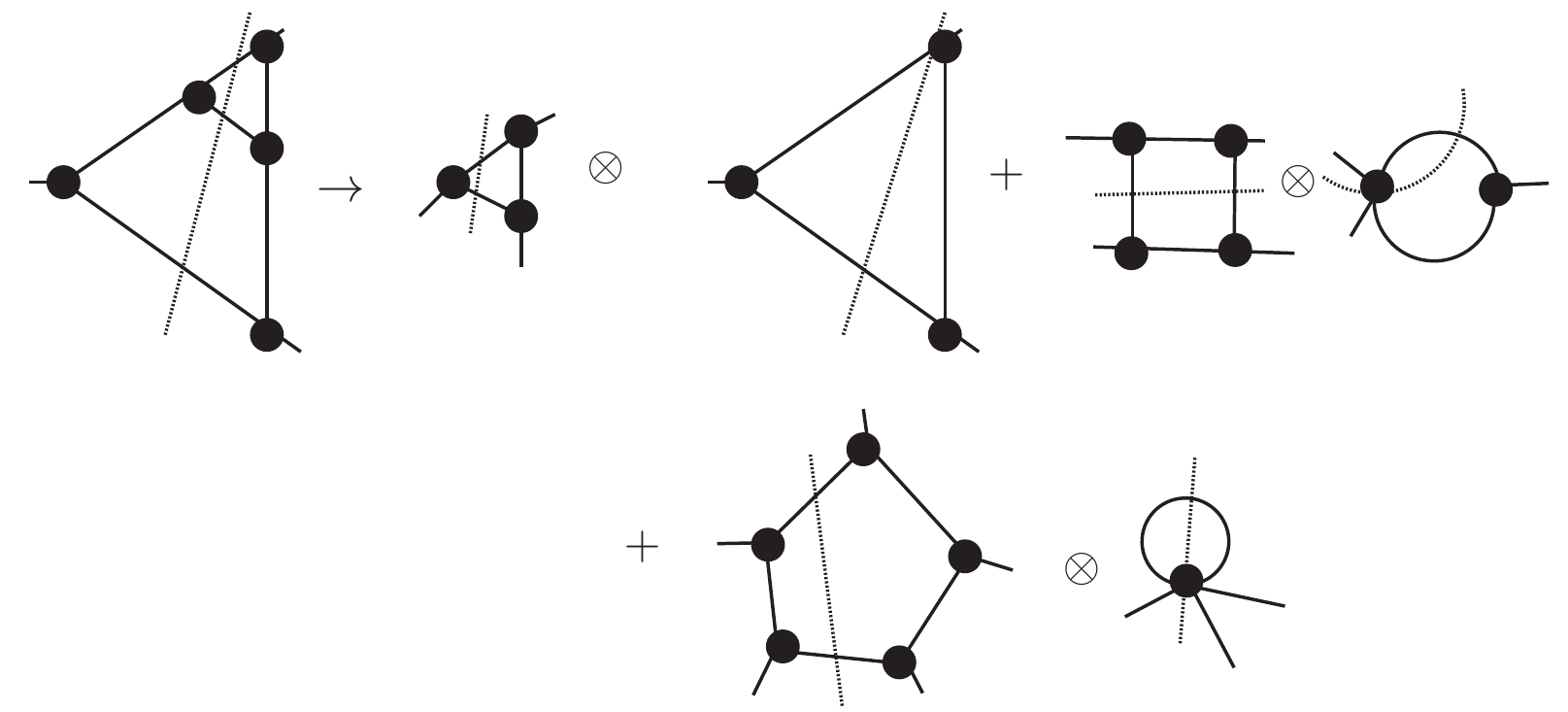}
\caption{$\Delta_{pC}$ regarded as a coaction maps to graphs in $H_C$ on the left and to graphs in $H_{pC}$ on the right.
}
\label{hpccoact}
\end{figure}
This is obvious from the coassociativity of $\Delta_{pC}$. We can regard $\Delta_{pC}$ simultaneously as such a coaction and a coproduct on $H_{pC}$.
\begin{rem}
Note that in the case that $\Gamma\in H_C^0$ has only propagator subgraphs, we have $\Delta_{pC}(\Gamma)\in H_C^0\otimes H_C^0$.
\end{rem}  
\subsubsection{$H_{nC}\to H_{pC}\otimes H_{nC}$}
This coaction concerns variations on non-principal sheets \cite{DirkEll}. It relates to the jewels of Vogtmann and collaborators \cite{jewel}. It uses $\Delta_{pC}$ to identify subgraphs $H_{pC}^>\ni \gamma\subsetneq \Gamma\in H_{nC}$, the vector space of graphs where we allow arbitrary subsets of edges or vertices to be cut disregarding the requirement that the graph is correspondingly cut into disconnected parts. 

We define
\[
\bar{\Delta}_{pC}(\Gamma)=\One\otimes \Gamma+\sum_{\gamma\subsetneq \Gamma, \gamma\in H_{pC}^>}\gamma\otimes \Gamma/\gamma.
\]
Figure~\ref{hpcbarcoact} gives an example.
The graphs $\Gamma \in H_{nC}$ where the number of separations is unity play a very special role here:
\[
\mathbf{nos}(\Gamma)=1\Leftrightarrow \bar{\Delta}_{pC}(\Gamma)=\One\otimes \Gamma.
\]
We do not provide further details as we will not use this coaction later on.
\begin{figure}[H]
\includegraphics[width=10cm]{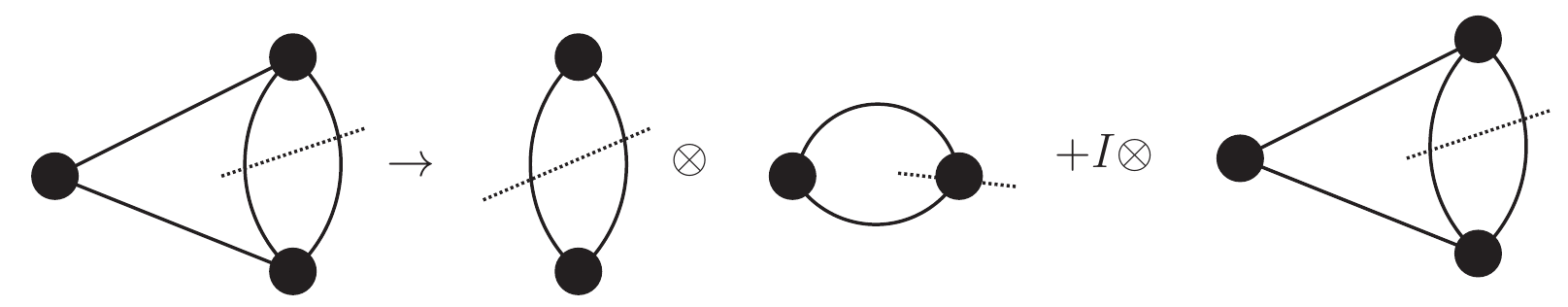}
\caption{$\bar{\Delta}_{pC}:H_{nC}\to H_{pC}\otimes H_{nC}$ maps here a graph $\Gamma$ in $H_{nC}\setminus H_{pC}$ to a graph $\gamma$ in $H_{C}\subsetneq H_{pC}$ on the left and a graph $\Gamma/\gamma$ in $H_{nC}$ on the right.}
\label{hpcbarcoact}
\end{figure}

\begin{rem}
The coactions above were based on $\Delta_{core}$. We can pursue a similar analysis using 
any $\Delta_V$ instead.
\end{rem}

\section{DSEs in core and quotient Hopf algebras}\label{DSE}
Combinatorial Green functions as formal series are defined as a sum of  all bridgeless graphs which have an identical external leg structure. Upon evaluation by Feynman rules they give 1-particle irreducible (1PI) Green functions as formal series in the couplings.

These combinatorial Green functions then satisfy certain functional equations that are combinatorial version of the Dyson--Schwinger equations (DSEs) of the theory.  This theory has been well-developed for the usual Green functions \cite{Karenbook,FoissyDSE} and has found use in other areas of mathmatics
\cite{Ralph,BorVogt}  and physics \cite{Balduf,KUvBY}.  Analogous results apply when we are working with cuts.  Developing this theory is the subject of the current section.  Other than the set-up itself, there are three main points of particular note, each of which is detailed below. First, there is a connection with the assembly maps of \cite{fourauthors}. Second, the invariant charges and their interaction with the coproduct are crucial.  Third, the cut inverse propagator is special and in particular we do not need to consider insertions on either side of the cut propagator.  This simplifies the combinatorics.

\subsection{Dyson--Schwinger and cut Dyson--Schwinger setup}\label{DSS}
We define combinatorial Green functions for scattering.
For a start we define appropriate series of graphs.

The first thing we need is the order in the couplings $g_\Gamma$ of a graph $\Gamma$, so define $g_n$ for each $n\geq 3$, set $g_2=1$ and define
\[
g_\Gamma:=\left(\prod_{v\in V_\Gamma} g_{\mathbf{val}v}\right)/ g_{l_\Gamma}.
\]
For a Cutkosky graph or a pair $(\Gamma,F)$, the order in the couplings is defined in the same way, since the cuts of a Cutkosky graph only involve edges, as does the cut defined by a spanning forest.

For a pre-Cutkosky graph $\Gamma$ we 
have couplings $g_p$ for each integer partition $p$.\footnote{An \emph{integer partition} is simply a finite multiset of positive integers.  Since the order does not matter, we can write it as a list where we take the convention to write the integers, called the \emph{parts}, in weakly decreasing order.  For example, $(2,1,1)$ is an integer partition and furthermore, we say it is a partition \emph{of} $4$ since the integers in the list sum to 4.  If we have a set partition of a set of size $n$ then the sizes of the parts give an integer parititon of $n$.}  The parts of the integer partition are the sizes of the parts of the set partition of a cut vertex.  If $v$ is a cut vertex let $p(v)$ be this integer partition.  Then define
\[
g_\Gamma:=\left(\prod_{v\in V_\Gamma} g_{p(v)}\right) / g_{p(\Gamma)},
\]
where $p(\Gamma)$ is the partition whose parts are the sizes of the parts of the set partition of the external edges of $\Gamma$ induced by the cut of $\Gamma$.

\begin{rem}
  Note that, as a physicist would expect, we do not have a coupling for uncut edges, though for notational convenience it is handy to set $g_2=1$.

  For cut edges we do need something like a coupling to keep track of the number of cut edges in the combinatorics below.  We use $g_{1,1}$ for this purpose.  One could further refine the situation by using different variables in place of $g_{1,1}$, one for each different mass appearing among the onshell edges. 
\end{rem}

\begin{figure}[H]
\includegraphics[width=14cm]{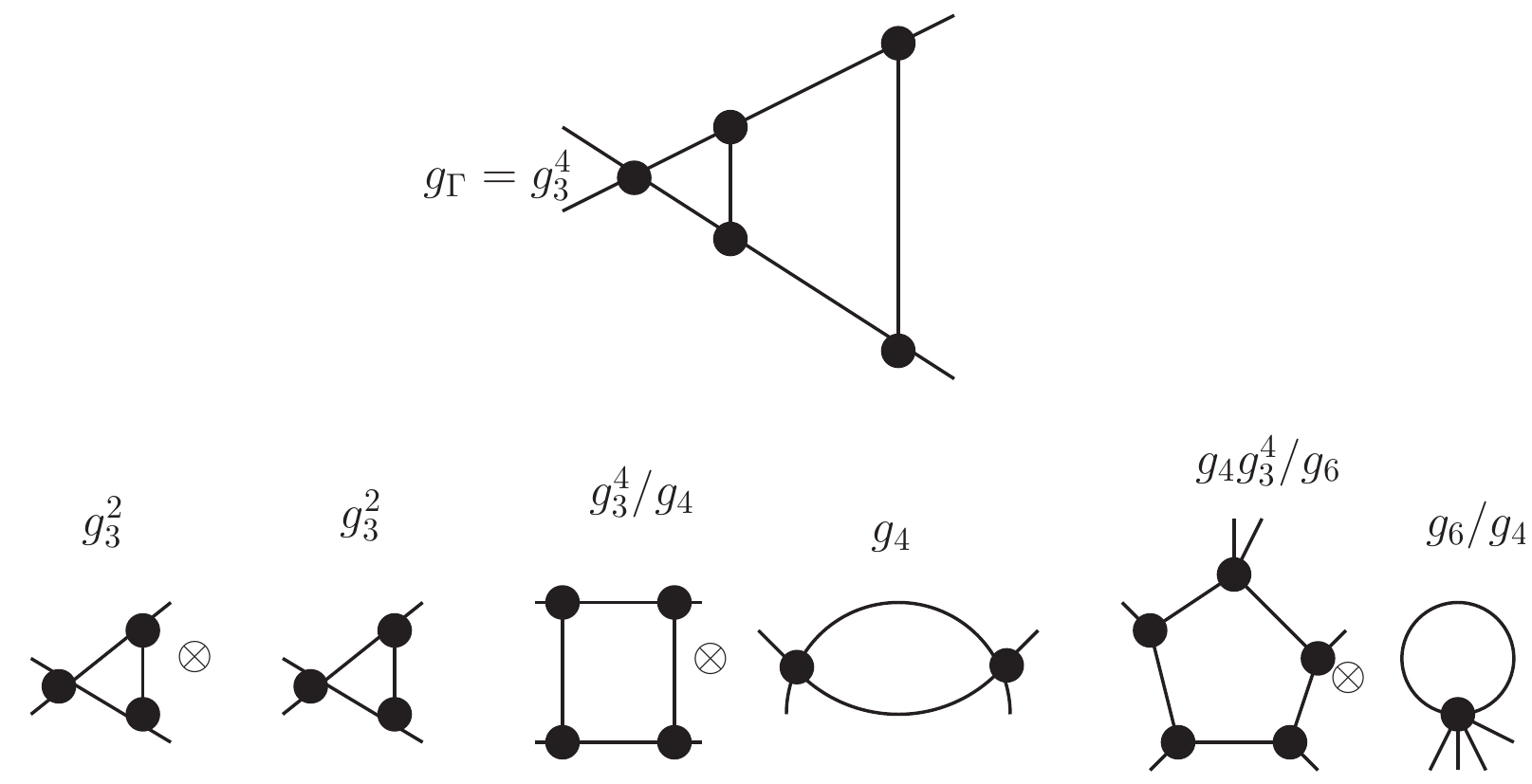}
\caption{We start with a two-loop 4-point graph $\Gamma$  which has order $g_\Gamma=g_4g_3^4/g_4=g_3^4$. Under $\tilde{\Delta}_{core}$ it decomposes into sub- and co-graphs  with orders 
$g_3^2\otimes g_3^2$, $g_3^4/g_4\otimes g_4$ and $g_4g_3^4/g_6\otimes g_6/g_4$
on the two sides.
$\Gamma$ allows for three cycles with four or six external legs and is based on three- or four-valent vertices.
 Note the appearance of denominators in the orders of sub- or co-graphs.
 If we let $\mathsf{Ord}:H_{core}\to\mathbb{R}$ be the map $\Gamma\to g_\Gamma$, we have 
 $\mathsf{Ord}(\Gamma)=\mathsf{Ord}(\Gamma^\prime)\otimes \mathsf{Ord}(\Gamma^{\prime\prime})$
 for any summand in the reduced coproduct. 
}\label{figeight}
\end{figure}

\begin{rem}
  If we wish to also include the order in $\hbar$, this can be obtained by scaling the $g_i$ by powers of $\hbar$, taking advantage of Euler's formula.  Specifically, for a core graph with $k_i$ vertices of degree $i+2$, for $i\geq 1$, and $n$ external edges, the correct power of $\hbar$ is $\hbar^{1-\frac{-n+\sum_{i} ik_i}{2}}$
  and so scaling each $g_i$ by $\hbar^{i/2}$ along with an overall scaling independent of the graph we will obtain the correct power of $\hbar$.

Note also that with $k_i$ and $n$ as above, the product $\left( \prod_{i\in N} g_{i+2}^{k_i-\delta_{i+2,n}}\right)$ is another expression for $g_\Gamma$ for a core graph $\Gamma$.

See  Figure~\ref{figeight} for an example.

  The negative powers in these products are useful as they make the coproduct behave nicely with the multigrading from the coupling constants.
\end{rem}

To build the combinatorial Green functions we are interested in the set of all graphs (whether core, Cutkosky, etc) with a fixed external edge structure.  From a combinatorial perspective we treat the external edges as labelled (so permuting them will give a different graph unless some automorphism of the rest of the graph can undo the permutation.) This implies that in the core case the external edge structure is an ordered list of half edges, but since there is only one type of half edge in this theory, the only information carried by the external structure is the number of external edges.  Then define the set of core graphs with $n$ external edges to be $S_{core}^n$, and the combinatorial Green function
\[
G_{core}^n = \One\pm \sum_{\Gamma\in S_{core}^n} \frac{g_\Gamma}{|\mathrm{Aut}(\Gamma)|}\Gamma
\]
where the sign is $+$ if $n>2$ and $-$ if $n=2$ and where $\mathrm{Aut}(\Gamma)$ is the set of automorphisms of $\Gamma$.

In the pre-Cutkosky and Cutkosky cases, the cuts give a set partition structure on the external edges.  The half edges are all of the same type, so if we forget the cyclic order of the external edges, then the remaining information carried by the external structure is only the sizes of the parts of the set partition, that is, the information is an integer partition.  This integer partition gives the number of external edges in each piece into which the graph must be cut.  Then define, for an integer partition $p$, the set of Cutkosky graphs with external structure $p$ to be $S_{C}^p$.  Define $S_{pC}^p$ similarly for pre-Cutkosky.  

For pairs $(\Gamma,F)$, the external leg structure is as in the Cutkosky case, but now we define the set $S_{GF}^p$ to be the set of pairs $(\Gamma,F)$ which are compatible with the partition $p$. 

Let $\not\sim$ be the imaginary part of the inverse free propagator (which evaluates to $\i \pi\delta(p^2-m^2)$, see Sec.(\ref{sec propagator}) below). 
For $\bullet \in \{C, pC, GF\}$ define the combinatorial Green functions to be
\[
G_{\bullet}^{1,1} = \not\sim + \sum_{\Gamma\in S_{\bullet}^{1,1}} \frac{g_\Gamma}{|\mathrm{Aut}(\Gamma)|}\Gamma,
\]
\[
G_{\bullet}^p = \sum_{\Gamma\in S_{\bullet}^p} \frac{g_\Gamma}{|\mathrm{Aut}(\Gamma)|}\Gamma
\]
when $p$ has at least two parts but $p\neq (1,1)$, and
\[
G_\bullet^{(n)} = \One\pm \sum_{\Gamma\in S_{core}^{(n)}} \frac{g_\Gamma}{|\mathrm{Aut}(\Gamma)|}G
\]
when $p$ is the partition with the single part $n$.  In the last case the sign is $-$ when $n=2$ and $+$ otherwise.  

Note when the partition is $(n)$ and we are in the Cutkosky or pre-Cutkosky case then all the graphs involved in the sum are essentially core graphs.  Specifically, the graphs are all of the form $(\Gamma,\Gamma)$ with $\Gamma$ core.  Identifying $\Gamma$ and $(\Gamma,\Gamma)$ in these cases we have that $G_\bullet^{(n)} = G_{core}^n$ for $\bullet \in \{C,pC\}$.  For $G_{GF}^{(n)}$ almost the same thing is true, except that we sum over all $(\Gamma,T)$ pairs with $T$ a spanning tree of $\Gamma$ and $\Gamma$ as in the sum of $G_{core}^n$.

Note also that the $G_{\bullet}^{1,1}$ case is quite special.  It gets a positive sign despite being an inverse propagator, as if we have a cut propagator insertion we only want to make one cut, not a sequence of cuts, along that propagator.\footnote{This avoids products of distributions $\Theta(p^2-M^2)$
with $M$ a sum of masses  of cut propagators. Products would  show up for repeated cuts at similar self-energies and would lead to an ill-defined product of distributions with coinciding support. 
The situation is more intricate when we have flavour indices at Green functions and their product is matrix-valued. Such a generalization is straightforward but needs more elaborate notation.}  For the uncut inverse propagator we take the convention to include the negative sign so that interpreting the  propagator as a geometric series gives no additional signs.  The $G_\bullet^{1,1}$ case has additional special properties, most notably that no other propagator insertions can sit beside it on the same propagator.  $G_\bullet^{1,1}$ is discussed in detail in Subsection \ref{sec propagator}.

Our task below is to write
$G_\bullet^p$ as a solution to a fixed-point equation in Hochschild cohomology, similarly to what has been done in previous work for combinatorial Dyson--Schwinger equations.

These formal series are solutions to fixed-point equations which are formulated using maps (Hopf algebra endomorphisms) $B_+^\gamma$, for suitable primitive graphs $\gamma$.   In the following, $\bullet$ can be $core$ or $pC$.
Under the Feynman rules one obtains integral equations which contain information beyond perturbation theory \cite{KUvBY,Karenbook,Borinsky, DunneB}.

Here we define $B_+^\gamma$ as a map
\[
B_+^\gamma: H_{\bullet}\to \langle H_{\gamma} \rangle\subsetneq \mathbf{Aug}_{\bullet}.
\]
$\langle H_\gamma\rangle$ is the $\mathbb{Q}$-linear span of graphs which have $\gamma$ as an
ultimate co-graph, that is, $\tilde{\Delta}_{\bullet}$
generates $Y\otimes \gamma$ for some $Y\in H_{\bullet}$.
Specifically, define
\be\label{bplus}  
B_+^\gamma(X):=\sum_\Gamma \frac{\mathbf{bij}(\gamma,X,\Gamma)}{|X|_V}\frac{1}{\mathrm{maxf}(\Gamma)}\frac{1}{[\gamma|X]} \Gamma
\ee
in the notation of \cite{anatomy,BrownKreimer}, where the sum is over $\Gamma$ such that $X\otimes \gamma$ appears in the coproduct of $\Gamma$ and the coefficient is designed to account for overcounting.

Since all the primitives in our context are 1-loop graphs, the coefficient in the definition of $B_+$ can be slightly simplified, but this will not be important for us.

\begin{rem}
  For a map $B$, the Hochschild one-cocyle property is
\be\label{Hochschild}
\Delta_{\bullet} \circ B(\cdot)=B(\cdot)\otimes \One+(\mathrm{id}\otimes B)\Delta_{\bullet}.
\ee 
  
The individual maps $B_+^\gamma$ are not Hochschild one-cocycles.  As in most interesting physical situations \cite{anatomy} we must take a sum of $B_+^\gamma$ and apply the resulting operator to appropriate sums of graphs to obtain the one-cocyle property.
For all our purposes we will be in such situations.

\end{rem}

It will be useful in the following to collect together a product of Green functions for each primitive graph.  Recalling section~\ref{pC necklaces} we are representing primitive graphs by necklaces, both in the cut and uncut cases.

To a pre-Cutkosky necklace $\omega$ associated to the primitive pre-Cutkosky graph $\gamma$ assign the product of Green functions over edges $e_i$ and vertices $v_i$,
\[
\Pi_\omega:=\prod_{i=1}^{v_\gamma} \frac{G_{v_i}}{G_{e_i}},
\]
where 
$G_{e_i}=(G^{1,1}_{pC})^{-1}$ for $e_i\in C_\gamma$, $G_{e_i}=G^2_{core}$ for $e_i\not\in C_\gamma$ and 
$G_{v_i}=G^{p(v_i)}_{pC}$ for $v_i$ a cut vertex and $G_{v_i}=G^{\textbf{val}(v_i)}_{core}$ for $v_i$ an uncut vertex.

Note that for a cut edge $e$ we have $G_{e}=(G^{1,1}_{pC})^{-1}$, this is because in the definition of $\Pi_\omega$ the edge factors $G_e$ all appear in the denominator, but the cut edge factors should be in the numerator since an edge can only be cut once; it cannot have a sequence of cut insertions put into it, while an uncut edge can have a sequence of (uncut) edge inesrtions.

The $\Pi_\omega$ are reminiscent of the combinatorial invariant charges.  This will be discussed further in Remark~\ref{pivsQ}.

The $\Pi_\omega$ can also be defined for core necklaces simply by taking the above definition in the case that no cuts appear.

Then, using the $\Pi_{\omega}$, we have the following lemma.
\begin{lem}\label{lemmagraphins}
  \[
  \sum_{\gamma \sim p} \frac{g_{\gamma}}{|\mathtt{Aut}(\gamma)|}B_+^{\gamma}(\Pi_{\omega(\gamma)})=\sum_{\Gamma\sim p} \frac{g_\Gamma}{|\mathtt{Aut}(\Gamma)|}\Gamma,
  \]
  where the sum on the left hand side is over all primitive pre-Cutkosky graphs graphs $\gamma$
  compatible with a chosen partition $p$ of external edges and $\omega(\gamma)$ is the necklace associated to $\gamma$,  while the sum on the right hand side is over all Cutkosky graphs $\Gamma$ which are compatible with $p$.
\end{lem}
\begin{proof}
  The coefficient in the definition of $B_+$ is designed to divide out by the overcounting so that each graph appears exactly $\frac{1}{|\mathtt{Aut}(\gamma)|}$ times.  For details, see \cite{anatomy,BV} from which know that Eq.(\ref{bplus}) is in accordance with graph-counting. See also \cite{FoissyDSE}.

  The only additional thing to prove in our case is that the couplings work out.  On the right hand side, we have $g_G$, that is a product of a coupling for each vertex of $\Gamma$ divided by the coupling corresponding to the vertex that would be given by the external edges of $\Gamma$ were the internal edges contracted.  On the left hand side, in the argument to $B_+$ contributes the product of the couplings for all vertices of the inserted graphs along with the inverse of the couplings for all vertices given by external edges for each insertion.  That is, the argument to $B_+$ contributes the product of the couplings for all vertices of the inserted graphs along with the inverse of the couplings for the vertices of $\gamma$.  Multplying by $g_{\gamma}$ 
  gives the correct order matching $g_\Gamma$.
  \end{proof}
\begin{rem}
In Figure~\ref{dPc} in the last line on the rhs we see a graph $\gamma$ representing a necklace $\omega$ with 
\[
\Pi_\omega=\frac{G_{core}^3\left[G_{pC}^{1,2}\right]^2}{[G_{core}^2]^3}.
\]
\end{rem}

\subsection{Assembly maps vs Hochschild 1-cocycles}\label{assemblysec}

The infinite series alluded to above and also in Sec.(\ref{DSS}) can be obtained as solutions to fixed point equations --- combinatorial Dyson--Schwinger equations --- using the maps $B_+^{x}$ such that $x$ is either in $H_{core}$ or $H_{pC}$.
\begin{rem} 
In such combinatorial Dyson--Schwinger equations, series of graphs with (possibly a common partition of) given external legs form a blob.  The blob represents a place to insert into, and doing the insertion in all possivle ways gives the series where any graph compatible with the blob is inserted term-by-term. 
\end{rem}
Such insertions assemble new graphs from elements in such blobs, using an underlying $x$  providing  a gluing map $B_+^x$. This sums over bijections between external edges of graphs in such a blob and half-edges of corollas or edges of $x$ as in Equation~\ref{bplus}.

For a single bijection
this is similar to the assembly maps 
of \cite{fourauthors}, and we exhibit the similarity in the following 
Fig.(\ref{assembly}). 
\begin{figure}[H]
\includegraphics[width=14cm]{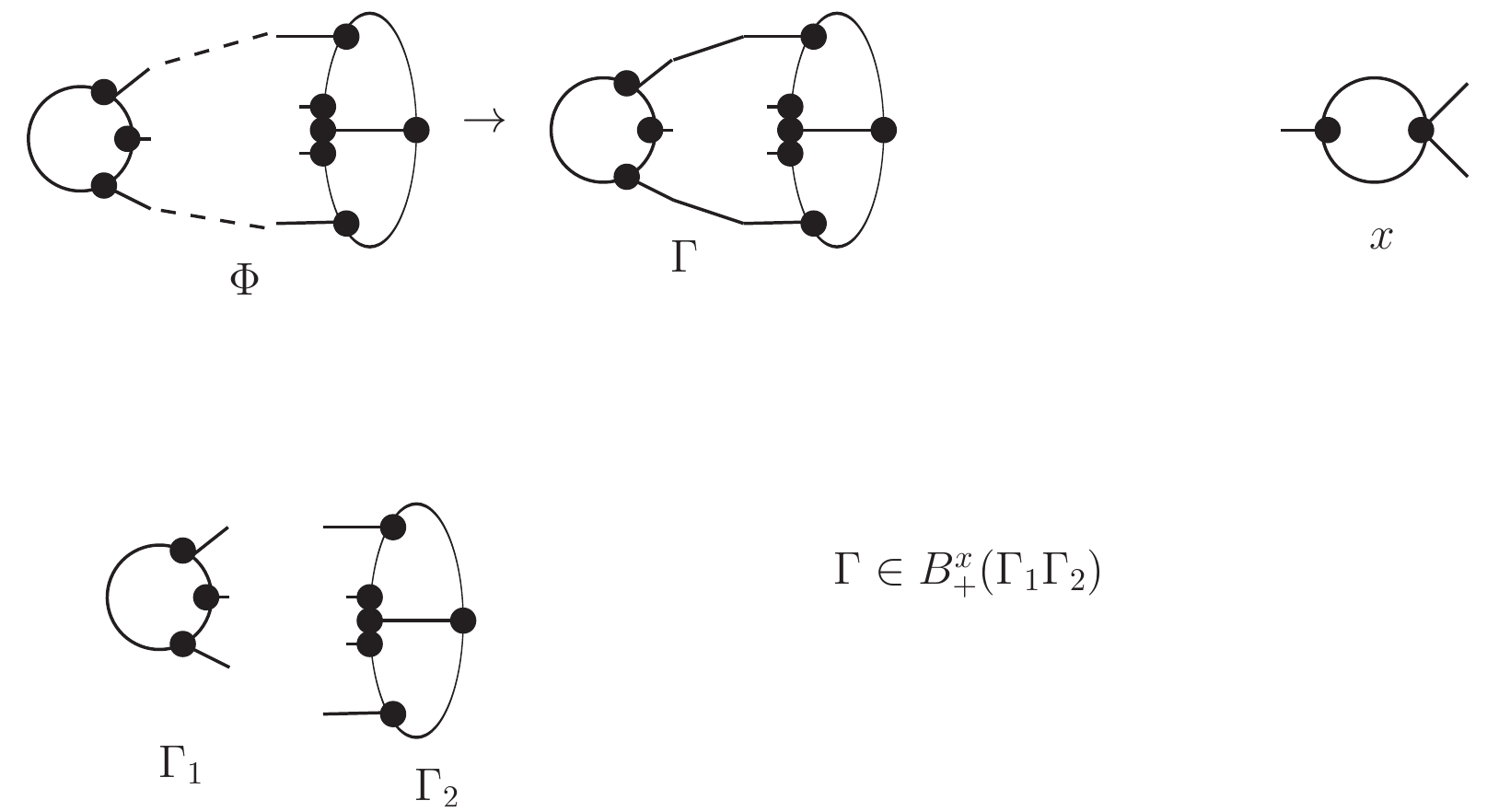}
\caption{Assembling. 
The dashed lines indicate an assemply map $\Phi$ (in the notation of \cite{fourauthors}, see in particular their Fig.(7)) which glues the graphs $\Gamma_1,\Gamma_2$ to obtain $\Gamma$. Similarly, $\Gamma$ shows up (we write $\Gamma\in B_+^x(\Gamma_1\Gamma_2)$) in the graphs formed by the map $B_+^x$, based on a Hochschild one-cocycle $B_+^x$ for the graph $x$, when acting on the product (disjoint union) $\Gamma_1\Gamma_2$ by inserting the three-point graph $\Gamma_1$ into the three-point vertex in $x$, and the four-point graph $\Gamma_2$ into the four-point vertex in $x$.  Note that $\mathsf{vcd}(\Gamma)=8$ and so $g_\Gamma=g^{\mathsf{vcd}(\Gamma)}$ if we set $g_n=g^{n-2}$.}
\label{assembly}
\end{figure}
The maps $B_+^x$ alluded to above appear as a sum over   gluing patterns $\phi$.  The gluing patterns are themselves captured by the graph $x$ into which we insert, along with the specific bijection of external edges of the inserted graph with half edges of the insertion places in $x$.  An important part of the structure of the gluing pattern is the graph $x$.

In \cite{fourauthors} gluing patterns were studied for graphs $\Gamma\in X_{n,s}$ where
the set $X_{n,s}$ is determined by 
\be\label{Xns}
\Gamma\in X_{n,s}\Leftrightarrow n=|\Gamma|,\,s=|L_\Gamma|. 
\ee

A gluing pattern allows one to assemble $\Gamma\in X_{n,s}$ from $x_i\in X_{n_i,s_i}$,
$1\leq i\leq k$ say, using a gluing pattern $\phi$ a a map which identifies endpoints 
of external edges $e\in L_{x_i}$ so as to glue them to edges $e\in E_\Gamma$.

In the notation of \cite{fourauthors} the virtual cohomological dimension
$\mathsf{vcd}$ for $X_{n,s}$ is $2n-3+s$.  We write for $\Gamma\in X_{n,s}$ simply $\mathsf{vcd}(\Gamma)=2n+s-3$ and it fulfills
\[
\mathsf{vcd}(\Gamma)=(k-1)+\sum_{i=1}^k \mathsf{vcd}(x_i).
\] 
This same counting is captured in our set up by replacing the coupling constants $g_n$ with powers of a single coupling $g$ following the rule $g_n=g^{n-2}$.  With this substitution we have 
\[
g_\Gamma=g^{\mathsf{vcd}(\Gamma)}.
\]

In \cite{fourauthors} the gluing patterns operate on single bijections and are therefore assoiative.

This misses the Lie-algebraic structures of graph composition which are apparent when one studies sums over bijections and hence sums over gluing patterns. 

In particular the Hochschild one-cocycles $B_+^x$ act to povide a sum of gluing patterns based on the pre-Lie algebraic operation of graph insertion as a sum over bijectios underlying DSEs.

An analysis of such sums in terms of assemply operations awaits clarification.
\begin{rem}
In \cite{fourauthors} graphs $\Gamma$ with $|\Gamma|=n$, $l_\Gamma=s$ were collected in such a set $X_{n,s}$. We then have that for any Green function we can write
\[
G^s_{core}=\One\pm \sum_{\Gamma\in S^s_core}\frac{g_\Gamma}{|Aut(\Gamma)|}\Gamma=
\One\pm \sum_{n\geq 1}\underbrace{\sum_{\Gamma\in X_{n,s}}\frac{g_\Gamma}{|Aut(\Gamma)|}\Gamma}_{=:c_{n,s}}. 
\]
The $c_{n,s}$ form sub-Hopf algebras.  See Lem.(\ref{lem Delta G core}) and Lem.(\ref{lemmainvchwithcuts}) for a generalization to Cutkosky graphs. Consequences for the groups $\Gamma_{n,s}$ studied in \cite{fourauthors} await clarification.  
\end{rem}

\subsection{Core, no cuts}
For the case of two external legs there are only two primitive core graphs, as illustrated in Figure~\ref{twoptprim}.  Recalling that the neckaces we use indicate the number of external edges at each vertex, or equivalently the degree minus 2, the necklaces corresponding to these two primitives are $\omega_1 = 2$ and $\omega_2 = 1,1$.
\begin{prop}\label{DSEcorenocutsvert}
With this notation we have for the case of two external legs
\begin{align*}
G_{core}^2& =\One-\frac{1}{2}B_+^{\omega_1}\left(g_4\Pi_{\omega_1}
\right)-\frac{1}{2}B_+^{\omega_2}\left(
g_{3}^2\Pi_{\omega_2}\right) \\
& =\One-\frac{1}{2}B_+^{\omega_1}\left(\frac{g_{4}G_{core}^{4}}{G_{core}^2}
\right)-\frac{1}{2}B_+^{\omega_2}\left(
\frac{g_{3}^2(G_{core}^{3})^2}{(G_{core}^2)^2}\right),
\end{align*}
see Fig.(\ref{twoptprim}),
and $\forall n\geq 3$,
\begin{align*}
  G_{core}^n&=\One+\frac{1}{g_n}\sum_{k=1}^n \sum_{|\omega|=n} \frac{1}{|\mathtt{Aut}(\gamma_\omega)|}B_+^{\gamma_\omega}\left(
  \left(\prod_{v\in V(\gamma_\omega)} g_{\mathbf{val}(v)}\right) \Pi_{\omega}
\right) \\
&=\One+\frac{1}{g_n}\sum_{k=1}^n \sum_{|\omega|=n} \frac{1}{|\mathtt{Aut}(\gamma_\omega)|}B_+^{\gamma_\omega}\left(
\frac{\prod_{v\in V(\gamma_\omega)} g_{\mathbf{val}(v)}G_{core}^{\mathbf{val}(v)}}{(G_{core}^2)^k}
\right),
\end{align*}
where $\gamma_\omega$ is the primitive graph associated to the necklace $\omega$ and the inner sums are over all necklaces of size $n$, that is with the associated graph having $n$ external edges.
\end{prop}
\begin{figure}[H]
\includegraphics[width=8cm]{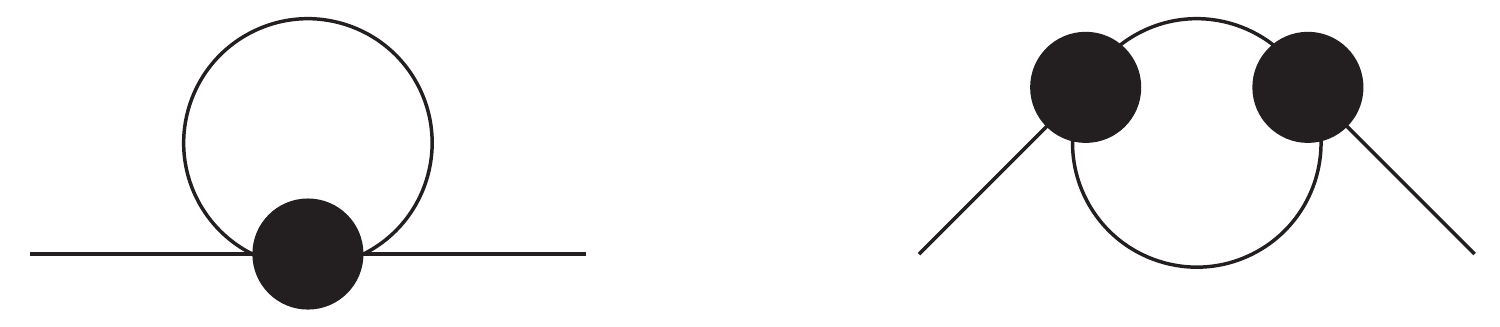}
\caption{The two graphs $B_+^{\omega_1}(\One)$ and $B_+^{\omega_2}(\One)$.
}
\label{twoptprim}
\end{figure}

\begin{cor}
The solution to these fixed point equations is
\[
G_{core}^2=\One-\sum_{\Gamma\in (H_{core})_2} \frac{g_\Gamma \Gamma}{|\mathtt{Aut}(\Gamma)|},
\]
and
\[
G_{core}^n=\One+\sum_{\Gamma\in  (H_{core})_n} \frac{g_\Gamma \Gamma}{|\mathtt{Aut}(\Gamma)|},
\]
where $(H_{core})_n$ is those core graphs with $n$ external edges.
\end{cor}
\begin{proof}
Follows from Lemma~\ref{lemmagraphins} by projecting onto those graphs which have $n$ external edges.
\end{proof}
\subsubsection{Sub-Hopf algebras and invariant charges $Q_i$}
The combinatorial version of the invariant charge is very important here, as it is elsewhere in Dyson-Schwinger analysis.  The most important consequence of the invariant charge is that it lets us write down a formula for the coproduct applied to Green functions, and in fact these formulas show that the failure of the Green function to be grouplike is controlled by the invariant charge.

Formulas of this form have been around since the Hopf algebraic approach to Dyson-Schwinger equations began.  This idea is already implicit in the proof of Theorem 2 
and discussion beforehand in \cite{CK}. Bringing the idea into an explicit formula, in a similar form to what we will give here, though in a more limited context, we have Lemma~4.6 of \cite{kythesis}.  Around the same time is Proposition~7 and Equation~10 of \cite{Walter2}, which is based on Proposition~16 of \cite{Walter3}.

Subsequent formulas along these lines can be found as Theorem 1
 of \cite{BorQED}, for QED,  as Equations 46 and 47 of \cite{HenryQED}, for QCD  as Equation 3.75 of \cite{Henrythesis}, and generalized to super- and non-renorma\-lizable theories as Proposition~4.2 of \cite{Prinzsusy}.

These formulas also have a nice physical interpretation: there is always one power of the invariant charge appearing on the left of the tensor for every power of the coupling that we're taking on the right, so the invariant charge comes with a power of the coupling, and hence it is behaving like a renormalization factor $Z_\alpha$ for the coupling $\alpha$.  Similarly the $n$-point Green function appearing on the right of the tensor is acting like a correponding renormalization factor $Z_n$. 

Let us now give the analogous formula for the core Hopf algebra.  We will give a pre-Cutkosky version in the next subsection.

Given a sequence $\mathbf{k} = (k_1, k_2, \ldots, k_i, \ldots)$, with all but finitely many terms $0$, we'll use a slightly adjusted version of multi-index notation writing $\mathbf{g^k}$ for
\[
\mathbf{g^k} = \prod_{i=1}^\infty g_{i+2}^{k_i}.
\]
We'll also use square bracket coefficient extraction notation.  That is $[\mathbf{g^k}]X$ is defined to be the coefficient of $\mathbf{g^k}$ in the series $X$.

We define the combinatorial invariant charges for the core Hopf algebra as follows.
\begin{defn}
\[
Q_{i+2}:=\frac{G^{i+2}_{core}}{\left[G^2_{core}\right]^{\frac{i+2}{2}}}.
\]
is the $i+2$th combinatorial invariant charge.
\end{defn}

With all this notation we find for the coproduct of the series $G_{core}^n$,
\begin{lem}\label{lem Delta G core}
  \[
  \Delta_{core} \left[\mathbf{g^k}\right]G^n_{core}=\sum_{j_i=0}^{k_i}\left[\mathbf{g^j}\right]\left(G^n_{core}
\prod_{i=1}^\infty Q_{i+2}^{k_i-j_i}\right)\otimes \left[\mathbf{g^{k-j}}\right]G^n_{core},
\]
for any multi-index $\mathbf{k}$, where the subtraction $\mathbf{k-j}$ is coordinatewise. 
\end{lem}
\begin{proof}
  First lets check that the same terms appear on each side and then check that the coefficients of each term match.

  A term on the left is a term in the coproduct of a graph $\Gamma$ with $n$ external edges and with  $k_i$ vertices of degree $i+2$ for $i+2\neq n$ and $k_{n-2}+1$ vertices of degree $n$.  Such a graph term comes from a way of building $\Gamma$ by an insertion.  In such an insertion every vertex of $\Gamma$ appears exactly once in either the subgraph or the co-graph, and additionally the co-graph has a vertex for each external structure among the inserted graphs, including empty insertions.  Since we set up our couplings to include a factor in the denominator for the external structure, the product of the couplings for the subgraph and co-graph give the product of the couplings for $\Gamma$.

  Furthermore, if the co-graph has $k_i-j_i$ vertices of degre $i+2$ for $i+2\neq n$ and $k_{n-2}- j_{n-2}+1$ vertices of degree $n$ then it has that many insertion places of those degrees and so the subgraph must come from the series $G_{core}^n\prod_{i=1}^\infty Q_{i+2}^{k_i-j_i}$.  This also gives the correct number of powers of $(G_{core}^2)^{-1}$ for the edge insertions because each vertex contributes degree many half edges to the graph and the same number of powers of $(G_{core}^2)^{-1/2}$ via its $Q$, except for one degree $n$ vertex which does not contribute powers of $G_{core}^2$.  Every half edge except for the $n$ external half edges is paited with another into an internal edge, and so the power of $(G_{core}^2)^{-1}$ is  exactly corresponding the number of internal edges.

  Thus the terms on the left all appear on the right, and by the same counting, all terms on the right also appear on the left.

  It remains to check that the coefficients agree on the two sides.  It is a standard fact of enumeration that labelled counting via exponential generating series and unlabelled counting weighted by automorphism factors are equivalent (see Lemma 2.14 of \cite{kythesis} for one exposition in a similar language to the present paper).  Working, then, in the labelled case, let us not collect terms with isomorphic graphs, then a each term appears with coefficient $1/|\Gamma|!$.  On the right, the co-graph $\Gamma/\gamma$ appears with coefficient $1/|\Gamma/\gamma|!$ and each inserted subgraph $\gamma_i$ also appears with coefficient $1/|\gamma_i|!$, but there are $\left(|\Gamma/\gamma|!\prod_i |\gamma_i|!\right)/|\Gamma|!$ many relabellings, so the specific term we are looking for appears with coefficient $1/|\Gamma|!$ as desired.  This argument is simply a rewording of the fact that the product of exponential generating series corresponds to the labelled product of the counted objects.

\end{proof}

The result above nicely illustrates why the negative powers in $g_G$ are useful, namely, they mean the couplings behave well with the coproduct on the graphs.

\subsection{Core, with cuts}

For $\omega$ a necklace compatible with a $j$-cut of a set $L$ of $n\geq j$ external edges,
we find similar systems. 

The cut propagator is a special situation, as analytically we can derive that we do not need to consider propagator insertions on either side of the cut in a cut propagator.  This derivation is the main goal of Subsection~\ref{sec propagator} and once we have it in hand, the combinatorics is simpler.

\subsubsection{The (inverse) propagator}\label{sec propagator}

For the inverse propagator which has two external edges the only non-trivial partition is $(1,1)$. We thus have
\be\label{propwithcuts}
G_{pC}^{1,1}=\not\sim-\sum_{\omega\sim \not\sim}B_+^{\omega}\left(\Pi_\omega \right),
\ee
where 
the following eight  cut necklaces $\omega$  are compatible with $\not\sim$:
\begin{figure}[H]
\includegraphics[width=8cm]{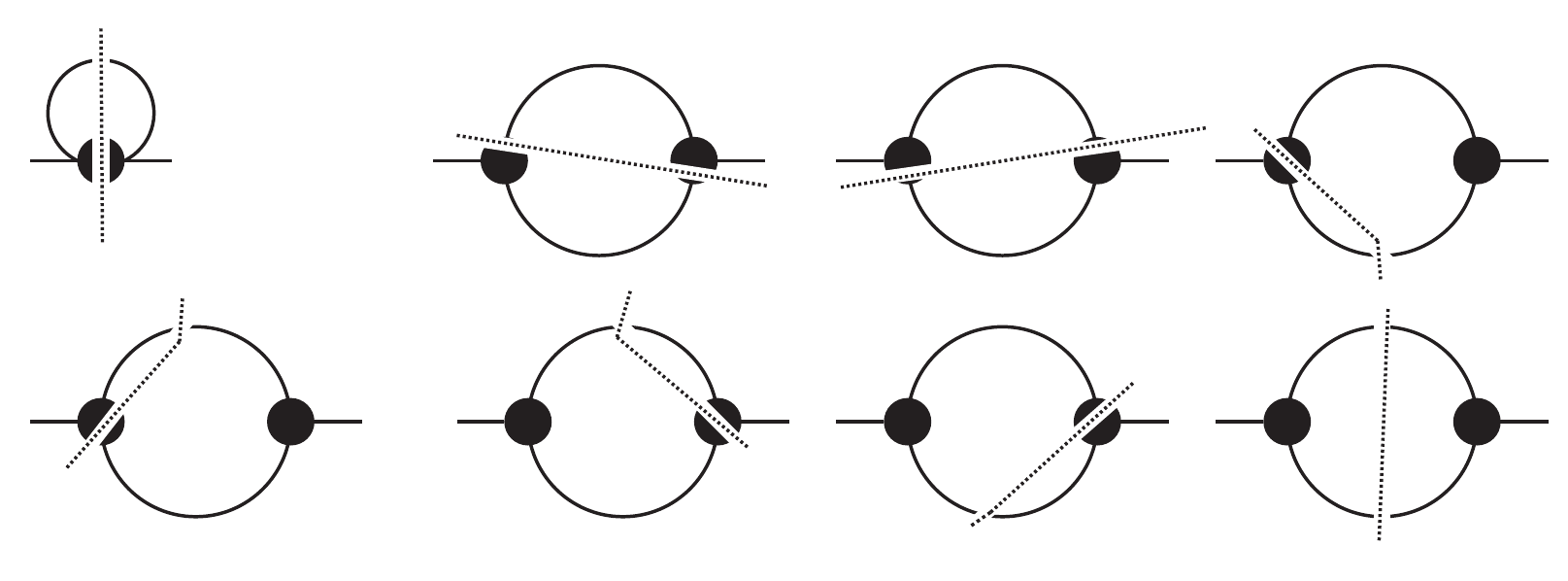}.
\caption{The eight necklaces $\omega$ for the inverse propagator.
}
\label{omegatwocut}
\end{figure}
Here,  $\Pi_\omega$ is given as follows. In the upper row from left to right:
\[
\{G_{pC}^{1,1}G_{pC}^{2,2},(G_{pC}^{1,2})^2(G_{core}^2)^2,(G_{pC}^{2,1})^2(G_{core}^2)^2,G_{pC}^{1,2}G_{pC}^{1,1}G_{core}^3G_{core}^2\}.
\]
In the lower row from left to right:
\[
\{G_{pC}^{2,1}G_{pC}^{1,1}G_{core}^3G_{core}^2,
G_{core}^3G_{pC}^{1,2}G_{pC}^{1,1}G_{core}^2,G_{core}^3G_{pC}^{2,1}G_{pC}^{1,1}G_{core}^2,(G_{core}^{3})^2(G_{pC}^{1,1})^2\}.
\]

In fact, the inverse propagator needs some more detailed care.
Let $\tilde{\sigma}(q^2,m^2)=\Phi_R(X^-)= 1-\bar{\sigma}(q^2,m^2)$ be the self-energy, where we singled out the constant $1=\Phi_R(\One)$, and $G^2_{core}$ is the corresponding combinatorial Green function.

The propagator is
\[
P(q^2,m^2)=\frac{1}{q^2-m^2-\sigma(q^2,m^2)}.
\]
We work with renormalization conditions
\[
\sigma(m^2,m^2)=0,
\]
and 
\[
f(q^2,m^2):=\partial_{q^2}\sigma(q^2,m^2)\rightarrow f(m^2,m^2)=0.
\]
We have $\bar{\sigma}(q^2,m^2)=\sigma(q^2,m^2)/(q^2-m^2)$,
with $\bar{\sigma}(m^2,m^2)=0$,
and then
\[
P=\frac{1}{(q^2-m^2)(1-\bar{\sigma}(q^2,m^2))}.
\]
We have
\[
\Im(P)=\Im\left(\frac{1}{q^2-m^2+i\iota}\right)\Re\left(\frac{1}{1-\bar{\sigma}}\right)+\Re\left(\frac{1}{q^2-m^2+i\iota}\right)\Im\left(\frac{1}{1-\bar{\sigma}}\right)
\]
As $\Im\left(\frac{1}{q^2-m^2+i\iota}\right)=\pi\delta(q^2-m^2)$, we find 
\beas
\Im\left(\frac{1}{q^2-m^2+i\iota}\right)\Re\left(\frac{1}{1-\bar{\sigma}}\right) & = &\\
\pi\delta(q^2-m^2)\Re\left(\frac{1}{1-\bar{\sigma}(q^2,m^2)}\right) & = &\\
\pi\delta(q^2-m^2)\Re\left(\frac{1}{1-\bar{\sigma}(m^2,m^2)}\right) & = &\\
\pi\delta(q^2-m^2) & . & 
\eeas
As $\Re\left(\frac{1}{q^2-m^2+i\iota}\right)=\frac{\mathrm{CP}}{q^2-m^2}$,
we find
\beas
\Re\left(\frac{1}{q^2-m^2+i\iota}\right)\Im\left(\frac{1}{1-\bar{\sigma}}\right) & = &\\
\frac{\mathrm{CP}}{q^2-m^2}\Im(\sigma(q^2,m^2)\frac{\mathrm{CP}}{q^2-m^2}
\Re\left(\frac{1}{\left(1-\bar{\sigma}(q^2,m^2)\right)^2}\right) & , &
\eeas
where we used that 
\[
\Im(\bar{\sigma}^n)=\Im(\bar{\sigma})\Re(\partial_{\bar{\sigma}}\bar{\sigma}^n)
\]
and $\Im(\bar{\sigma})=\Im(\sigma)\frac{\mathrm{CP}}{q^2-m^2}$.

Accordingly, the combinatorial Green function for a cut self-energy starts as 
\[
G^{1,1}_{pC}=\not\sim t_{m} - g^2 B_+^{\omega}\left((G^3_{core})^2\left(\frac{\tilde{G}^{1,1}_{pC}}{(G^2_{core})^2}+\not\sim t_m\right)^2\right)+\cdots,
\]
where $G^3_{core}$ is the combinatorial Green function for the 3-point vertex, $\omega=c1c1$ in the notation of Section~\ref{pC necklaces}, that is, $\omega$ is the cut one-loop bubble (the last in the lower row in Figure~\ref{omegatwocut}), $t_m$ is an indeterminant, one for each distinct mass $m$, where $m$ is the mass of the cut edge, and $\tilde{G}^{1,1}_{pC}=-G^{1,1}_{pC}+\not\sim t_m$.
Furthermore, $\Phi_R(\not\sim)=\pi\delta(q^2-m^2)$.

By putting $t_m$ in the combinatorial Green function for the cut self-energy, we can then extract the graphs with cuts consisting of $a_1$ edges of mass $m_1$ cut, $a_2$ edges of mass $m_2$ cut and so on, by extracting the coefficient of $t_{m_1}^{a_1}t_{m_2}^{a_2}\cdots$ from $G^{1,1}_{pC}$.  The $t_m$ can all be set to $1$ to simplify the equation in the case where we do not want access to this information.  Note that the same coefficient extraction will give the cuts with edges of those masses cut in the vertex functions as well, even though the Dyson-Schwinger equations for the vertex functions have no explicit appearance of any $t_m$.  The recursive appearances of $t_m$ from occurences of $G^{1,1}_{pC}$ suffice because only in the self-energy function is a direct edge-cut possible.

Furthermore, note that we can set 
\[
\not\sim\equiv \not\sim (G^2_{core})^2,
\]
 due to the propagator renormalization conditions $\Phi_R(G^2_{core})(m^2,m^2)=1$. Accordingly, we can set for $\omega=c1c1$,
\[
\left((G^3_{core})^2\left(\frac{\tilde{G}^{1,1}_{pC}}{(G^2_{core})^2}+\not\sim t_m\right)\right)=
\left( G^2_{core} Q^2\left(\tilde{G}^{1,1}_{pC}+\not\sim t_m\right)\right),
\]
with $Q=G^3_{core}/(G^2_{core})^{\frac{3}{2}}$. Other $\omega$ use different $Q$
given as the coefficient of $G_{core}^2\left(\tilde{G}^{1,1}_{pC}+\not\sim t_m\right)$
in the argument of $B_+^\omega$.

We can thus write the coproduct on Green functions 
$X^{P(j)}_{pC}$ by coefficient extraction on such Green functions, on $Q$ and on $\tilde{G}^{1,1}_{pC}$, see Lem.(\ref{lemmainvchwithcuts}).

\subsubsection{$|L_G|\gneq 2$}
We now treat  $|L_G|>2$.

Using necklaces which are pre-cut graphs in the sense of Definition~\ref{precutgraphsdefn} we find 
\begin{prop}\label{DSEcorewithcutsvert}
With this notation we have for the case of $n=2$  external legs
Equation~\ref{propwithcuts} with eight contributing necklaces for $G^{1,1}_{pC}$
and $\forall n\geq 3$, when $p$ has only one part Proposition~\ref{DSEcorenocutsvert} applies, while for $p$ with at least two parts we have
\begin{align*}
  G_{pC}^p&=\frac{1}{g_n}\sum_{k=1}^n \sum_{|\omega|=n,\omega\sim p} \frac{1}{|\mathtt{Aut}(x_\omega)|}B_+^{\gamma_\omega}\left(
  \left(\prod_{v\in V(\gamma_\omega)} g_{p(v)}\right) \Pi_{\omega}
\right) \\
&=\frac{1}{g_n}\sum_{k=1}^n \sum_{|\omega|=n,\omega\sim p} \frac{1}{|\mathtt{Aut}(\gamma_\omega)|}B_+^{\gamma_\omega}\left(
\frac{\prod_{v\in V(\gamma_\omega)} g_{p(v)}G_{pC}^{p(v)}(G_{pC}^{1,1})^{(e_{\gamma_\omega}-e_{\tilde{\gamma}_\omega})}}{(G_{core}^2)^{e_{\tilde{\gamma}_\omega}}}
\right),
\end{align*}
where $\gamma_\omega$ is the primitive graph associated to the necklace $\omega$, $\tilde{\gamma}_\omega$ is this graph with the cuts done (see Section~\ref{sec precut}),
$p(v)$ is the partition of vertex $v$ in $\gamma_\omega$  and the inner sums are over all necklaces of size $n$ compatible with the partition $p$ of $n$ external legs, that is with the associated graph $\gamma_\omega$ having $n$ external edges
and $h_0(\gamma_\omega)=|p|$.
\end{prop}
The solutions to theses fixed-point equations are sums over Cutkosky graphs despite the fact that the necklaces $\omega$ in $B_+^\omega$ can be pre-Cutkosky.

\subsubsection{Sub-Hopf algebras and invariant charges $Q_i$}
Analogously to Lemma~\ref{lem Delta G core}, we can use the invariant charges to understand how to take the coproduct of the series  $G^{p}_{pC}$.  Here the notion of multi-index is generalized so that $\mathbf{k} = (k_{p_1}, k_{p_2}, \ldots)$ where the $k_i$ are indexed by partitions and all but finitely many are $0$, and then $\mathbf{g^k} = \prod_q g_q^{k_q}$ where the product runs over partitions $q$.  Note that the index shift was not built in to the partitioned couplings.

\begin{defn}
Define the cut invariant charges for a partition $q$ to be
\[
Q_{q}:=\frac{G^{q}_{pC}}{\left[G^2_{core}\right]^{\frac{|q|}{2}}}.
\]
if $q\neq (1,1)$ and
\[
Q_{1,1}:=G^{1,1}_{pC}G^2_{core}.
\]

\end{defn}

\begin{lem}\label{lemmainvchwithcuts}
  \[
  \Delta_{pC} \left[\mathbf{g^k}\right]G^p_{pC}=\sum_{j_q=0}^{k_q}\left[\mathbf{g^j}\right]\left(G^p_{pC}
\prod_{q} Q_{q}^{k_q-j_q}\right)\otimes \left[\mathbf{g^{k-j}}\right]G^p_{pC},
\]
for any partiton $p$ and multi-index $\mathbf{k}$, where the subtraction $\mathbf{k-j}$ is coordinatewise and the product is over partitions $q$. 
\end{lem}

\begin{proof}
The proof is analogous to the proof of Lemma~\ref{lem Delta G core} with the addion of the need to distinguish between cut and uncut propagators.  The cut propagators are counted by the powers of $g_{1,1}$ and $Q_{1,1}$ is designed to substitute a $G^{1,1}_{pC}$ for an inverse core propagator Green function each time a cut propagator appears in the co-graph.  This gives the correct insertions.
\end{proof}

\begin{rem}\label{pivsQ}
  Recall $\Pi_\omega$ from Section~\ref{DSE}.  Both $\Pi_\omega$ and the $Q_i$ are products of combinatorial Green functions and inverses of combinatorial Green functions, and they are closely related.
  The $\Pi_\omega$ are associated to a necklace, while the $Q_i$ are essentially associated to a corolla.
  
  We can express $\Pi_\omega$ through invariant charges $Q_{v_i}$, $v_i\in V_{x_\omega}$ and inverse propagator functions $G_{core}^2,G_{pC}^{1,1}$, simply by multiplying around the necklace. Specifically, if a necklace $\omega$
contributes to $G_{pC}^p$ with $p$ a partition of, say, $n$ external legs (we include the core case $p=(n)$) then 
\[
\Pi_\omega=\prod_{i=1}^{v_g} G_{pC}^p \frac{Q_{v_i}}{Q_p}
(Q_{1,1})^{|C_g|},
\]  
with $Q_{v_i}=G_{v_i}/(G_{core}^2)^{{\bf{eval}}(v_i)/2}$, $Q_p=G_{pC}^p/(G_{core}^2)^{n/2}$ and $Q_{1,1}=G_{pC}^{1,1}G_{core}^2$,
which replaces each coupling $g_i$ apparent in $g_{x_\omega}$ by the corresponding 
Green function. 

The relation between necklace and corolla is also a special case of the relationship between a graph and its planar dual, so in some sense the $\Pi$ and the $Q$ are dual to each other.
\end{rem}

\subsubsection{Example}

{\allowdisplaybreaks
\begin{align*}
  G^{1,2}_{pC} & = g^2B_+^{\includegraphics{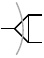}}\left(\frac{(G^{1,1}_{pC})^2 (G^3_{core})^3}{G^2_{core}}\right) + g^2B_+^{\includegraphics{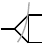}+\includegraphics{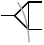}+\includegraphics{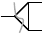}+\includegraphics{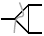}}\left(\frac{G^{1,2}_{pC}G^{1,1}_{pC}(G^3_{core})^2}{(G^2_{core})^2}\right) \\
 & \qquad + g^2B_+^{\includegraphics{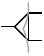}+\includegraphics{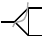}+\includegraphics{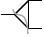}}\left(\frac{(G^{1,2}_{pC})^2G^3_{core}}{(G^2_{core})^3}\right)\\
 & \qquad + g^2B_+^{\includegraphics{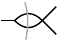}}\left(G^4_{core}G^3_{core}(G^{1,1}_{pC})^2\right) + g^2B_+^{\includegraphics{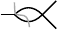}+\includegraphics{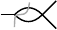}}\left(\frac{G^4_{core}G^{1,2}_{pC}G^{1,1}_{pC}}{G^2_{core}}\right) \\
  & \qquad + g^2B_+^{\includegraphics{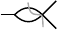}+\includegraphics{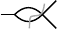}}\left(\frac{G^{1,3}_{pC}G^3_{core}G^{1,1}_{pC}}{G^2_{core}}\right) + g^2B_+^{\includegraphics{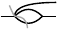}+\includegraphics{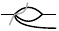}}\left(\frac{G^{2,2}_{pC}G^3_{core}G^{1,1}_{pC}}{G^2_{core}}\right) \\
  & \qquad + g^2B_+^{\includegraphics{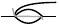}+\includegraphics{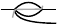}}\left(\frac{G^{2,2}_{pC}G^{1,2}_{pC}}{(G^2_{core})^2}\right)  + g^2B_+^{\includegraphics{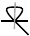}}\left(\frac{G^{2,3}_{pC}}{G^2_{core}}\right) + \cdots
\end{align*}
}
It is worth emphasizing that, as one would see from expanding out the example, while these equations live in the pre Cutkosky universe, the series when expanded out only include Cutkosky graphs, since Green functions for the cut vertices have no constant term, and so for any given graph, with all blobs substituted, the only remaining cuts are edge cuts in accordance with Lemma~\ref{lemmagraphins} and Equation~\ref{propwithcuts}, the latter equation having a constant term $\not\sim$. 

\subsection{The coaction on Green functions}\label{coactGreen}
The coaction $\rho$ from Section~\ref{cointbi} and the related Galois coation (see Section~\ref{gfcoaction}) can be expressed at the level of Green functions.  The nicest way to see this is to put each Green function into a matrix.  The matrix also nicely illustrates how the Green function can be refined by number of vertices and number of cut edges as well as by loops, and connects to other work of one of us \cite{BlochKreimerOS,coaction,MarkoDirk}.

To begin with, we will describe how to build a lower triangular matrix $M_\Gamma$ from the Galois conjugates of a core graph $\Gamma$.
Suppose $\Gamma$ has $n$ external edges.  We will view $\Gamma$ as having additional external edges of momentum $0$ at all vertices which do not already have an external edge, so that we can cut all non-self-loop edges of $\Gamma$ without resulting in non-physical cuts.
The entries of $M_\Gamma$ are all pre-Cutkosky graphs, though we will, without further comment, identify a core graph $\Gamma'$ with the pre-Cutkosky graph $(\Gamma',\Gamma')$

The first column of $M_\Gamma$ is built from the set of all graphs which can be obtained from $M_\Gamma$ by contracting edges of a spanning forest of $\Gamma$.\footnote{We will consider two graphs obtained in this way to be the same if they are isomorphic where we take the external edges to be labelled but the internal edges to be unlabelled.  Equivalently, since we have put extra external edges so that each vertex has at least one external edge, we can consider two graphs obtained in this way to be the same if they are isomorphic when we take the vertices, but not the edges, to be labelled.  In particular, there is only one such graph on one vertex.}  The column has one such graph in each entry.  The ordering is given by the number of vertices, and within graphs with the same number of vertices an aribtrary order is chosen.  So, the top left entry of $M_\Gamma$ is the graph with one vertex, as many loops as $\Gamma$ and as many external edges as $\Gamma$.  The entry below this has two vertices, as may some further entries below. Next come three vertex graphs, and finally the bottom left entry of $M_\Gamma$ is $\Gamma$ istelf.  

The main diagonal of $M_\Gamma$ is built from the same graphs as the first column, in the same order, except that all non-loop edges are cut.  The top left entry is common to the first column and the main diagonal, but fortunately, it has no non-loop edges, so this is consistent.

For $1<j<i$, the $(i,j)$th entry of $M_\Gamma$ is the pre-Cutkosky graphs whose underlying graph is the $(i,1)$ entry of $M_\Gamma$ and whose associated graph with any compatible forest contracted is the $(j,j)$ entry of $M_\Gamma$, if such a graph exists, and is $0$ otherwise.  Observe that if such a graph exists then it is unique because the external edges are taken as distinguishable, so we know which vertices of $\Gamma$ have been combined in the $(i,1)$ and $(j,j)$ entries and either these combinations are incompatible, or they are compatible and we know exactly which edges to cut to get the $(i,j)$th entry from the underlying graph.

Note that $M_\Gamma$ is built of Galois conjugates of $\Gamma$.

\medskip

The next step is that we want to upgrade $M_\Gamma$ to include all graphs of the Green function $G^{(n)}_{core}$.  We will call this new matrix $M$.  $M$ is an infinite lower triangular matrix whose entries are pre-Cutkosky graphs with $n$ external edges with nonzero momenta.

The first column of $M$ consists of all graphs appearing in $G^{(n)}_{core}$ with their symmetry factors as coefficients and ordered first by loop number, then within a loop order, ordered by number of vertices.  The main diagonal of $M$ consists of these same graphs in the same order with all edges cut.

Analytically the main diagonal is obtained from the first column by replacing what would be obtained from the Feynman rules by its leading singularity.

For $1<j<i$, if the $(i,1)$ and $(j,j)$ entry appear together in some $M_\Gamma$, then the $(i,j)$th entry is the sum of the corresponding entries from all such $M_G$, each scaled by their symmetry factors, otherwise the $(i,j)$th entry is $0$.\footnote{The reason that non-trivial sums are possible in general is that we do not have a universal labelling of vertices, only of the original $n$ external edgs, and so different contractions of the $(i,1)$th entry can potentially give the underlying graph of the $(j,j)$th entry.  This could be avoided by, as in $M_\Gamma$ adding labelled external edges of momentum $0$ to all vertices of all graphs, at the cost of requiring extra copies of many graphs, one for each different way the graph can be obtained as a contraction of a larger one.  The number of extra copies gets quite large as with many vertices we need many external edges and then we have the copies of the graphs with fewer vertices obtained by contracting these larger ones.}

See Figure~\ref{finalexample}.

\begin{figure}[H]
\includegraphics[width=12cm]{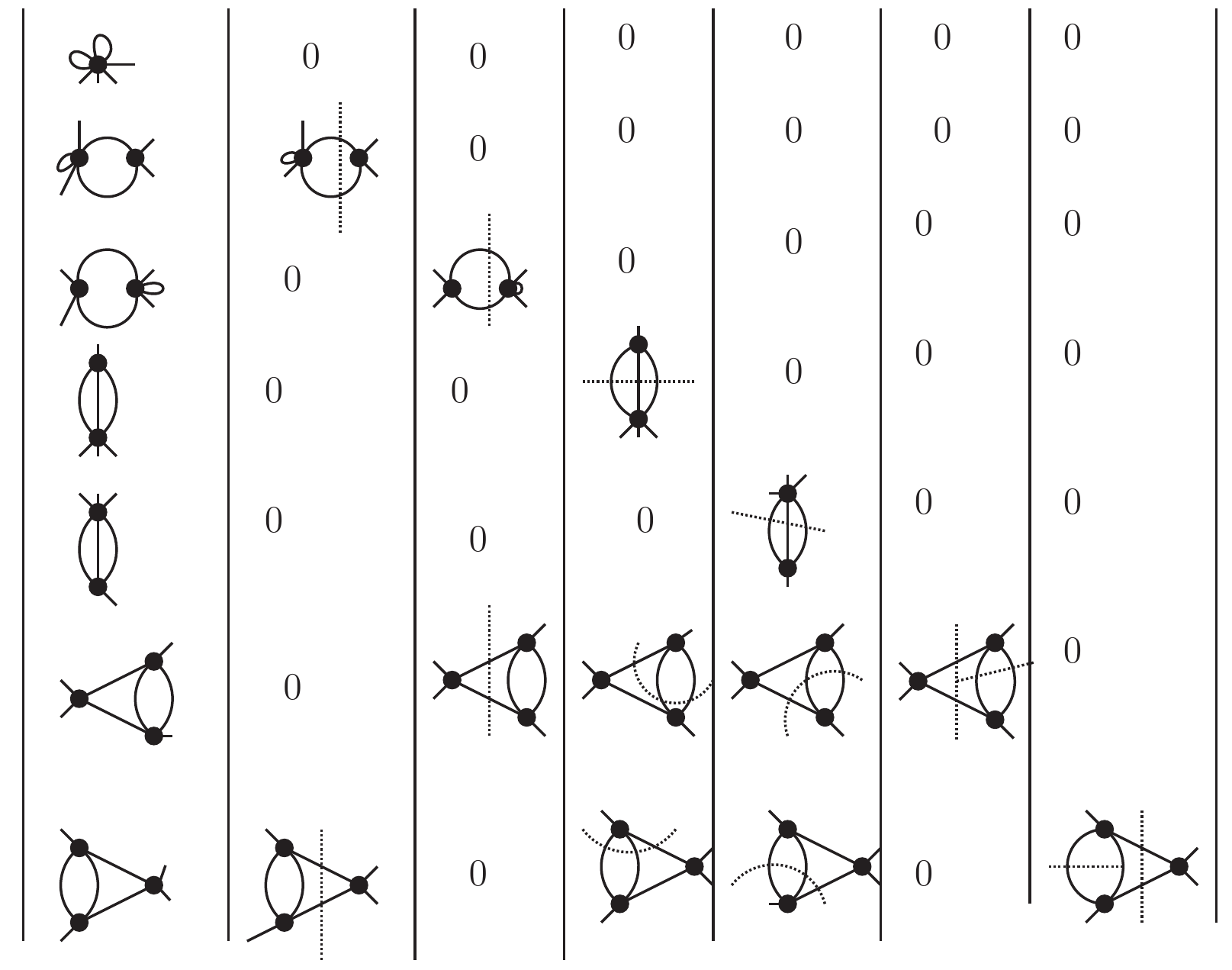}.
\caption{A $7\times 7$ example of a matrix $M$. Note that we can regard any graph $\Gamma$ in this matrix as representing a sum $\sum_{F\sim E_{on}} (\Gamma,F)$ over pairs $(\Gamma,F)$ with $F$ compatible
with the set $E_{on}$ indicating the chosen refinement of $L_\Gamma$ in accordance with Lem.(\ref{lemmafrakC}). Note that we omitted symmetry factors $|\mathtt{Aut}(\gamma)|$ for any graphs $\gamma$ in this matrix for brevity and division by this symmetry factor is understood for each graph consistent with the appearance of such factors in 
combinatorial Green functions.}
\label{finalexample}
\end{figure}

\begin{rem}
  Note that in Figure~\ref{finalexample} recuperating all the matrix entries from the diagonal entries amounts to full cut reconstructibility. This is feasible upon moving to the left by iterating dispersion integrals and moving up by contracting edges $e$, amounting to an operator product expansion for the merger of the two vertices given by source $s(e)$ and target $t(e)$, see the discussion in \cite{MarkoDirk}. It is an underlying motivation of this paper that the structure of these matrices results from an underlying cubical chain complex \cite{rational}. It systematizes the study of cut reconstructibility in a way which generalizes studies which  at first loop order were undertaken by Britto and collaborators \cite{Britto}.
\end{rem}

Next, the cointeracting bialgebra structure discussed in Section~\ref{cointbi}  acts on the matrix $M$.

Consider Figure~\ref{finalexample}. In the leftmost column we list some two-loop graphs of $G^{(4)}_{core}$.
The uppermost entry has two loops but only one vertex.  The next four entries have two loop and two vertices; they count as different when the external edges are partitoned differently.  The final two entries are the two core graphs with two loops and three vertices.  The leftmost column also determines the diagonal by putting every edges on-shell.

Using our results in Section~\ref{cointbi} we finally have
\begin{thm}\label{prop matrix coaction}
For $j=1$ the map 
\[
\rho_M: M_{ij}\to \sum_{k=1}^\infty \Phi_R(M_{kj})\otimes m_{\mathbb{C}}(S_R^\Phi\otimes\mathrm{id})\bar{\Delta}(M_{ik})
\]
is a map $\rho_M: \mathsf{Gal}_{(G,T)}\to \mathsf{Gal}^{\mathfrak{m}}
\otimes \mathsf{Gal}^{\mathfrak{dr}}$  which fulfills 
\[
\rho_M=\sum_{T\in \mathcal{T}}\rho_\Phi\circ w,
\]
based on the coaction $\rho$ of Appendix~\ref{appB}.\\
For $M_{ij}$, $j\geq 2$, $\rho_M=\rho_\Phi\circ w$ is a map  $\mathsf{Gal}_{(G,T)}\to \mathsf{Gal}^{\mathfrak{dr}}\otimes \mathsf{Gal}^{\mathfrak{dr}}$ based on the coproduct that is given by $\rho$ when there is no tadpole restriction.
\end{thm}
\begin{rem}
Note that the columns $C_j$, $(C_j)_i=(M)_{i,j}$ can be generated by
Proposition~\ref{DSEcorewithcutsvert}. In particular also the diagonal $D_i=(M)_{i,i}$
can be generated by setting all core Green functions to $\One$ in all invariant charges. See also Remark~\ref{corevspc}.
\end{rem}

\subsection{Pairs $(\Gamma,F)$ or $H_C$}
The matrix $M$ and Proposition~\ref{prop matrix coaction} both link together the Green functions and their Dyson--Schwinger equations with the Galois conjugates.  The Dyson--Schwinger equations are defined in terms of pre-Cutkosky graphs while the latter are defined in terms of pairs $(\Gamma,F)$.  This naturally leads into the question of Dyson-Schwinger eqations for pairs $(\Gamma,F)$.

However, when we consider cuts via pairs $(\Gamma,F)$, the situation with respect to Dyson--Schwinger equations is quite different.  The valid co-graphs in this case must still be $(\Gamma,F)$ pairs, and so we cannot contract a subgraph which would yield a cut vertex.  We cannot even contract a cut propagator subgraph, as in the co-graph this would join two branches of the forest that were separate in the original graph, resulting either in a cycle or a forest with fewer trees, in either case not a valid co-graph.

This means that there are many more primitive graphs but what is inserted into them is simpler -- only uncut graphs.  The primitive graphs in $H_{GF}$ are of two forms: core primitives, that is one loop uncut graphs, and cut graphs where no loops are left intact by the cut, the graphs $\Gamma_j$ in Equation~\ref{primitivesgpc} below.  These are exactly the primitives as cut subgraphs are not allowed in the coproduct and a cut graph with an uncut loop has that uncut loop as a subdivergence. 

Dyson--Schwinger equations can still be formed and they generate the same Green functions since the graphs which ultimately appear in the $G_{pC}$ are all Cutkosky graphs.  It amounts to the same thing to write Dyson--Schwinger equations in the vector space $H_C$.  Whether in $H_C$ or $H_{GF}$, since there are no vertex cuts, the Dyson--Schwinger equations insert core graphs into core primitives and into graphs with no loop left intact.  The only difference between the formulation in $H_C$ compared to $H_{GF}$ is whether or not graphs are split into sums over compatible forests.

In $H_C$ we can set up the Dyson--Schwinger equations as follows and the $H_{GF}$ case is analogous.  For any partition $p$ of $L_\Gamma$ we can re-write the corresponding Green function 
$G^p_{pC}$ as generated from $1$-cocycles $B_+^\gamma$, with $\gamma\in H_C$ and $||\gamma||=0$,
$|\gamma|\geq 1$, $\gamma\sim p$. With $P_0$ the projection $H_C\to H_C^{||0||}$, consider
$G_{pC}^{p,0}=P_0(G_{pC}^p)$.  

We write
\[
G_{pC}^{p,0}=\sum_j c_j \Gamma_j, \Gamma_j\in H_C^{||0||},
\]
where $c_j$ contains $g_{\Gamma_j}$ and symmetry factors. Then for each $\Gamma_j$ define 
\[
\Pi_j^p:=\frac{\prod_{v\in V_{\Gamma_j}}G_{core}^{\mathbf{val}(v)}}{\prod_{e\in E_{\tilde{\Gamma_j}}}G_{core}^2}.
\]
We have 
\be\label{primitivesgpc}
G_{pC}^{p}=\sum_j c_j B_+^{\Gamma_j}(\Pi_j^p),
\ee
by construction. For any non-trivial partition $p$, $\Gamma_j$ has no loop left intact.

This way of writing the Dyson--Schwinger equations emphasizes different aspects than the Dyson--Schwinger equations from elsewhere in this paper.  The $H_C$ or $H_{GF}$ perspective generates the same Green functions as the $H_{pC}$ approach, but builds them by inserting core Green functions into graphs with no loop left intact.  The $H_C$ or $H_{GF}$ perspective does not capture the structure of the cuts in a rich way.  The insertions are core graphs and all the cut structure is in the primitives, so these Dyson--Schwinger equations don't really build the cuts, they only work around them.

The $H_C$ or $H_{GF}$ perspective does bring up the question of characterizing the Cutkosky graphs with no loop left intact.  
Such a graph after cutting falls into a collection of cycle-free graphs, that is, it is a forest.  The external edges of the graphs act like roots in this forest, but note that there may be multiple roots in one tree of the forest.  The case where one or more trees of the forest has no external edges will not add to monodromy apparent from varying external momenta but can contribute when we vary internal masses or are interested in variations off the principal sheets.

This says that the non-core primitives in $H_{GF}$ are multiply rooted forests along with gluing information giving a pairing of the leaves of the forest so that no two leaves of the same tree are paired.

Turning this around we might consider fixing a forest and looking at different ways to glue it according to this rule.  One interesting thing to consider is when the gluing gives a graph which is primitive in some renormalization Hopf algebra.  The simplest case of this will be studied in \cite{DLY}.

\section{Conclusions}\label{sec concl}
There are conclusions to be drawn with regards to combinatorial Dyson--Schwinger equations, coactions and the cointeraction.
\subsection{Combinatorial Dyson--Schwinger equations}
Combinatorial Dyson--Schwinger equations are fixed-point equations which originate from the Hopf algebra structure underlying core or pre-Cutkosky Hopf algebras.

As their solutions they generate infinite series over core or Cutkosky graphs \cite{Karenbook}.
Applying Feynman rules to such equations delivers systems of integral equations 
which generalize the well-studied Dyson--Schwinger equations for Green functions in renormalizable quantum field theories. 

One finds asymptotic series
in coupling parameters as formal solutions to these equations obtained from the recursive nature of perturbation theory \cite{Michibook}. Sophisticated methods of resummation exists and are recently complemented by methods of resurgence. One finds that the non-perturbative contributions which is missing in these asymptotic series can be systematically generated from the asymptotic series themselves \cite{DunneB}.

We hope that this approach can in the future also be used here where the systems of 
combinatorial Dyson--Schwinger equations given in
Proposition~\ref{DSEcorenocutsvert} and Proposition~\ref{DSEcorewithcutsvert} are more involved. In particular when fixing the number of cut edges for a Green functions (computing the amplitude for a $k$-particle cut with fixed $k$) our approach allows using a Fubini type decomposition into a cut graph augmented (dressed) by proper renormalized Green functions corresponding to full propagators and vertices
as in Eq.\eqref{primitivesgpc}. 

Finally given that Cutkosky cuts are intimately related to an understanding of the infrared sector of the theory \cite{Sterman} we hope that our results on cut Dyson-Schwinger equations open an avenue to make progress in this direction as well
as indicated in Rem.(\ref{corevspc}). In particular the recent treatment of infrared singularities \cite{CapIR} in the context of loop tree duality \cite{Cap1,Cap2,Cap3} indicates an interplay of infrared and short distance singularities akin to the cointeraction studied in this work.  
\subsection{Coactions}
There are two bialgebras and accompanying coactions whose use in  physics amplitudes we illuminated.

i) The first $\bar{\Delta}_{core}$ in Section~\ref{hcorehpc} considers the treatment of loop integrals for loops build from  offshell edges which appear in a cut Feynman graph. This coaction is needed in the above Fubini type decomposition of any Cutkosky graph $G$ into a cut bare diagram $y$ which has no loop left intact and loop graphs $x$ providing radiative corrections at internal vertices and at internal edges of $y$ as in Remark~\ref{remclustersep}.

ii) The second concerns the decomposition of a Feynman graph or amplitude according to its Hodge structure. The incidence bialgebra with its coproduct $\rho$ on proper Cutkosky graphs and coaction $\rho$ on core graphs delivers the Hodge decomposition 
captured by the coaction $\rho$ when acting on Green functions as in  Section~\ref{coactGreen}. This delivers a solid mathematical foundation to an analytic understanding of amplitudes and in particular to the notion of coaction as promoted in one-loop examples in \cite{Britto}. Furthermore it is the starting point for a study of amplitudes as realizations of cubical chain complexes \cite{BlochKreimerOS}. We pointed out relations to the notion of assembly maps in Section~\ref{assemblysec}. First consequences for Green functions were developed in \cite{MarkoDirk}.
\subsection{Cointeracting bialgebras}
We clarified in this paper how these two coactions i) and ii) above interact. Fortunately recent progress by Lo\"ic Foissy \cite{cointpapers} and others \cite{Manchon} provided us with the notion of cointeracting bialgebras. It turns out that the two coaction provide exactly that: two cointeracting bialgebras. We quote \eqref{cointeq}:
\[
m_{1,3,24}\circ(\rho\otimes\rho)\circ \Delta_c=(\Delta_c\otimes\mathrm{id})\circ\rho,
\]
which allows us to either first decompose a graph or amplitude according to its Hodge structure and then take care of the loop integrals, or vice versa first to renormalize the amplitude and then determine the Hodge structure in sub- and co-loops,
leading to the map \eqref{rhoaction}
\[
\rho_\Phi:=(m_{\mathbb{C}}\otimes_{\mathbb{Q}}\mathrm{id})(S_R^\Phi\otimes\Phi\otimes {\bar{\Phi}})(\circ w^{-1})^{\otimes 3}\circ(\Delta_c\otimes \mathrm{id})\circ \rho= (\Phi_R\otimes_{\mathbb{Q}} {\bar{\Phi}})\circ (w^{-1})^{\otimes 2}\circ\rho,
\]
involving both $\rho$ and $\bar{\Delta}_{core}$ in $\bar{\Phi}=m_\mathbb{C}(S_R^\Phi\otimes\Phi)\circ\bar{\Delta}_{core}$, and therefore illuminating the Hodge structure of Cutkosky cuts and renormalization in one cointeraction.
\appendix

\section{Graph set up}\label{graph appendix}

Here we give a formal set up for graphs and related notions as a rigorous underpinning for all that has been discussed above.

First, we have to define graphs, and variations of that notion adopted to the need to study Cutkosky cuts. 
In particular for a given partition of the external edges $L$ into $k$ parts we 
define a corresponding notion of graph with cuts corresponding to the parititon generalizing the notion of a Feynman graph.  Physically, the cut edges are edges which are put on-shell.

For this we have to single out the internal edges which are put on-shell. We notate such a graph as a pair of two graphs, the first given by forgetting the distinction of those on-shell edges, the second by regarding those on-shell edges as two unconnected half-edges which gives us an appropriate combinatorial handle on
the analytic structure of amplitudes. 

Alternately, we can encode a graph with cuts using a graph and a spanning forest where the cut edges are those going between different trees of the forest.
The map from a pair $(\Gamma,F)$ of a graph $\Gamma$ and a spanning forest $F$ to a pair 
of two graphs $(\Gamma,H)$ indicating the same cut is surjective but not injective: there can be several forests putting the same set of edges on-shell.

\subsection{Graphs and Cut Graphs}\label{set graphs}
We start by defining graphs in a way which is well suited to our needs.  See \cite{Michibook, kythesis} for similar quantum field theory inspired set-ups of graphs.  The standard definition of combinatorial map, see \cite{LZgraphs}, is also closely related.
\subsubsection{Graphs}\label{sec graphs}
Given a set $S$ a \emph{partition} (or \emph{set partition}) $\mathcal{P}$ of $S$ is a decomposition of $S$ into disjoint nonempty subsets whose union is $S$.  The subsets forming this decomposition are the \emph{parts} of $\mathcal{P}$.  The parts of a partition are unordered, but it is often convenient to write a partition with $k$ parts as $\dot{\cup}_{i=1}^k S_i = S$ with the understanding that permuting the $S_i$ still gives the same partition.  A partition $\mathcal{P}$ with $k$ parts is called a $k$-partition and we write $k=|\mathcal{P}|$.
We will also need notation for restricting a partition to a subset.  For $R\subseteq S$, and $P=\dot\cup_i S_i$ a partition of $S$, write $P|_R$ for the partition of $R$ whose parts are the nonempty $S_i\cap R$. 

\begin{defn}
A \emph{graph} $\Gamma$ is a tuple $\Gamma=(H_\Gamma, \mathcal{V}_\Gamma, \mathcal{E}_\Gamma)$ consisting of 
\begin{itemize}
\item $H_\Gamma$, the set of half-edges of $\Gamma$,
\item $\mathcal{V}_\Gamma$, a partition of $H_\Gamma$ with parts of cardinality at least 3 giving the vertices of $\Gamma$,
\item $\mathcal{E}_\Gamma$, a partition of $H_\Gamma$ with parts of cardinality at most 2 giving the edges of $\Gamma$.
\end{itemize}
\end{defn}
We do not require all parts of $\mathcal{E}_\Gamma$ to be of  cardinality 2.
The parts of cardinality 2 we call the edges (or internal edges) of $\Gamma$; we denote the set of them by $E_\Gamma$ and set $e_\Gamma:=|E_\Gamma|$.  The parts of cardinality 1 we call the external edges of $\Gamma$; we denote the set of them $L_\Gamma$ and set $l_\Gamma:=|L_\Gamma|$. Also we set $v_\Gamma:=|\mathcal{V}_\Gamma|$.

In this formulation a graph $H=(H_H, \mathcal{V}_H, \mathcal{E}_H)$ is a subgraph of $\Gamma=(H_\Gamma, \mathcal{V}_\Gamma, \mathcal{E}_\Gamma)$ if $H_H\subseteq H_\Gamma$, $\mathcal{V}_H = \mathcal{V}_\Gamma|_{H_H}$ and $E_H\subseteq E_\Gamma$. Many of our subgraphs will have the additional property that every part of $\mathcal{V}_H$ is a part of $\mathcal{V}_\Gamma$ (so corollas are either included or not in the subgraph), but this is not included in the definition because, while it is the correct condition for subFeynman diagrams, it is not a desired condition for spanning trees and forests which we also view as subgraphs.

We say that a graph $\Gamma$ is connected if there is no partition of $H_\Gamma$ into two sets $H_\Gamma(1),H_\Gamma(2)$ such that all parts of cardinality two of $\mathcal{E}_\Gamma$
are either in $H_\Gamma(1)$ or in $H_\Gamma(2)$.

The partition $\mathcal{V}_\Gamma$ collects half-edges of $\Gamma$ into vertices.  This formulation of graphs does not distinguish between a vertex and the corolla of half-edges giving that vertex.  However, it is sometime useful to have notation to distinguish when one should think of vertices as vertices and when one should think of them as corollas.  Consequently, let $V_\Gamma$, the set of vertices of $\Gamma$, be a set in bijection with the parts of $\mathcal{V}_\Gamma$, $|V_\Gamma|=v_\Gamma=|\mathcal{V}_\Gamma|$.  This bijection can be extended to a map $\nu_\Gamma:H_\Gamma\rightarrow V_\Gamma$ by taking each half edge to the vertex corresponding to the part of $\mathcal{V}_\Gamma$ containing that vertex.
For $v\in V_\Gamma$ define 
\[
\mathbf{c}_v:=\nu_\Gamma^{-1}(v)\subset H_\Gamma,
\]
to be the corolla at $v$, that is the part of $\mathcal{V}_\Gamma$ corresponding to $v$. 
A graph $\Gamma$ as above can be regarded as a set of corollas determined by $\mathcal{V}_\Gamma$ glued together according to $\mathcal{E}_\Gamma$.\footnote{To make the connection with combinatorial maps, since all parts of $\mathcal{E}_\Gamma$ are of cardinality at most 2, $\mathcal{E}_\Gamma$ uniquely determines an involution of $H_\Gamma$ which takes a half edge to itself if it is alone in its part and to its part-mate otherwise.  In combinatorial maps this involution is usually called $\alpha$ for \emph{ar\^ete}, French for edge.  Combinatorial maps differ from graphs in that there is a cyclic order for the half-edges around each vertex, thus the partition $\mathcal{V}_\Gamma$ is upgraded to a permutation of $H_\Gamma$ with one cycle for each vertex.  This permutation is usually called $\sigma$ for \emph{sommet}, French for vertex.}

For an edge $e\in E_\Gamma$, if $|\nu_\Gamma(e)|=1$, we say $e$ is a self-loop at $v$, with $\nu_\Gamma(e)=\{v\}$.

We emphasize that we allow multiple edges between vertices and allow self-loops as well. 

We write $|\Gamma|:=|H^1(\Gamma)|=e_\Gamma-v_\Gamma+1$ for the number of independent loops, or the dimension of the cycle space of the connected  graph $\Gamma$.
Note that for disjoint unions of graphs $H_1,H_2$, we have $|H_1\dot{\cup} H_2|=|H_1|+|H_2|$.

A graph is bridgeless if $(\Gamma-e)$ has the same number of connected components as $\Gamma$ for any $e\in E_\Gamma$.  A graph is 1PI or 2-edge-connected if it is both bridgeless and connected, equivalently if $(\Gamma-e)$ is connected for any $e\in E_\Gamma$. 
Here,
for $\Gamma=(H_\Gamma, \mathcal{V}_\Gamma, \mathcal{E}_\Gamma)$, we define
\[
(\Gamma-e):= (H_\Gamma, \mathcal{V}_\Gamma, \mathcal{E}'_\Gamma)
\]
where $\mathcal{E}'_\Gamma$ is the partition which is the same as $\mathcal{E}_\Gamma$ except that the part corresponding to $e$ is split into two parts of size $1$.

The removal $\Gamma-X$ of edges forming a subgraph $X\subset \Gamma$ is defined similarly 
by splitting the parts of $\mathcal{E}_\Gamma$ corresponding to edges of $X$.
$\Gamma-X$ can contain isolated corollas.

Note that this definition is different from graph theoretic edge deletion as all the half-edges of the graph remain and the corollas are unchanged.  
We neither lose vertices nor half-edges when removing an internal edge. We just unglue
the two corollas connected by that edge, or to put it another way, the edge is split into two external edges by separating the two half edges forming it.

The graph resulting from the contraction of edge $e$, denoted $\Gamma/e$ for $e\in E_\Gamma$, is defined to be
\be\label{edgecont}
\Gamma/e = (H_\Gamma-e, \mathcal{V}'_\Gamma, \mathcal{E}_\Gamma-e)
\ee
where $\mathcal{V}'_\Gamma$ is the partition which is the same as $\mathcal{V}_\Gamma$ except that in place of the parts $\mathbf{c}_v$ and $\mathbf{c}_w$ for $e=\{v,w\}$, $\mathcal{V}'$ has a single part $(\mathbf{c}_v\cup \mathbf{c}_w) - e$.\footnote{We often use $-$ for the set difference, e.g.\ $H_\Gamma-e=H_\Gamma\setminus e$.}
Likewise we define $\Gamma/X$, for $X\subseteq \Gamma$ a (not necessarily connected) graph, to be the graph obtained from $\Gamma$ by contracting all internal edges of $X\subseteq \Gamma$.

Intuitively we can think of $\Gamma/X$ as the graph resulting by shrinking all internal edges of $X$ to zero length:
\be\label{lengthcont}
\Gamma/X=\Gamma|_{\mathrm{length}(e)=0,e\in E_X}.
\ee
This intuitive definition can be made into a precise definition if we add the notion of edge lengths to our graphs, but doing so is not to the point at present.

We let $\mathbf{val}(v):=|\mathbf{c}_v|$ the degree or valence of $v$ and $\mathbf{eval}(v):=|L_v|$
the number of external edges at $v$.

Our interest lies in cutting bridgeless graphs into disconnected pieces
allowing for arbitrary partitions of the set $L_\Gamma$. This is achieved by cuts
which partition $L_\Gamma$ by either removing edges from $\Gamma$ or by partitioning corollas $\mathbf{c}_v$.
\begin{figure}[H]
\includegraphics[width=6cm]{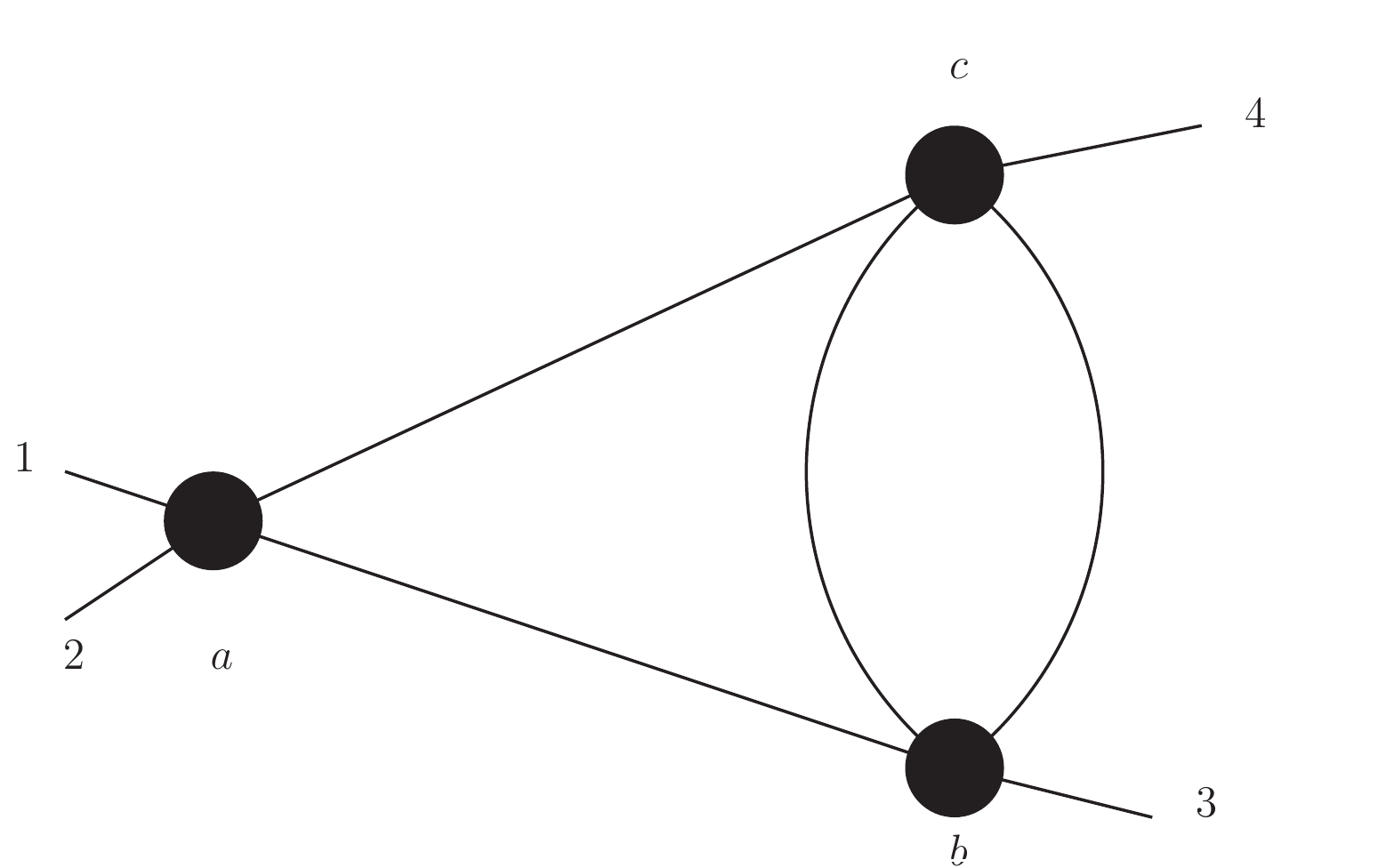}.
\caption{A graph. The external edges labeled $1,2,3,4$ provide four half-edges $\{l_1,l_2,l_3,l_4\}= L$. The vertices $a,b,c$ are all four-valent with
$\mathbf{eval}(a)=2,\mathbf{eval}(b)=1=\mathbf{eval}(c)$. Between vertices $b,c$ there are two edges $e_{bc}(1),e_{bc}(2)$ corresponding to four half-edges $e_{bc}^b(1),\ldots,e_{bc}^c(2)$, and four more half-edges are provided by edges $e_{ab},e_{ac}$. $|H|=12$, and $\mathcal{E}$ consists of four parts of cardinality two giving the four internal edges and four of cardinality one, namely one for each $l_i$.  $\mathcal{V}_\Gamma$ partitions the twelve half-edges into three corollas of cardinality four constituting the three four-valent vertices.
}
\label{dunce}
\end{figure}
\subsection{Cuts}\label{subsec cuts}
Consider a bridgeless connected graph $\Gamma$. We have 
\[
1=h_0(\Gamma)=|\Gamma|-e_\Gamma+v_\Gamma.
\]
If we want to cut $\Gamma$ by removing edges, the Euler characteristic demands that we remove at least two edges.

From a physicist's viewpoint the cut edges can also be regarded as marked edges which are put on-shell when we apply Feynman rules.

\subsubsection{Refinements}
Given two partitions $\mathcal{P}$ and $\mathcal{P}'$ of a set $S$, we say $\mathcal{P}'$ is a \emph{refinement} of $\mathcal{P}$ if every part of $\mathcal{P}'$ is a subset of a part of $\mathcal{P}$.  Intuitively, $\mathcal{P}'$ can be made from $\mathcal{P}$ by splitting some parts.  The set of all partitions of $S$ with the refinement relation gives a lattice called the \emph{partition lattice} (see for example Chapter 4 of \cite{Gbook}).  The covering relation in this lattice is the special case of refinement where exactly one part of $\mathcal{P}$ is split into two parts to give $\mathcal{P}'$.

We call a refinement maximal if it is an $|S|-1$-refinement of a set $S$, that is $S$ is refined down to singletons. 

\subsubsection{Cuts}
Let us now consider cuts. In general we will use cuts which decompose a graph $\Gamma$ into a disjoint union
\[
\dot{\cup}_{i=1}^k \Gamma_i,
\]
of $k$ graphs $\Gamma_i$ which induce a $k$-partition of $L_\Gamma$.

Such a cut can be obtained by either removing edges from the graph, or by splitting vertices $v$
and therefore partitioning their corollas $\mathbf{c}_v$.  In the way we have defined graphs, this means that a cut can be obtained from refining $\mathcal{E}_\Gamma$ or refining $\mathcal{V}_\Gamma$.  We also consider the situation where both can be refined.

\begin{rem}\label{rem normal}
  We could augment this set up by carrying a maximal chain of refinements along with every split vertex.  This would allow us to add restrictions to the refinement chain.  One restriction in particular that is useful for quantum field theory would be to require the first cut of a vertex to be \emph{normal}, that is to include internal edges on both sides of the cut.  This corresponds to the fact that the physical motivation for considering cuts is to describe monodromy of amplitudes.  If the first split of a vertex is not normal in this way then it will not lead to monodromy of any Feynman integral.

  The reader can check that the maximal chain can be carried though all the definitions and operations that we study in this paper.
\end{rem}

\subsubsection{Pre-cut graphs}\label{sec precut}
\begin{defn}\label{precutgraphsdefn}
A pre-cut graph $\Gamma$ is a pair of graphs $((H_\Gamma, \mathcal{V}_\Gamma, \mathcal{E}_\Gamma), (H_\Gamma, \mathcal{V}_H, \mathcal{E}_H))$ on the same half-edges $H_\Gamma$ such that $\mathcal{V}_H$ refines $\mathcal{V}_\Gamma$ and $\mathcal{E}_H$ refines $\mathcal{E}_\Gamma$.
\end{defn}

By abuse of notation the pre-cut graph and the unrefined graph making it up have the same name ($\Gamma$ in the above).  This is because for physics applications we want to regard the pre-cut graph as being the original $\Gamma$ with the cut edges and split corollas marked, so we view it as a decoration of $\Gamma$, or as $\Gamma$ with extra structure added.  

In view of this, it will also be useful to have the notation $C_\Gamma\subset E_\Gamma$ for the edges which are cut, that is for those edges in $E_\Gamma$ which are not edges in $E_H$.

We will still need unambigious notation for the two graphs making up a pre-cut graph $\Gamma$.  Given a pre-cut graph $\Gamma$ we will use the notation
\begin{align*}
  \hat{\Gamma} & = (H_\Gamma, \mathcal{V}_\Gamma, \mathcal{E}_\Gamma), \text{ and}\\
\tilde{\Gamma} & =(H_\Gamma, \mathcal{V}_H, \mathcal{E}_H).  
\end{align*}
We will call  $\tilde{\Gamma}$ the associated graph.  When it is sufficiently clear we we will write $\Gamma=(\Gamma,H)$ as shorthand for the two graphs making up a pre-cut graph and may simply use $H$ for the associated graph.

\begin{defn}
For a pre-cut graph $\Gamma$ we set $|\Gamma|:=|\hat{\Gamma}|$ and $||\Gamma||:=|\tilde{\Gamma}|.$
\end{defn}

Note that $(\Gamma,\Gamma)$ is a pre-cut graph as the trivial refinement is a refinement.

There is an $h_0(\tilde{\Gamma})$-partition $L_\Gamma(h_0(\tilde{\Gamma}))$ of $L_\Gamma$. We have 
\[
L_\Gamma(h_0(\tilde{\Gamma}))=\tilde{\Gamma}/E_{\tilde{\Gamma}},
\]
which is a $h_0(\tilde{\Gamma})$-partition of the vertex $\hat{\Gamma}/E_{\hat{\Gamma}}$.
\begin{figure}[H]
\includegraphics[width=8cm]{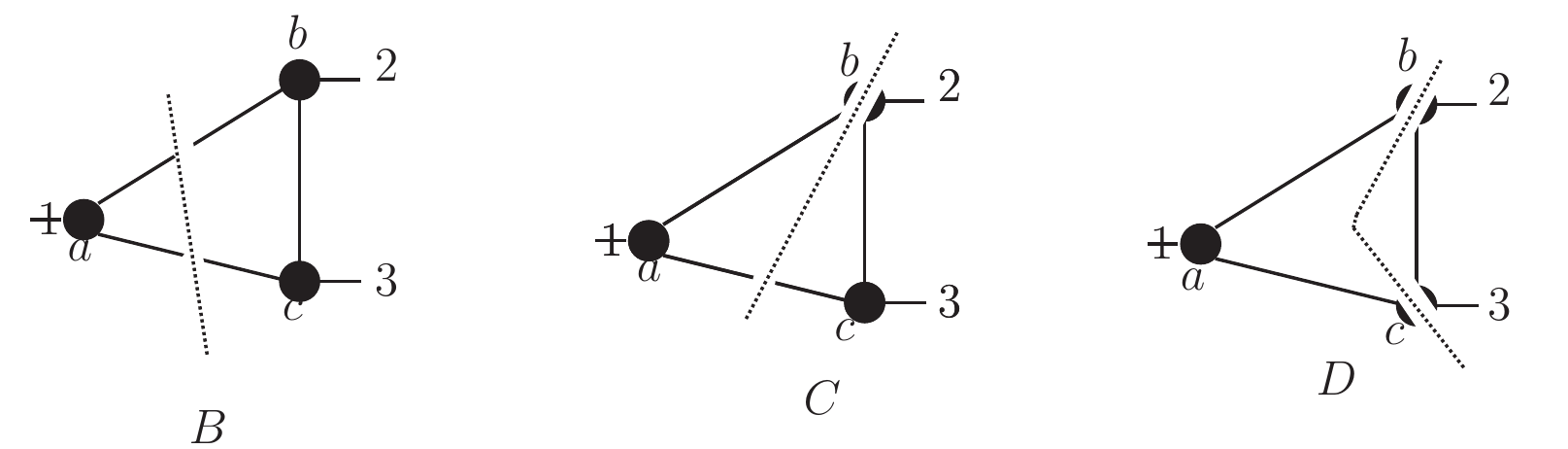}
\caption{Pre-cut graphs. They all correspond to a 2-partition of $L=\{\{l_1\},\{l_2,l_3\}\}
\sim \tilde{\Gamma}/E_{\tilde{\Gamma}}$ for $\Gamma$ any of the three graphs.
$B$ on the left is also a cut graph, for $C$ in the middle we have a normal cut of vertex $b$, for $D$ on the right both vertices $b,c$ have normal cuts.
For $B$, the set $C_B$ has two edges, for $C$ one, and for $D$ none. }
\label{PreCut}
\end{figure}

\subsubsection{Cut graphs}
\begin{defn}
  A pre-cut graph $\Gamma=((H_\Gamma, \mathcal{V}_\Gamma, \mathcal{E}_\Gamma), (H_\Gamma, \mathcal{V}_\Gamma, \mathcal{E}_H))$, that is one where no corollas are split, is a \emph{cut graph}.
\end{defn}
We often want to restrict the notion of cut or pre-cut graph to the notion of Cutkosky or pre-Cutkosky graph. For that we first need to discuss spanning forests, which is the topic of Section~\ref{sec span}.

\subsubsection{Connection to the cut space}

It is worth a brief digression to draw the comparison between this set up and a different but related object which is standard in graph theory, the cut space.  All the graph theory facts outlined below are standard and can be found for instance in \cite{GRbook}.

Cuts in this context are edge cuts, not allowing splits of corollas. 
  Furthermore, they are minimal in the sense that such a cut is defined by a bipartition of the vertices of the graph, the cut itself being exactly those edges with one end in each part.

More specifically, given a graph $\Gamma$, assign orientations arbitrarily.  For each nontrivial ordered bipartiton of the vertices, form the signed sum of the edges of the corresponding cut, with positive sign if the edge is oriented from the first part to the second part and negative sign otherwise.  The span of these vectors is the \emph{cut space}.  It is a subspace of the vector space spanned by the edges of the graph.

Equivalently, the cut space is the row space of the signed incidence matrix of $\Gamma$.  From this characterization we see that a basis of the cut space is given by cuts which detatch exactly one vertex, running over all but one vertex of each connected component of the graph.  For $\Gamma$ connected, note that given an edge $e$ in a spanning tree $T$ of $\Gamma$, $T-e$ has exactly two components and so defines a vertex bipartition of $\Gamma$ and hence a cut.  Thus, in the connected case another natural basis is given by choosing a spanning tree and taking the cuts defined by each edge of the spanning tree.

The cut space does not tell us directly about anomalous thresholds in quantum field theory.  For instance, consider the triangle graph with edges $a$, $b$, $c$, oriented cyclically.  Then $a-b$, $b-c$ and $c-a$ are all elements of the cut space and in fact any two of them generate it.  However there is no way to obtain $a\pm b\pm c$ for any choice of signs as a linear combination of $a-b$ and $b-c$, so the more general cut of the triangle which cuts all three edges leaving three pieces does not appear directly in the cut space.  We do obtain $a-2b+c$ in the cut space, so provided we are working over a field of characteristic not equal to 2 we do see an element with the correct support for the cut of all three edges of the triangle, but not the correct coefficients.

More generally, working over a field of characteristic $0$, define the function from an element $c$ of the cut space to the set of edges which is the support of $c$ (that is the edges which have a nonzero coefficient in $c$).  This map gives us the sets of edges which are cuts in the sense which we are predominantly working with in this paper.  In particular this is a function from the cut space of $\Gamma$ to the space of cut graphs on the fixed base graph $\Gamma$.  However, this is not a linear map as taking supports is not a linear operation.  The interplay of these two different vector spaces and the non-linear function connecting them hints at the subtlety of the study of cuts in quantum field theory.

A different way to obtain simple edge sets from linear combinations of edges is to work over the field with two elements, $\mathbb{F}_2$ (in some sources the term cut space refers to this vector space).  This is not so useful for the purposes of the present paper, but one important fact of note is that the cut space over $\mathbb{F}_2$ is the orthogonal complement to the cycle space of the graph.  The cycle space is defined similarly as the span of vectors which are the sum of edges forming cycles.  Over more general fields the analogue of the cycle space is called the flow space and remains the orthogonal complement of the cut space over the same field.  Additionally, if $\Gamma$ is planar the cut space of $\Gamma$ is the flow space of its planar dual. The flow space has a very physicsy feel; it is essentially what we integrate over when doing momentum space integrals except that momenta are vectors rather than scalars.

An important conjecture due to Tutte is that every bridgeless graph has an element in the flow space over $\mathbb{F}_5$ whose support is all edes of the graph.  Such a flow is called a nowhere zero 5-flow.  Tutte made this conjecture in 1954 \cite{T5flow} and it is demed highly important by the graph theory community.  The evident difficulty of this conjecture again shows the subtlety of taking supports.

\subsection{Spanning Forests and pre-Cutkosky graphs}\label{sec span}
We now proceed to define the notion of Cutkosky and pre-Cutkosky graphs and to discuss the spanning trees and spanning forests that correspond to them. 

\subsubsection{Cutkosky and pre-Cutkosky graphs}

\begin{defn}
  \mbox{}
  \begin{itemize}
    \item A pre-Cutkosky graph $\Gamma$ is a pre-cut graph for which every edge $e\in C_\Gamma$ has the property that the two ends of $e$ are in different components of $\tilde{\Gamma}$.
    \item A Cutkosky graph $\Gamma$ is a cut graph for which every edge $e\in C_\Gamma$ has the property that the two ends of $e$ are in different components of $\tilde{\Gamma}$.
  \end{itemize}
\end{defn}

\subsubsection{Spanning trees and forests}\label{subsec span}
\begin{defn}
A spanning tree $T=(H_T,\mathcal{V}_T,\mathcal{E}_T)$ of a connected graph $\Gamma=(H_\Gamma,\mathcal{V}_\Gamma,\mathcal{E}_\Gamma)$ is a connected subgraph $T\subseteq \Gamma$ such that $H_T\subseteq H_\Gamma$, $H_T\cap L_\Gamma=\emptyset$, $V_T=V_\Gamma$,  which has no cycles, i.e. is simply connected, $v_T-e_T=1$.
\end{defn}
$E_T\subseteq E_\Gamma$ is the set of edges of the spanning tree.

For all $T\in \mathcal{T}(\Gamma)$ and $e\in E_\Gamma\setminus E_T$ there is a unique cycle in $T\cup e$.  This is called the \emph{fundamental cycle} $l(T,e)$ associated to $T$ and $e$.  For any fixed spanning tree $T$, the fundamental cycles associated to $T$ and each of the edges of $E_\Gamma\setminus E_T$ give a basis for the cycle space of $\Gamma$.

\begin{defn}
A spanning $k$-forest $F$ is similarly a disjoint union $\dot{\cup}_{i=1}^k T_i$ of $k$  trees $T_i\subseteq \Gamma$,
such that $\cup_i V_{T_i}=V_\Gamma$. Note $|\Gamma|=|\Gamma/F|$ for any spanning forest $F$ of $\Gamma$.
\end{defn}
$E_F$ is the set of edges of $F$. $e_F=\sum_i e_{T_i}$.
A spanning 1-forest is a spanning tree.

Equivalently, a spanning $k$-forest is the result of taking a spanning tree and removing $(k-1)$ edges from it.
\begin{defn}
A spanning tree $T$ of a pre-cut graph $\Gamma=((H_\Gamma,\mathcal{V}_\Gamma,\mathcal{H}_\Gamma),(H_\Gamma,\mathcal{V}_H,\mathcal{E}_H))$ is 
a spanning tree of each component of the associated graph $\tilde{\Gamma}$.

A spanning $k$-forest $F$ of a pre-cut graph $\Gamma$ is the result of removing $k-1$ edges from a spanning tree of $\Gamma$.
\end{defn}

Given a spanning $k$-forest $F$ of a pre-cut graph $\Gamma$, there are a number of different sets of edges which will be important.  First the edges of the forest themselves are important.  Second are the edges of $\Gamma$ which are not in $F$ but join distinct components of $F$.  Thinking of $F$ as a spanning tree $T$ with some edges removed then all the edges of $T-F$ are in this second class, as well, typically, as others.  Third are the edges of $G$ which are not in $F$ but have both ends in the same tree of $F$. 
The second and third sets of edges above are those which will ultimately be put on-shell, while those in the first set remain off-shell.

We will use the notation $\breve{E}_F$ for the second of the above sets of edges: \[\breve{E}_F:=\{e\in E_{\Gamma}|\; e\not\in E_{F}, |\nu_{G/F}(e)\cap V_{\Gamma/(F\cap \Gamma)}|=2\}\]

\begin{defn}
A spanning $k$-forest $F$ for a pre-cut graph $\Gamma$ is a compatible spanning forest if 
$\emptyset=C_\Gamma\cap E_F$ and $C_\Gamma=\breve{E}_F$.
\end{defn}
That is, a spanning forest is compatible if 
the vertex partition induced by the trees of the spanning forest agrees with the cut or pre-cut of $\Gamma$.  Compatibility ensures that the spanning forest 
is in accordance with the chosen refinements $\mathcal{V}_H,\mathcal{E}_H$.
\begin{lem}
  \mbox{}
  \begin{itemize}
  \item A pre-Cutkosky graph $\Gamma$ is a pre-cut graph $\Gamma$ for which a compatible spanning forest $F$ exists.  
  \item A Cutkosky graph $\Gamma$ is a cut graph for which a compatible spanning forest $F$ exists.
  \end{itemize}
\end{lem}

\begin{proof}
  Assume $\Gamma$ is Cutkosky or pre-Cutkosky.  Choose a spanning tree of each component of $\tilde{\Gamma}$, and call that forest $F$.  Since all cut edges go between two different components of $\tilde{\Gamma}$, all cut edges are in $\breve{E}_F$.  Since $F$ consists of spanning trees of the components of $\tilde{\Gamma}$, no other edges are in $\breve{E}_F$, and hence $F$ is a compatible forest.

  Assume $\Gamma$ is cut or pre-cut and has a compatible forest $F$, then all cut edges are in $\breve{E}_F$ and hence join disjoint components of $\tilde{\Gamma}$ so $\Gamma$ is Cutkosky or pre-Cutkosky respectively.
\end{proof}

Note $h_0(\tilde{\Gamma})=h_0(F)$ for a compatible $F$.
\begin{figure}[H]
\includegraphics[width=7cm]{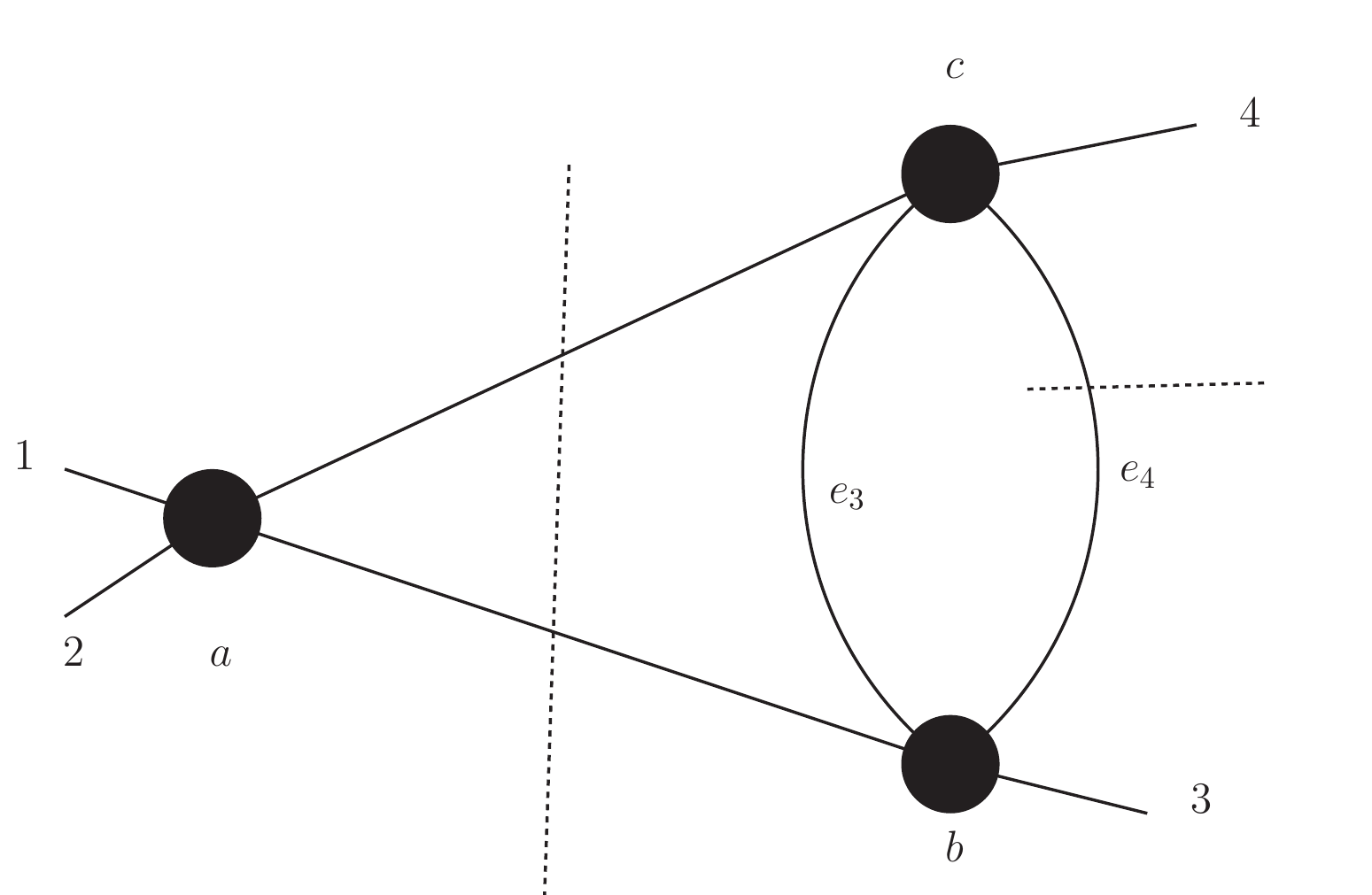}
\caption{Neither Cutkosky nor pre-Cutkosky. The only possible spanning forest is vertex $a$ together with edge $e_3$. But edge $e_4\in C_\Gamma$ connects with both ends to the same component $e_3$ of the spanning forest. On the other hand, in previous Figure~\ref{PreCut} we have from left to right a Cutkosky and two pre-Cutkosky graphs.}
\label{noncutkoskygraph}
\end{figure}

\begin{rem}
Accompanying an ordering of edges in spanning trees is the notion of a cubical chain complex \cite{rational, BlochKreimerOS}  which we use in Secs.(\ref{SecDec},\ref{cointbi},\ref{coactGreen}) when we study sector decompositions and Hodge structures underlying the combinatorics of cut graphs.
\end{rem}

Observe that $\tilde{\Gamma}$ is a union 
\[
\tilde{\Gamma}=\bigcup_i \Gamma_i
\]
of $k$ mutually disjoint  connected  graphs $\Gamma_i\subseteq \tilde{\Gamma}$,
$1\leq i\leq k$ such that $V_{\Gamma_i} = V_{T_i}$ and $\tilde{\Gamma}_i$ includes all internal edges of $\tilde{\Gamma}$ with both ends in $V_{T_i}$ (that is, $\Gamma_i$ is the  subgraph induced by $V_{T_i}$.) The pair $(\Gamma,F)$ defines hence a set of pairs $(\Gamma_i,T_i)$. The one-vertex graph $r_i:=\Gamma_i/T_i$ defines a rose $r_i$  
on $|\Gamma_i|$ petals.

Let us finally collect the various edge sets in play in a definition.
\begin{defn}\label{onoffedges}
The set $E_{\mathit{off}}= E_{\mathit{off}}(\Gamma)$ of edges of a pair of a pre-cut graph $\Gamma$ with compatible forest $F$ is the set $E_{\mathit{off}}=E_F$. The set $E_{on}= E_{on}(\Gamma)$ is the complement of $E_{\mathit{off}}$ in the set of all internal edges of the graph $\hat{\Gamma}=(H_\Gamma,\mathcal{V}_\Gamma,\mathcal{E}_\Gamma)$. Note that $E_{\hat{\Gamma}}=E_{on}\cup E_{\mathit{off}}$. 
\end{defn}
Note that $E_{on}$ contains the edges in $\breve{E}_F$ as well as the edges $e\in r_i$ which shrink to self-loops in the quotient $\tilde{\Gamma}/E_F$.
\subsection{Sub- and co-graphs}\label{subco}
We already have a notion of subgraph.  Now we need to extend to a notion of subgraph for pre-cut graphs and all the other types of graphs we have built on the pre-cut graphs.  Since these will be subgraphs that we will use for building coproducts we will take a more restricted notion of subgraph than what we used for the basic graphs.  Specifically, these subgraphs will be full at the vertices in the sense that if a vertex appears in the subgraph then its whole corolla must appear, and furthermore these subgraphs will be bridgeless, but not necessarily connected.  
If $\gamma\subsetneq \Gamma$ is such a subgraph then 
to the pair $(\gamma,\Gamma)$, we can assign a co-graph $\Gamma/\gamma$ as described below.
Sub- and co-graphs as defined here are the ones we need in Section~\ref{Hopf} to build Hopf algebra structures on the types of graphs introduced so far.

First consider $\langle H_\bullet \rangle$, $\bullet\in \{core,pC,\ldots\}$ to be the $\mathbb{Q}$-vectorspace generated by single bridgeless graphs of the indicated type.
This gives rise to a free commutative $\mathbb{Q}$-algebra structure generated by such graphs by considering disjoint union as a commutative product.

To prepare for more structure we now turn to investigate sub- and co-graphs.
\subsubsection{Subgraphs}
Consider a pre-cut graph $\Gamma=((H_\Gamma,\mathcal{V}_\Gamma,\mathcal{E}_\Gamma),(H_\Gamma,\mathcal{V}_H,\mathcal{E}_H))$. We say that a pre-cut graph $\gamma=\dot{\cup}_i \gamma_i$, with each $\gamma_i$ connected and bridgeless, is a proper subgraph $\gamma\subsetneq \Gamma$ of $\Gamma$ if and only if each $\gamma_i$ is a proper subgraph of $\Gamma$. The definition for connected subgraphs is as follows:
\begin{defn}
A connected pre-cut graph $\gamma=((H_\gamma,\mathcal{V}_\gamma,\mathcal{E}_\gamma),(H_\gamma,\mathcal{V}_h,\mathcal{E}_h))$ is a subgraph 
of the pre-cut graph $\Gamma$, $\gamma\subseteq \Gamma$ if and only if
\begin{enumerate}
\item $H_\gamma\subseteq H_\Gamma$,
\item each part of $\mathcal{V}_\gamma$ is a part of $\mathcal{V}_\Gamma$ and each part of $\mathcal{V}_h$ is a part of $\mathcal{V}_H$,
\item $\mathcal{E}_\gamma$ is a refinement of $\mathcal{E}_\Gamma|_{H_\gamma}$ and $\mathcal{E}_h$ is a refinement of $\mathcal{E}_H|_{H_\gamma}$,
\item (recalling that every part of a $\mathcal{E}$-set is either of size $1$ or $2$) for any part of $\mathcal{E}_H$ which is refined into two parts in $\mathcal{E}_h$ the corresponding part in $\mathcal{E}_\Gamma$ is also refined into two parts in $\mathcal{E}_\gamma$,
\item $\gamma$ is bridgeless.
\end{enumerate}
A subgraph $\gamma\subseteq \Gamma$ is proper, $\gamma\subsetneq \Gamma$, if either $H_\gamma\subsetneq H_\Gamma$ or $\mathcal{E}_\gamma$ is a nontrivial refinement of $\mathcal{E}_\Gamma$.
\end{defn}
The third last restriction says that a subgraph of a pre-cut graph cannot define extra of its edges to be cut; it can only inherit the cuts from the original pre-cut graph. 

To emphasize one of the comments before, note that under this definition a subgraph must be bridgeless, unlike what one might expect from a pure graph theoretical perspective.  This is as it should be for our Hopf algebras.

Note that a graph $\Gamma=((H_\Gamma,\mathcal{V}_\Gamma,\mathcal{E}_\Gamma),(H_\Gamma,\mathcal{V}_\Gamma,\mathcal{E}_\Gamma))$ without cuts has subgraphs $\gamma=((H_\gamma,\mathcal{V}_\gamma,\mathcal{E}_\gamma),(H_\gamma,\mathcal{V}_\gamma,\mathcal{E}_\gamma))$ without cuts.  Compared to our original definition of subgraph of a graph $(H_\Gamma, \mathcal{V}_\Gamma, \mathcal{E}_\Gamma)$, the subgraph $\gamma$ as defined in this section is more restrictive as $\gamma$ must be bridgeless and must be full at the vertices.

Furthermore a pre-cut graph which is pre-Cutkosky will have pre-Cutkosky subgraphs, as compatible forests are preserved in the above definitions, likewise for cut graphs and Cutkosky graphs due to the fullness at the vertices. 

\begin{lem}\label{lem precut subs}
  Let the pre-cut graph $\gamma$ be a subgraph of the pre-cut graph $\Gamma$.  Then
  \begin{itemize}
  \item $\hat{\gamma}$ is a subgraph of $\hat{\Gamma}$ in the sense of Section~\ref{sec graphs},
  \item $\hat{\gamma}$ determines $\gamma$,
  \item if $\gamma'$ is any subgraph of $\hat{\Gamma}$ that is bridgeless and every part of $\mathcal{V}_{\gamma'}$ is a part of $\mathcal{V}_\Gamma$ then $\gamma' = \hat{\gamma}$ for a pre-cut subgraph $\gamma$ of $\Gamma$.
  \end{itemize}
\end{lem}

  \begin{proof}
    The first point is immediate from the definitions since $E_\Gamma$ is a set in bijection with the size $2$ parts of $\mathcal{E}_\Gamma$, so saying $E_\gamma\subseteq E_\Gamma$ just means that $\mathcal{E}_\gamma$ refines $\mathcal{E}_\Gamma$.

    For the second point, $\hat{\gamma}$ immediately determines $H_\gamma$, $\mathcal{V}_\gamma$ and $\mathcal{E}_\gamma$.  Each part of $\mathcal{V}_h$ is a part of $\mathcal{V}_H$ and $\mathcal{V}_h$ is a partition of $H_\gamma$, so $\mathcal{V}_h$ is simply those parts of $\mathcal{V}_H$ which are contained in $H_\gamma$.  This is a refinement of $\mathcal{V}_\gamma$, since $\mathcal{V}_H$ is a refinement of $\mathcal{V}_\Gamma$.  The third last point of the definition of a pre-cut subgraph tells us that $\mathcal{E}_h$ refines $\mathcal{E}_H|_{H_\gamma}$ only when $\mathcal{E}_\gamma$ refines $\mathcal{E}_\Gamma|_{H_\gamma}$ and so $\mathcal{E}_h$ is also determined by $\hat{\gamma}$. 

    For the final point we can follow the steps of the second point provided $\gamma'$ is bridgelss and each part of $\mathcal{V}_{\gamma'}$ is a part of $\mathcal{V}_\Gamma$.    \end{proof}

\subsubsection{Co-graphs}\label{sec co}
For $\gamma\subsetneq \Gamma$, we can form the co-graph $\Gamma/\gamma$. We first consider the case of a connected subgraph $\gamma$ with $l_\gamma\geq 3$.

\begin{defn}
For $l_\gamma\geq 3$, the co-graph $\Gamma/\gamma$ is the graph
\[
((H_{\Gamma/\gamma},\mathcal{V}_{\Gamma/\gamma},\mathcal{E}_{\Gamma/\gamma}),(H_{\Gamma/\gamma},\mathcal{V}_{H/h},\mathcal{E}_{H/h})),
\]
where
\begin{enumerate}
\item $H_{\Gamma/\gamma}=(H_\Gamma\setminus H_\gamma)\dot{\cup}L_\gamma$,\\
\item $\mathcal{V}_{\Gamma/\gamma} = (\mathcal{V}_\Gamma|_{H_\Gamma\setminus H_\gamma}) \dot\cup L_\gamma$ and $\mathcal{V}_{H/h} = (\mathcal{V}_H|_{H_\Gamma \setminus H_\gamma})\dot\cup \gamma/\tilde{\gamma}$ where we regard $\gamma/\tilde{\gamma}$ as a partition of $L_\gamma$ into parts according to the connected components of the associated graph $\tilde{\gamma}$.
\item $\mathcal{E}_{\Gamma/\gamma}= \mathcal{E}_\Gamma|_{H_{\Gamma/\gamma}}$ and $\mathcal{E}_{H/h} = \mathcal{E}|_{H_{\Gamma/h}}$.
\end{enumerate}
\end{defn}  

This is different for the case $l_\gamma=2$. In that case we define
\begin{defn}
For $l_\gamma=2$, the co-graph $\Gamma/\gamma$ is
\[
((H_{\Gamma/\gamma},\mathcal{V}_{\Gamma/\gamma},\mathcal{E}_{\Gamma/\gamma}),(H_{\Gamma/\gamma},\mathcal{V}_{H/h},\mathcal{E}_{H/h})),
\]
where
\begin{enumerate}
\item $H_{\Gamma/\gamma}=(H_\Gamma\setminus H_\gamma)$,\\
\item $\mathcal{V}_{\Gamma/\gamma} = \mathcal{V}_\Gamma|_{H_{\Gamma/\gamma}}$ and $\mathcal{V}_{H/h} = \mathcal{V}_H|_{H_{\Gamma/\gamma}}$,
\item 
  let $e_1, e_2 \in E_G$ be the two unique edges such that  $L_g\cap e_i\not=\emptyset$, and let $h_1, h_2$ be the two half edgs in those edges which are not in $L_\gamma$, then $\mathcal{V}_{\Gamma/\gamma} = (\mathcal{V}_\Gamma|_{H_{\Gamma/\gamma} \setminus \{h_1, h_2\}}) \dot\cup \{h_1, h_2\}$ and $\mathcal{V}_{H/h} = (\mathcal{V}_{H}|_{H_{G/g} \setminus \{h_1, h_2\}}) \dot \cup \{h_1\}, \dot\cup \{h_2\}$.
\end{enumerate}
\end{defn}
Note that in the above, when we take a partition and then $\dot\cup$ one or two additional sets then we are adding those additional sets as parts to the partition, and likewise the $\dot\cup$ of two partitions of disjoint sets is the partition whose parts are the parts of the two partitions.

\begin{rem}
  This is one point where including the chain of refinements with a normality condition does need consideration as by this definition $\Gamma/\gamma$ may or may not be a pre-cut graph as the normality condition of vertex cuts may or may not be satisfied.  If we were to require a normality condition then when we used this co-graph construction for a Hopf algebra, we would restrict to only those subgraphs which do give the normality condition for $\Gamma/\gamma$ and hence for which $\gamma$ and $\Gamma/\gamma$ are both pre-cut.
\end{rem}

For $\gamma=\dot{\cup}_{i=1}^k \gamma_i$ we set $\Gamma/\gamma=\Gamma/\gamma_1/\gamma_2\cdots/\gamma_k$ dividing from left to right.  By the disjointness of the $\gamma_i$, changing the order of the $\gamma_i$ does not affect the co-graph.

Figure~\ref{coprodself} gives an example for an uncut graph. 
\subsubsection{Decomposing graphs}
With the notion of sub- and co-graph, we note that each choice of a sub-graph $\gamma\subsetneq \Gamma$ gives rise to a pair $(\gamma,\Gamma/\gamma)$. Summing over all such pairs, sometimes with appropriate restrictions, gives rise to a coproduct which we discuss in Section~\ref{Hopf}.
\begin{figure}[H]
\includegraphics[width=12cm]{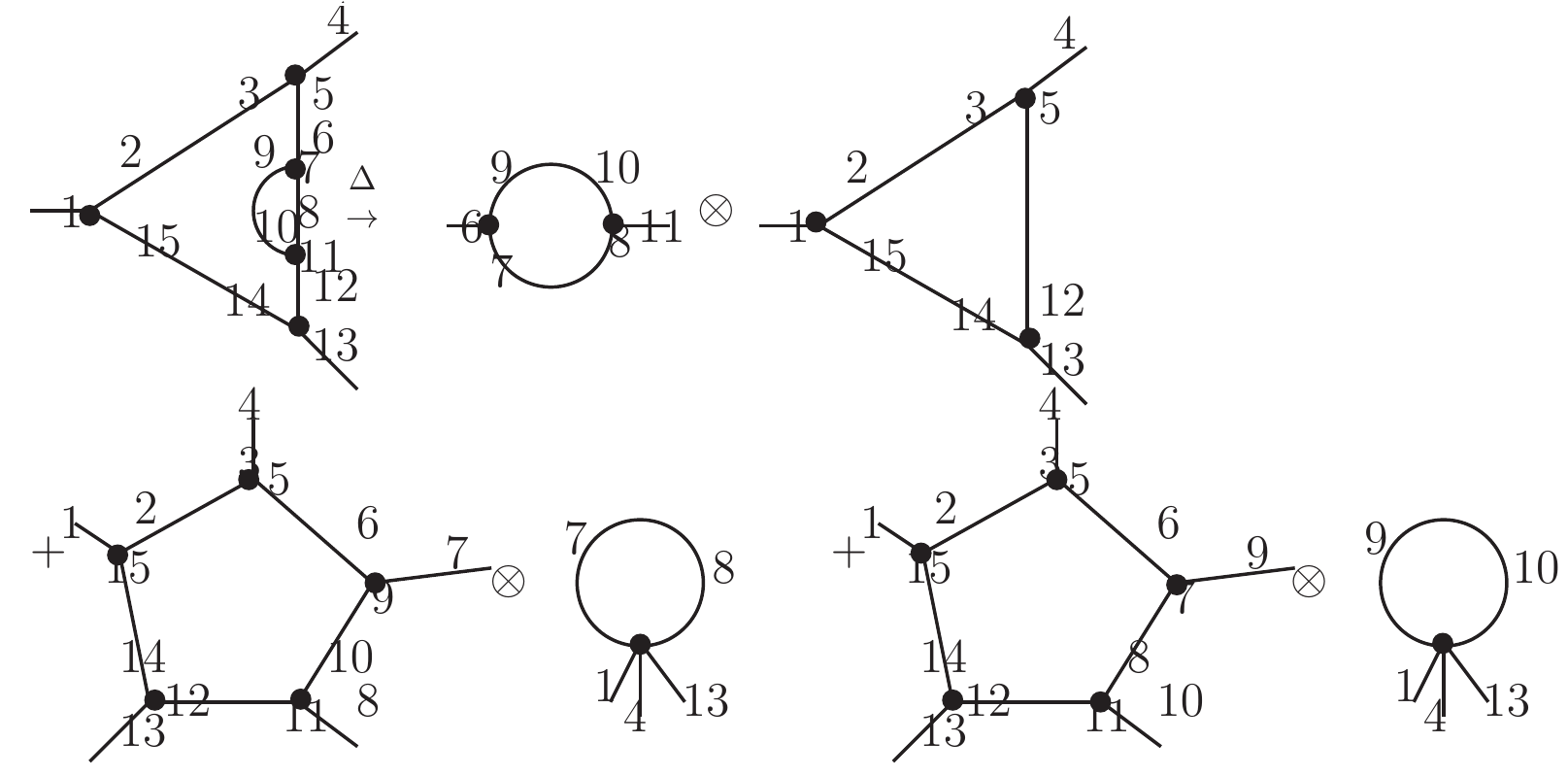}
\caption{An uncut  graph with labeled half-edges. We introduce a map $\Delta$ which decomposes the graph into its subgraphs $\gamma$ and co-graphs $\Gamma/\gamma$ and notate this as $\gamma\otimes \Gamma/\gamma$. In the first term on the rhs, note that we have $l_\gamma=2$ and half-edges $6,11$ are external to $\gamma$ but do not appear as half-edges of $\Gamma/\gamma$. Note that
half-edges $5,12$ form an edge in $\Gamma/\gamma$. The next term gives a subgraph $\gamma$ with $H_\gamma=H_\Gamma$, but half-edges $7,8$ are external in $\gamma$ while internal in $\Gamma$ and $\Gamma/\gamma$. In the last term, half-edges $7,8$ and $9,10$ switch roles.}
\label{coprodself}
\end{figure}
We can also consider pairs $(\Gamma,F)$ of a graph and a spanning forest.
To such a pair, as a sub-construction we have to find pairs $(\gamma,f)$ of sub-graphs and sub-forests and corresponding pairs $(\Gamma/\gamma,F/f)$ of co-graphs and co-forests.

For a pair $(\Gamma,F)$ of a graph and a forest and $l_\gamma\geq 3$, to any subgraph $\gamma$ of $\Gamma$  as defined above
we have an accompanying subforest $f:=F\cap \gamma$ (formed precisely of those edges of $F$ which are also edges of $\gamma$) and a corresponding co-graph $\Gamma/\gamma$ and co-forest $F/f$ (this contraction is usual graph contraction). Note that the co-forest, as defined here, is not necessarily a forest, however we will only be interested in the situations where $F/f$ is a forest.  To this end we say that $(\gamma,f)$ is sub to $(\Gamma,F)$, if and only if $\gamma$ is a subgraph of $\Gamma$ and $F/f$ is a forest of $\Gamma/\gamma$.  As an example where this is important, consider Figure~\ref{G and spT}.  In the figure, note that there is not a term for the subgraph where the internal edges of $\gamma$ consist of ${2,3}, {5,6}, {7,8}, {11,12}, {14,15}$.  For this value of $\gamma$ and the marked spanning tree, the co-forest is the loop ${9,10}$ whish is certainly not a forest.  Additionally, what should be a spanning subtree of $\gamma$ is actually a spanning subforest, with one component of the forest being an isolated vertex; this is not a problem in and of itself, but this change in number of components indicates that the co-forest will contain cycles and so not be a forest.

For $l_\gamma=2$, we proceed as above with the understanding that if the half-edges in $L_\gamma$, $h_1(\gamma),h_2(\gamma)$, are both in the spanning forest $F$ of $\Gamma$, then the edge formed of their other halves in $\Gamma/\gamma$ is in the spanning forest of $\Gamma/\gamma$, else it is not.
This is in accordance with the treatment of self-enery sub-graphs in \cite{BrownKreimer}.

The next two figures are instructive.   
\begin{figure}[H]
\includegraphics[width=12cm]{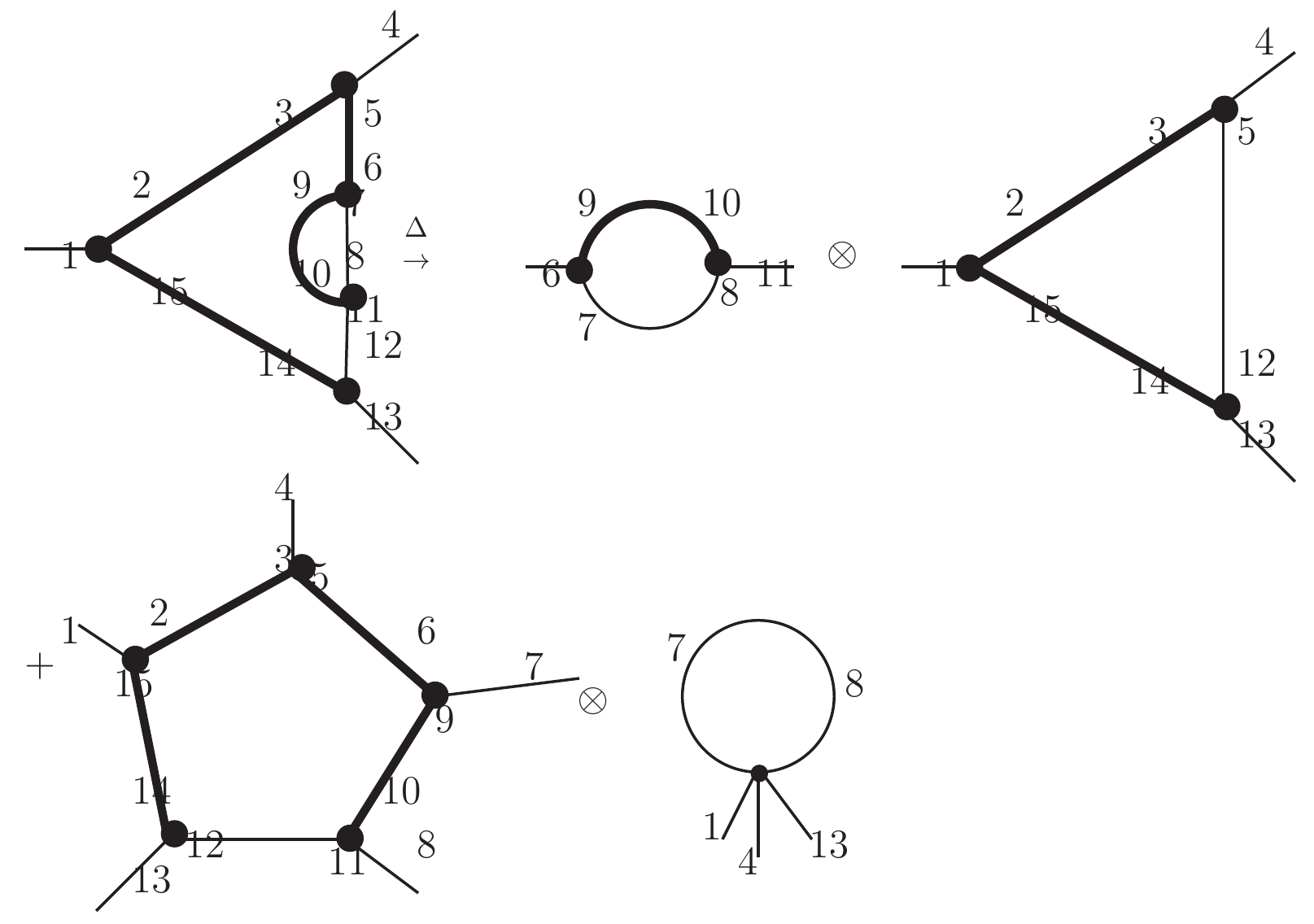}
\caption{A graph and spanning tree with labeled half-edges. The spanning tree is indicated by thickened edges, so the pair $(\Gamma,T)$ is given by  graph $\Gamma$ with bold edges for $T$. Note that in the first co-graph on the right, the edge $(5,12)$ is not part of the spanning tree, even if the half-edge $5$ was part of the spanning tree of $\Gamma$.}
\label{coprodselfspTone}\label{G and spT}
\end{figure}
It is also instructive to study an example with forests.
\begin{figure}[H]
\includegraphics[width=12cm]{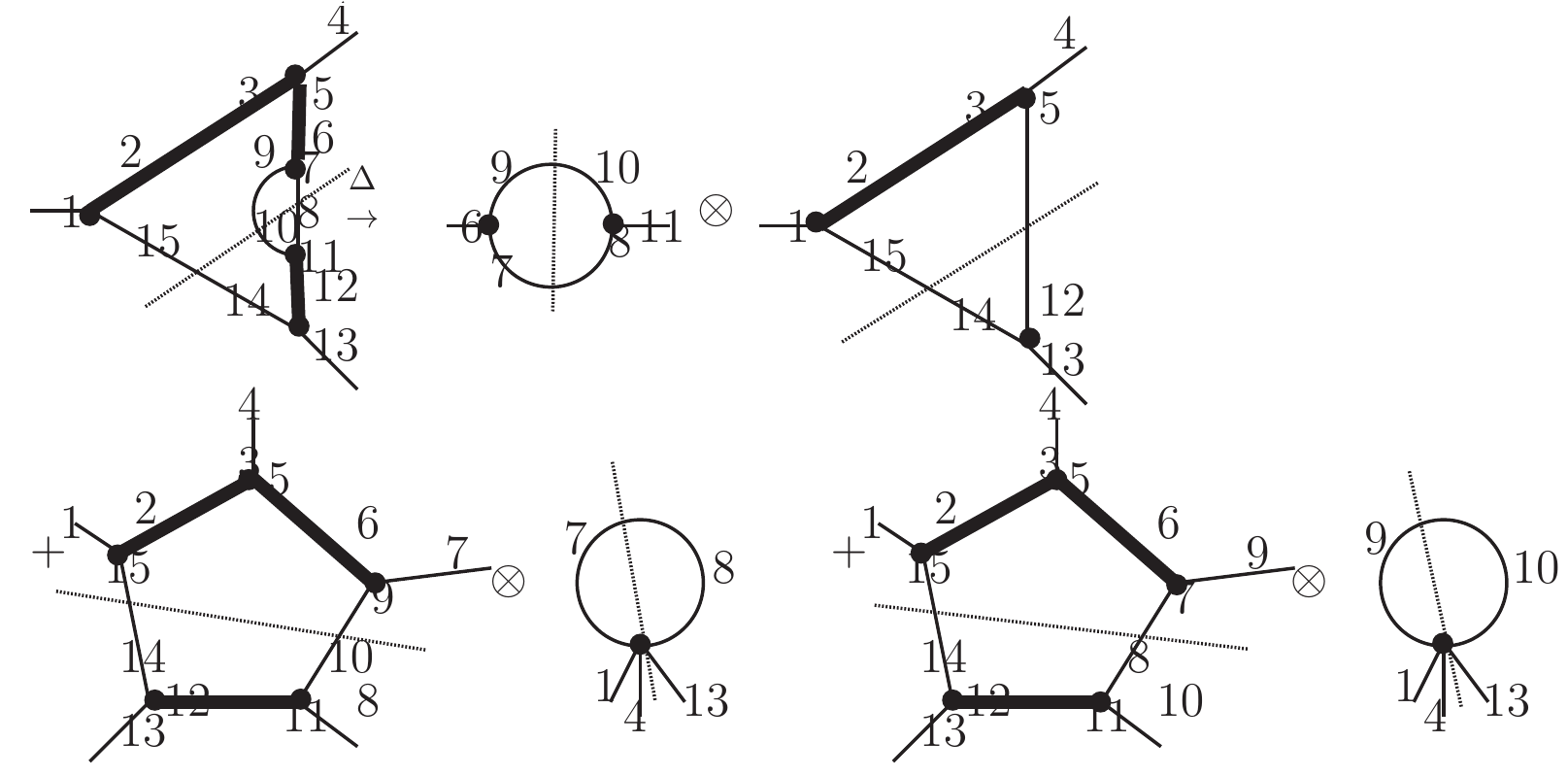}
\caption{A graph and a spanning forest again indicated by bold edges. We also indicate the separation of the graph induced by the Cutkosky cut.
While we start with a graph which is Cutkosky, the decomposition generates graphs which are pre-Cutkosky.}
\label{coprodselfspFone}
\end{figure}

\subsubsection{Composing graphs}\label{composing}
Dual to a decomposition of graphs, we can compose large graphs from smaller ones. This is based on summing over all  bijections between the external legs $L_\gamma$ of a graph $\gamma$ to be inserted with the half-edges at an insertion place of a graph $\Gamma$ in which we insert. For $l_\gamma=2$, all internal edges $e\in E_\Gamma$
are insertion places and we insert by identifying the two half-edges of $L_\gamma$ with the two half-edges of $e$, whilst for $l_\gamma \geq 3$, we sum for each vertex $v\in G$ over all bijections of $L_\gamma$ with $\mathbf{c}_v$, and we take the pairing between $L_\gamma$ and $\mathbf{c}_v$ given by the bijection as defining new edges. Figure~\ref{composingg} gives an example. The resulting compositions give rise to a pre-Lie structure as discussed in Section~\ref{Hopf}.

Note that the parts of the vertex partition $\mathcal{V}_\Gamma$ provide the insertion places for the insertion of vertex graphs and the parts of the edge partition $\mathcal{E}_\Gamma$ provide the insertion places for self-energy graphs.

When inserting (pre)-Cutkosky graphs, the partition of $L_\gamma$ for $l_\gamma\geq 3$ given by the cut has to match the vertex partitions of the parts of $\mathcal{V}_\Gamma$ in $\mathcal{V}_H$ and similarly for $l_\gamma=2$ and edge partitions.

\begin{figure}[H]
\includegraphics[width=12cm]{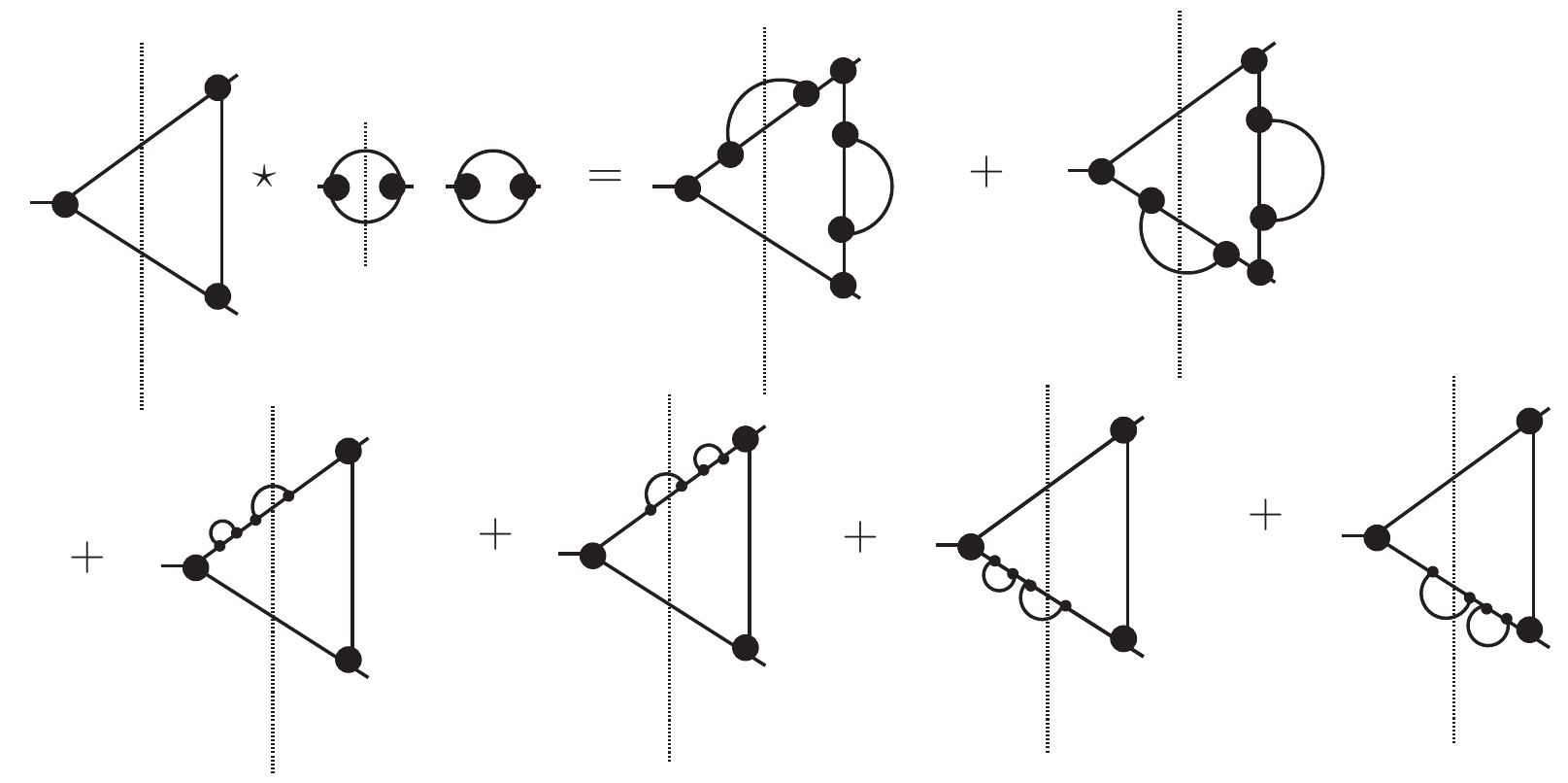}
\caption{Composing for the example of Cutkosky graphs.}\label{composingg}
\end{figure}
The generalization to the composition of pairs $(\gamma,f)$ with pairs $(\Gamma,F)$  is straighforward
as the spanning forests of the parts combine to the spanning forest of the composed graph.

Also the graphs $\gamma$ to be inserted do not have to be bridgeless as for the insertion only the set $L_\gamma$ is relevant; no information on the internal structure of $\gamma$ is needed. 
\section{Feynman rules}\label{appC}
It is useful to collect some properties of Feynman rules which illuminate why the coactions in Section~\ref{coactions}  are useful.
\subsection{The space of external parameters}\label{QL}
Below we consider two variants of the renormalized Feynman rules $\Phi_R: H_\bullet \to \mathbb{C}$,
$\Gamma\to\Phi_R(\Gamma)$. Both have in common that the function $\Phi_R(\Gamma)=\Phi_R(\Gamma)(q)$ is a function of a vector $q$ and of masses. We first describe this dependence. 

The Feynman integral $\Phi_R(\Gamma)(q)$ is function of a vector $q\in Q^L=Q^L(\Gamma)$, where $Q^L$ is a real 
vector space of dimension $l_\Gamma(l_\Gamma-1)/2$ spanned by Lorentz invariants.\footnote{We assume $l_\Gamma\leq D$. Else, the vector space $Q^L$ has a  dimension $\mathtt{dim}(Q^L)\leq l_\Gamma(l_\Gamma-1)/2$  as $D$ dimensions can only accomodate $D$ independent momenta. 
For example for $D=4$, $\mathtt{dim}(Q^L)=4l_\Gamma-10$ for $l_\Gamma\geq 4$. See Sec.(6) in \cite{BlochKr1loop} and Sec.(6-2-4) in \cite{ItzZ}. This does not alter our approach in any essential way.} 
\[
q(e_i)\cdot q(e_j)=
(q(e_i)^2+q(e_j)^2-(q(e_i)-q(e_j))^2)/2,\, \forall e_i,e_j\in L_G.
\]   
Here $q(e_i)\in\mathbb{M}^D$ is a $D$-dimensional vector $(q_0(e_i),q_1(e_i),\ldots q_{D-1}(e_i))^{\mathsf{T}}$, $q_j(e_i)\in\mathbb{C}$, $j=0,\ldots, D-1$ and
$q_j(e_i)^2:=q_j(e_i)\bar{q_j}(e_i)=|q_j(e_i)|^2$.
\[
q(e_i)\cdot q(e_i)= q(e_i)^2=|q_0(e_i)|^2-|q_1(e_i)|^2-\ldots -|q_{D-1}(e_i)|^2.
\]
Note $|q_i(e_j)|\geq 0$ while $-\infty\leq q(e_i)^2\leq+\infty$. 

A choice of a two-partition $p_2:V_1\dot{\cup} V_2=V_G$ of vertices of $\Gamma$ defines a vector $s=s(p_2):=(\sum_{j\in V_1}q_{e_j})^2\in Q^L$. This decomposes $Q^L=Q^L_\perp\times \mathbb{R}_0$ where $\mathbb{R}_0$ is the one-dimensional real vectorspace spanned by $s$. We write $\Phi_R(\Gamma)(q)=\Phi_R(\Gamma)(s,q_\perp)$ correspondingly.

In physics one often tends to complexify and replaces $Q^L=Q^L(\mathbb{R})$ through
$Q^L=Q^L(\mathbb{C})$ by allowing the scalar products $q(e_i)\cdot q(e_i)$ to take complex values. For a mathematical analysis see Sec.(6) of \cite{BlochKr1loop}.
We follow this set-up but restrict to the complexification of a normal threshold variable $s$ defined by any 2-partition.

We set $Q^L(\mathbb{C})\supsetneq Q^L_\mathbb{C}:=Q^L_\perp\times \mathbb{C}_0$ where $\mathbb{C}_0$ is the one-dimensional complex vectorspace where we allow $s\in\mathbb{C}$ whilst we keep
$Q^L_\perp=Q^L_\perp(\mathbb{R})$ real. Note that $Q^L_\mathbb{C}=Q^L_\mathbb{C}(p_2)$ depends on the chosen 2-partition $p_2$ and different choices for $p_2$ lead to different subspaces $Q^L_\mathbb{C}$ of the $p_2$-independent  $Q^L(\mathbb{C})$.

For chosen $p_2$ and fixed $q_\perp$ 
and any $\Gamma$ we consider $\Phi_R(\Gamma)(s,q_\perp)$ as a function 
\[
\Phi_R(\Gamma)(\cdot,q_\perp):X_{|\Gamma|,l_\Gamma}\times\mathbb{C}\to \mathbb{C},
\]
 where all mass squares $m_e^2$, $e\in E_\Gamma$ are kept fixed and
 $X_{|\Gamma|,l_\Gamma}$ is defined in Equation~\ref{Xns}.

For a core graph $\Gamma$ a cut through edges $e\in C_\Gamma$ compatible with $p_2$ determines a Cutkosky graph $\Gamma\setminus C_\Gamma$ such that $\Phi_R(\Gamma)(s,q_\perp)$ has a singularity at $s=s_0:=(\sum_ {e\in C_e}m_e)^2$ and for $s>s_0$, 
\[
\Im(\Phi_R(\Gamma)(s,q_\perp))=\Phi_R(\Gamma\setminus C_\Gamma).
\]

Refinements of $p_2$ determine further singular points $s_i\in\mathbb{R},i\gneq 0$ 
with monodromy $\Phi_R(\Gamma)(s_i,q_\perp)$.
In Section~10 of \cite{MarkoDirk} the reader finds an example for the triangle graph 
with a normal threshold $s_0=s_{normal}>0$ in a variable $s$ determined by a 2-partition of its three vertices and  a further anomalous threshold at a point $s_{anom}=s_1$, $0<s_1<s_0$ determined by refining the partition to the
leading  3-partition of the vertices. 
\begin{rem}
For any chosen $p_2$  and real $q_\perp$ the domain,
\[
D_{p_2}=\{q\in Q^L(\mathbb{C})\cap \Im(\Phi_R(\Gamma)(s,q_\perp))=0\}\subseteq \mathbb{R}_{s_-}\times Q^L_\perp(\mathbb{R}),\,\forall p_2,
\]
where $\mathbb{R}_{s_-}$ is the real open half-axis $[-\infty,s_-[$ and $s_-$ the lowest (anomalous) threshold in all refinements of $p_2$. It is commonly conjectured that the intersection over all 2-partitions $p_2$ of $V_\Gamma$, $\cap_{p_s} D_{p_2}$, is non-empty and allows for analytic continuations of $\Phi_R(\Gamma)(q)$ to boundary values of analytic functions.  
\end{rem}
We now define $\Phi_R(\Gamma)(q)$ in parametric and momentum space renormalization. 
Here $q\in Q^L$. In the notation $\Phi_R(\Gamma)(q)$ it is understood that the renormalized Feynman rules $\Phi_R$ are not only parametrized by $q$ but also by a renormalization point $\mu\in Q^L$ and we use  $\Phi_R(\Gamma)(q)$ as a shorthand for
$\Phi_R(\Gamma)(q,\mu)$. Furthermore for parametric renormalization we give results for $\Gamma$ a core graph and for momentum space renormalization we give results for $\Gamma$ a more general Cutkosky graph. An extended discussion is in \cite{MarkoDirk}.
\subsection{Parametric Feynman rules}\label{parametricFR}
Parametric Feynman rules are based on the use of the two well-known Symanzik polynomials $\mathfrak{s}_\Phi,\mathfrak{s}_\psi$, see for example \cite{BrownKreimer}.\footnote{We let $\mathfrak{s}_\Phi$ be the second Symanzik polynomial with masses and $\mathfrak{s}_\psi$ be the first Symanzik polynomial.}

Renormalized amplitudes in $D$ dimensions of spacetime in the parametric representation deliver 
for a graph $\Gamma$ a renormalized  integrand $\mathsf{Int}_R(\Gamma)$ which is a  form
\be\label{FRInt}
\mathsf{Int}_R(\Gamma)(q,p)=\sum_{F}(-1)^{|F|}\frac{\ln\left(\frac{{\mathfrak{s}_\phi}(\Gamma/F){\mathfrak{s}_\psi}(F)+{\mathfrak{s}_\phi}_0(F){\mathfrak{s}_\psi}(\Gamma/F)}{{\mathfrak{s}_\phi}_0(\Gamma/F){\mathfrak{s}_\psi}(F)+{\mathfrak{s}_\phi}_0(F){\mathfrak{s}_\psi}(\Gamma/F)}\right)}{{\mathfrak{s}_\psi}^{D/2}(\Gamma/F){\mathfrak{s}_\psi}^{D/2}(F)}\Omega_\Gamma,
\ee
in the notation of \cite{BrownKreimer,BlochKreimerOS}.
Here $p\in \mathbb{P}_\Gamma= \mathbb{P}^{e_\Gamma-1}(\mathbb{R}_+)$ varies as a function of the edge lengths $A_e\geq 0$ and
$q\in Q^L=Q^L(\Gamma)$ is a point in the vectorspace of Lorentz invariants $q(h)\cdot q(g)$ spanned by vectors $q(h),q(g)$, for all  $g,h\in L_\Gamma$ as above. Furthermore,
\[
\Omega_\Gamma=\sum_j(-1)^j dA_1\wedge\cdots\wedge dA_{j-1}\wedge dA_{j+1}\wedge\ldots \wedge dA_{e_\Gamma},
\]
where $dA_j$ is ommitted. For ${\mathfrak{s}_\phi}_0$ in the above the subscript ${}_0$ indicates 
a renormalization point $\mu\in Q^L$ to be used in the evaluation of the second Symanzik polynomial ${\mathfrak{s}_\phi}:Q^L \times \mathbb{P}_\Gamma\to\mathbb{C}$.

The renormalized Feynman integral is
\[
\Phi_R(\Gamma)(q)=\int_{\mathbb{P}_\Gamma}\mathsf{Int}_R(\Gamma)(q,p).
\]
Note that the above is true if the graph $\Gamma$ and therefore all its co-graphs is overall logarithmic divergent ($\omega(\Gamma)=0$) and so are its divergent subgraphs.

Other degrees of divergence $\omega(\Gamma)>0$ of $\Gamma$  
demand subtractions such that the first $\omega(\Gamma)+1$ Taylor coefficients vanish when expanding $\Phi_R(\Gamma)(q,p)$ in $q$ around the renormalization point $q_0$
and modify the above formula slightly. See \cite{BrownKreimer}.  
Here $\omega(\Gamma)=(|\Gamma|\times D-2e_\Gamma)/2$ in the conventions of \cite{BrownKreimer}.

Note that $\mathsf{Int}_R(\Gamma)(q,p)$ is a log-rational function of $p$ which is well-defined in the interior of the simplex $A_e\gneq 0$ by assumption (so we assume there are no singularities when a subset of the $A_e$ becomes large, a corresponding singularity would be known as an infrared singularity in physicics parlance).

It is also well-defined along $A_e\geq 0$ including the boundaries thanks to the signed forest sum for the antipode $S_{ren}$ in a suitable  renormalization Hopf algebra $H_{ren}$, a suitable quotient (see Section~\ref{quotient})  of the core Hopf algebra,
\be\label{forestsum}
S_{ren}(\Gamma)=-\sum_{F}(-1)^{|F|}F\times \Gamma/F,
\ee
apparent in Equation~\ref{FRInt}.
The latter sum reflects the presence of the antipode $S_R^\Phi$ in $\Phi_R=m_\mathbb{C}(S^\Phi_R\otimes \Phi)\Delta_{ren}$
and $S_R^\Phi=-R\left(m_\mathbb{C}(S^\Phi_R\otimes \Phi\circ P_{aug})\Delta_{ren}\right)$.
$P_{aug}$ is the projection into the augmentation ideal $\mathbf{Aug}(H_{ren})$.

This goes back to the interplay between Zimmermann's forest formula and factorization properties of the above ${\mathfrak{s}_\psi}$ and ${\mathfrak{s}_\phi}$ polynomials 
\[
{\mathfrak{s}_\psi}(\Gamma)={\mathfrak{s}_\psi}(\Gamma/F){\mathfrak{s}_\psi}(F)+R^\Gamma_F,\,{\mathfrak{s}_\phi}(\Gamma)={\mathfrak{s}_\phi}(\Gamma/F){\mathfrak{s}_\psi}(F)+\tilde{R}^\Gamma_F,
\]
with remainders $R^\Gamma_F$, $\tilde{R}^\Gamma_F$ which are of higher degree in the edge  variables $A_e,e\in E_F$  provided by $F$ than is ${\mathfrak{s}_\psi}(F)$.

Hence poles along $A_e=0$ are eliminated \cite{BrownKreimer} in the signed forest sum Eq.(\ref{forestsum}) above. Note that for
${\mathfrak{s}_\psi}(\Gamma)$ to vanish we need that a subset $A_e,\, e\in E$ tends to zero where the set $E$ covers a loop.  

\subsection{Momentum space Feynman rules}\label{momspaceFR}
For momentum space Feynman rules we follow \cite{MarkoDirk,BrownKreimer} and quote from \cite{MarkoDirk} where details can be found.
 
For a core graph $\Gamma$ the  renormalized integral $\Phi_R(\Gamma)(q)$, $q\in Q^L$, is given as 
\bea
\Phi_R(\Gamma)(q) & = & \int_{-\infty}^{\infty}\prod_{i=1}^{|\Gamma|}dk_{i,0}
\prod_{j=1}^{|\Gamma|}\int d^{D-1}\vec{k}(j)
\underbrace{\left(\frac{1}{\prod_{e\in E_\Gamma}Q_e}\right)^R}_{\mathsf{Int}(\Gamma)(q,k)}\nonumber\\
 & = &
\sum_{i=1}^{\xi_\Gamma}\left\{ \left(\prod_{j=1}^{|\Gamma|}\int d^{D-1}\vec{k}(j)\right) \times \left(\prod_{j=1}^{|\Gamma|}
\sum_{T\in\mathcal{T}(\gamma_j^{(i)})}
\mathbf{\bar{pf}}(T)\right)^R_{k_0(j)=+\sqrt{\vec{k}(j)^2-m_{\acute{T}}^2+i\eta}}\right.\nonumber\\
 & & \left.\times
\frac{1}{+\sqrt{\vec{k}^2-m_{\acute{T}}^2+i\eta}}\right\}.\nonumber
\eea
The above renormalized 
product over quadrics 
$\left(\prod_{e}\frac{1}{Q_e}\right)^R$ corresponds to a signed forest sum over such products where in each element $\gamma\in F$ of a forest $F$ the momenta external to the internal momenta of $\gamma$ are evaluated according to the kinematic renormalization conditions for graphs with $l_\Gamma$ external legs.
The unrenormalized product $\left(\prod_{e}\frac{1}{Q_e}\right)$ delivers the unrenormalized Feynman integrand in momentum space, see  \cite{MarkoDirk} for details.
 
This can be written as a sum over all spanning trees of $\Gamma$ together with a sum of all orderings of the space like integrations in accordance with the flag structure and we find $\Phi_R(\Gamma)(q)  =  \sum_{T\in\mathcal{T}(\Gamma)}\Phi_R({\Gamma_T})$,
\bea\label{FRmomexplicit}
\Phi_R({\Gamma_T})  & = &   \sum_{\sigma\in S_{|\Gamma|}}\int_{0<s_{\sigma(|\Gamma|)}<\cdots<s_{\sigma(1)}<\infty}\left(\prod_{e\in E_T}
\frac{1}{Q_e}\right)^R_{k(j)_0^2=s_j+m_j^2,\,j\not\in E_T}\times\\
 & \times & \frac{1}{+\sqrt{\vec{k}^2-m_{\acute{T}}^2+i\eta}}
\prod_{j\not\in E_T}ds(j).\nonumber
\eea
Here $Q_e$ are the quadrics assigned to internal edges $e$ and $\xi_\Gamma$ is the number of flags (summands) in $\tilde{\Delta}_{core}^{|\Gamma|-1}(\Gamma)$.

For a Cutkosky graph $\Gamma$, $|\Gamma|\gneq ||\Gamma||$ so that $\bar{\Delta}_0(\Gamma)=\gamma\otimes \Gamma/\gamma$, we have 
\[\Phi_R(\Gamma):=\bar{\Phi}(\Gamma),\]
as only the subloops left intact, provided by $\gamma$,  have to be renormalized and so
\bea\label{Fubini}
\Phi_R(\Gamma)(q) & = & \int\prod_{j=1}^{|\Gamma/\gamma|}d^Dk_j \left(\frac{1}{\prod_{e\in E_{F_0}}Q_e}\right)_{\cap_{f\in E_{on}^{\Gamma/\gamma}}(Q_f=0)}\times \nonumber\\
 & \times & \underbrace{\sum_{t\in\mathcal{T}(\gamma)}\overbrace{\int\prod_{j=1}^{|\gamma|}d^Dk_j\left(\prod_{e\in E_t}\frac{1}{Q_e}\right)^R_{\cap_{f\in (E_\gamma-E_t)}(Q_f=0)}}^{\Phi_R({\gamma_t})}}_{\Phi_R(\gamma)},
\eea
where the superscript $R$ indicates to sum over all terms needed for renormalization as before, using that the renormalization Hopf algebra $H_{ren}$ is a quotient of $H_{core}$ and coacts accordingly as we used in  Sec.(\ref{hcorehpc}).
\begin{rem}\label{higherdegreediv}
Let us consider how to handle higher degrees of divergence $\omega(\Gamma)>0$ when renormalizing products 
\[P(\Gamma)(q):=\left(\prod_{e\in E_T}
\frac{1}{Q_e}\right).
\]
For a Feynman graph $\Gamma$ with space of external kinematics $Q^L=Q^L(\Gamma)$, the complex vector space of invariants $q(e)\cdot q(f)$, $e,f\in L_\Gamma$, fix a real point $\mu\in Q^L$
in generic position (away from monodromies).
The product $P(\Gamma)=P(\Gamma)(q),q\in Q^L$ is a function of $q\in Q^L$.
For $\omega(\Gamma)=0$, renormalization of the overall divergence proceeds by subtraction
$P^R(\Gamma)(q,\mu)= P(\Gamma)(q)-P(\Gamma)(\mu)$. The Taylor multi-variable expansion  $T(P(\Gamma))(q,\mu)$ expansion of $P(\Gamma)(q)$ near $\mu$ is 
\[
T(P(\Gamma))(q,\mu)=\sum_{|\alpha|>0}\frac{(q-\mu)^\alpha}{\alpha!}(\partial^\alpha P(\Gamma))(\mu),
\]
where a multi-index notation adopted to the fact that $Q_\Gamma$ is a higher dimensional vector space is understood. Let $T(P(\Gamma))(q,\mu)^{[j]}$ be the terms up to order $|\alpha|^j$ in this, so  $T(P(\Gamma))(q,\mu)^{[0]}=P(\Gamma)(\mu)$ and we have 
\[
P^R(\Gamma)(q,\mu)=P(\Gamma)(q)-T(P(\Gamma))(q,\mu)^{[\omega(\Gamma)]},
\]
for general $\omega(\Gamma)\geq 0$. The iteration of this renormalization to treat divergent subgraphs is compatible with the Hopf algebra structure of renormalization and the corresponding forest formula \cite{Collins,K}.
\end{rem}

\section{Cointeraction}\label{appB}
In this appendix we define and prove the existence of  cointeracting bialgebras underlying monodromies and renormalization.
\subsection{The incidence coalgebra $I_N$  of a set $N$}\label{incidence}
We start from a partially ordered set $N$  with partial order $\subseteq$.
For $x, y \in N$, $x \subseteq y$, define the interval
\[
[x, y] = \{z \in N|x \subseteq  z \subseteq  y\}.
\]
$N$ gives rise to an incidence coalgebra  denoted by $I_N$  upon setting
\[ 
\rho: I_N\to I_N\otimes I_N,\,\rho([x, y]) = \sum_{
x\subseteq z\subseteq y}
[x, z] \otimes [z, y],
\] 
for the coproduct $\rho$  and
\[
\hat{\One}([x, y]) = \delta_{x,y}, 
\]
for the counit $\hat{\One}:I_N\to\mathbb{Q}$, where
$\delta_{x,y} = 1,$ $x = y$, $\delta_{x,y} = 0$ otherwise.  See \cite{Sincidence} for details.

Note that $\rho([x, x]) = [x, x] \otimes [x, x]$ is group-like. 

Consider the $\mathbb{Q}$ vector space $I_N^0$  generated by intervals $[\emptyset,x]$, $x\in N$. As long as $N\neq \emptyset$, we have $I_N^0\subsetneq I_N$ as vector spaces.  The coproduct $\rho$ maps $I_N^0$ to $I_N^0 \otimes I_N$.

In what follows we will work with a modified version of $\rho$ where we carry around some additional structure.  Additionally, the modified $\rho$ will put some restrictions on $z$ which are too close to $x$ in the coproduct formula (in a sense that will be defined precisely in Section~\ref{sec global co}).  This will make it a coproduct only on a smaller space and a coaction more generally.

\subsection{Cointeraction motivation}\label{subsec global set up}
We begin with a graph and spanning tree $(\Gamma,T)$ which will remain fixed for the following.  As discussed in the section on Galois conjugates, Section~\ref{subsec galois}, we will be interested in the sitation where we have two disjoint sets $p$ and $q$ of edges of $T$ where we view $p$ as edges to contract and $q$ as edges to put on shell, or more combinatorially, edges to remove from $T$, leaving a spanning forest and hence inducing a cut of $\Gamma$.  That is, we are interested in $(\Gamma/p, T/p\backslash q)$.

For the cointeraction the graph structure itself will not be important.  All that we will need is that given an edge $e\in E_\Gamma\backslash E_T$ we have an associated subset of $E_T$, namely the fundamental cycle associated with $e$, but we will only need that it is a subset.  Let $f:(E_\Gamma\backslash E_T)\rightarrow \mathcal{P}(E_T)$ be the function giving this association, where $\mathcal{P}$ indicates the power set.
In other words, $f(e)=t_e$, but it will be useful in the following to have function notation for this association.  Furthermore, define $E_L= E_\Gamma\backslash E_T$.

For the purposes of the Hopf algebra structure on $(\Gamma,T)$, the subgraphs that play a role are exactly those which are formed of a union of fundamental cycles, and so to characterize such a subgraph we only need to give a subset of $E_\Gamma\backslash E_T$. 

 Furthermore, in the output of the coproduct, the co-graph is also determined by a subset of $E_\Gamma\backslash E_T$, provided we also know the subgraph (as a subset).  The reason for this is that the subset for the subgraph tells us which fundamental cycles to contract and the subset for the co-graph tells us which remaining fundamental cycles to keep.  The same holds for the results of iterated coproducts, simply contract everything appearing to the left and let the subset determine which cycles to keep for the current graph.

In view of this, all we need are subsets $B$ of $E_\Gamma\backslash E_T$ to play the role of subgraphs and co-graphs.  To keep track of the tree edges which are cut or contracted, we need an interval in $\mathcal{P}(E_T)$ which we interpret as $[q, E_T\setminus p]$.  
For a subset $B$, the interval should be in $\mathcal{P}(f(B))$.

Finally, for the second bialgebra structure, we may want to indicate some edges of $E_L$ which are allowable as tadpoles\footnote{i.e. self-loops} while others are not.  Let $E_M\subseteq E_L$ be the set of edges which may not appear as tadpoles. This interpolates between the situation where all tadpoles are allowed, and the situation where all tadpoles vanish, as they do in a kinematic renormalization scheme or for massless tadpoles in dimensional regularization, and so allows one framework to cover both cases.  One could have a renormalization scheme where some tadpoles vanish but others do not, based, say, on the masses of the edges.  This would also fit the framework.

\subsection{Cointeracting bialgebras}\label{sec global co}
Given two sets $E_T\subseteq E_\Gamma$ and a function $f:(E_\Gamma\backslash E_T)\rightarrow \mathcal{P}(E_T)$ we will build a cointeracting bialgebra structure as follows.

The underlying vector space of the bialgebras is the span of the set of formal symbols $B^{[A_1, A_2]}$ where $B\subseteq E_L$ and $[A_1, A_2]$ is an interval in $\mathcal{P}(f(B))$.  Let this set be $\mathcal{A}$.

Write $A|_{A'}$ for the set $A$ restricted to the set $A'$, that is for $A\cap A'$.
The product is
\[
m(B_1^{[A_1, A_2]}, B_2^{[A_3, A_4]}) =
\begin{cases}
  (B_1\cup B_2)^{[A_1\cup A_3, A_2\cup A_4]} & \text{if } B_1\cap B_2=\emptyset,\  A_1|_{A'} = A_3|_{A'}\\
    & \quad \text{ and } A_2|_{A'} = A_4|_{A'}; \\
0 & \text{otherwise,}
\end{cases}
\]
where $A'=f(B_1)\cap f(B_2)$.
The unit for this product is $\One = \emptyset^{[\emptyset, \emptyset]}$.

The product looks a bit messy, but the idea is not so complicated.  We take the union of the fundamental cycles defining the subgraphs and the corresponding union of the intervals, provided the intervals agree on any edges which are shared between the fundamental cycles and provided no fundamental cycle appears in its entierty in both subgraphs.

The coproduct is
\[
\Delta_c(B^{[A_1, A_2]}) = \sum_{\substack{B_1\subseteq B \\ f(B_1)\cap A_1 = \emptyset}} B_1^{[A_1\cap f(B_1), A_2\cap f(B_2)]} \otimes (B\backslash B_1)^{[A_1\cap f(B\backslash B_1), A_2\cap f(B\backslash B_1)]}
\]
The counit for this coproduct is $\hat{\One}_{\Delta_c}(\One) = 1$ and $\hat{\One}_{\Delta_c}(B^{[A_1, A_2]}) = 0$ for $B\neq \emptyset$.

Let $\mathcal{A}_p$ be the subspace of $\mathcal{A}$ spanned by $B^{[\emptyset, A]}$.  With these operations we get a Hopf algebra $(\mathcal{A}_p, m, \Delta_c)$, which we can see either by directly checking the required properties.  
The coproduct $\Delta_c$ gives a coaction of $\mathcal{A}$ on $\mathcal{A}_p$, $\Delta_c: \mathcal{A} \rightarrow \mathcal{A}_p \otimes \mathcal{A}$, since $A_1\cap f(B_1)=\emptyset$ so the terms on the left hand side of the coproduct are always in $\mathcal{A}_p$.

In the graph case we see that $\mathcal{A}$ is closely related to $H_{GF}$ and $\mathcal{A}_p$ is closely related to $H_{GT}$.  However, they are not the same Hopf algebras.  The product is different as the $m$ defined here is not disjoint union, and furthermore $\Delta_{GF}$ does not agree with $\Delta_c$ off of $\mathcal{A}_p$ as there is not typically a $\Gamma\otimes \One$ term in $\Delta_c$ off $\mathcal{A}_p$.

Now we want to build a second bialgebra structure with the same product.  This second bialgebra is essentially the incidence bialgebra structure on $\mathcal{P}(E_T)$ with three adjustments.  First the set $B$ is carried along, and second the product is $m$, rather than the usual direct product of intervals.  We can see $m$ as the result of first taking the direct product and then modding out by the ideal which sets to 0 the products of intervals which are incompatible according to the definition of $m$.

The third adjustment is to bring in $E_M$, the edges which are not permitted to become tadpoles.  Note that an edge $e\in E_M$ would be a tadpole if all the edges in the fundamental cycle of $e$ are contracted, that is if $f(e)\subseteq f(B)\backslash A_2$ for the graph associated to $B^{[A_1, A_2]}$.

Let $\mathcal{A}_m$ be the subspace of $\mathcal{A}$ spanned by the $B^{[A_1, A_2]}$ for which there is no $e\in E_m\cap B$ such that $f(e)\subseteq f(B)\backslash A_2$.  $\mathcal{A}_m$ is the subspace where the graphs do not have any forbidden tadpoles.  Additionally define $\mathcal{A}_e$ to be the subspace of $\mathcal{A}$ spanned by the $B^{[A_1, A_2]}$ for which there is no $e\in E_m\cap B$ such that $f(e)\subseteq f(B)\backslash A_1$.  $\mathcal{A}_e$ is the subspace where no contractions of forest edges of the graphs can give forbidden tadpoles.

Note that $m:\mathcal{A}_m\otimes \mathcal{A}_m \rightarrow \mathcal{A}_m$ since for the product to be nonzero the upper ends of the two intervals must agree on the intersection of the images of $f$.  Similarly $m:\mathcal{A}_e\otimes \mathcal{A}_e \rightarrow \mathcal{A}_e$ by the analogous argument involving the lower ends of the two intervals.

With this set up, define
\[
\rho(B^{[A_1, A_2]}) = \sum_{\substack{A_1\subseteq A \subseteq A_2\\ \nexists e\in E_m \cap B, f(e)\subseteq f(B)\backslash A}} B^{[A_1, A]}\otimes B^{[A, A_2]}
\]

\begin{lem}
  $(\mathcal{A}_e, m, \rho)$ is a bialgebra and $\rho:\mathcal{A}_m \rightarrow \mathcal{A}_m \otimes \mathcal{A}_e$ is a coaction.
\end{lem}

\begin{proof}
  This can be checked directly.

  Alternately, for the first part notice that in $\mathcal{A}_e$ the $E_m$ condition plays no role, so if we follow the standard incidence Hopf algebra construction with the $B$s as extra decorations, then to move from here to $\mathcal{A}_e$ we only need to mod out by the terms which are set to $0$ in the definition of $m$.  These form a Hopf ideal since if the interval disagrees on $f(B_1)$ or $f(B_2)$ then at least one side of any term in the coproduct will also disagree, and if the $B$s are not disjoint then they remain not disjoint on both sides of each term of the coproduct.
  
For the second part, we can reinterpret the problem by calculating $\rho$ with $E_m =\empty$ and then setting any graphs with forbidden tadpoles to $0$ at the end.  Then Lemma~\ref{lem Hopf proj} gives the result.
  
\end{proof}

More than two bialgebras, we have cointeracting bialgebras in the sense of \cite{cointeractiontalk, cointpapers}. 

\begin{thm}\label{thmgenerators}
On $\mathcal{A}$
\begin{itemize}
\item $\rho(\One)=\One\otimes\One$.
\item $m_{1,3,24}\circ(\rho\otimes\rho)\circ \Delta_c=(\Delta_c\otimes\mathrm{id})\circ\rho$, with:
\[
m_{1,3,24}:\,\mathcal{A}\otimes \mathcal{A}\otimes \mathcal{A}\otimes \mathcal{A}\to \mathcal{A}\otimes \mathcal{A}\otimes \mathcal{A},
\]
\[
m_{1,3,24}(w_1,w_2,w_3,w_4)=w_1\otimes w_3\otimes (w_2 w_4).
\]
\item $\forall w_1,w_2\in \mathcal{A}$, $\rho(w_1 w_2)=\rho(w_1)\rho(w_2)$,
\item $\forall w\in\mathcal{A}$, $(\hat{\One}_{\Delta_c}\otimes \mathrm{id})\circ\rho(w)=\hat{\One}_{\Delta c}(w)\One $.
\end{itemize}
In particular, With the coaction $\rho$, $(\mathcal{A}_p, m,\Delta_c)$ and 
$(\mathcal{A}_e,m,\rho)$ are in cointeraction.
\end{thm}

\begin{proof}
  The first point is immediate from the definition.  For the second point we have
  \begin{align*}
    & (\Delta_c\otimes \text{id})\circ \rho (B^{[A_1, A_2]}) \\
    & = \sum_{\substack{A_1 \subseteq A \subseteq A_2 \\ \nexists e\in E_m\cap B, f(e)\subseteq f(B)\backslash A}} \Delta_c(B^{[A_1, A]}) \otimes B^{[A, A_2]} \\
    & = \sum_{\substack{A_1 \subseteq A \subseteq A_2 \\ B_1\subseteq B \\ f(B_1)\cap A_1 = \emptyset \\ \nexists e\in E_m\cap B, f(e)\subseteq f(B)\backslash A}} B_1^{[A_1\cap f(B_1), A \cap f(B_1)]} \otimes (B\backslash B_1)^{[A_1\cap f(B\backslash B_1), A \cap f(B\backslash B_1)]} \otimes B^{[A, A_2]} 
  \end{align*}
  and
  \begin{align*}
    & m_{1,3,24}\circ(\rho\otimes\rho)\circ \Delta_c(B^{[A_1, A_2]}) \\
    & =\sum_{\substack{B_1\subseteq B \\ f(B_1)\cap A_1 = \emptyset}} m_{1,3,24}(\rho(B_1^{[A_1\cap f(B_1), A_2 \cap f(B_1)]})\otimes \rho((B\backslash B_1)^{[A_1 \cap f(B\backslash B_1), A_2 \cap f(B\backslash B_1)]})) \\
    & = \sum_{\substack{B_1\subseteq B \\ f(B_1)\cap A_1 = \emptyset \\ A_1 \cap f(B_1) \subseteq A_3 \subseteq A_2 \cap f(B_1) \\ A_1 \cap f(B \backslash B_1) \subseteq A_4 \subseteq A_2 \cap f(B \backslash B_1) \\ \nexists e\in E_m \cap B_1, f(e) \subseteq f(B_1) \backslash A_3 \\ \nexists e \in E_m\cap (B\backslash B_1), f(e) \subseteq f(B\backslash B_1)\backslash A_4} } m_{1,3,24}(B_1^{[A_1\cap f(B_1), A_3]}\otimes B_1^{[A_3, A_2 \cap f(B_1)]}\\
    & \hspace{6cm}\otimes (B\backslash B_1)^{[A_1 \cap f(B\backslash B_1), A_4]}\otimes (B\backslash B_1)^{[A_4, A_2 \cap f(B\backslash B_1)]}) \\
    & = \sum_{\substack{A_1 \subseteq A \subseteq A_2 \\ B_1\subseteq B \\ f(B_1)\cap A_1 = \emptyset \\ \nexists e\in E_m\cap B, f(e)\subseteq f(B)\backslash A}} B_1^{[A_1\cap f(B_1), A \cap f(B_1)]} \otimes (B\backslash B_1)^{[A_1\cap f(B\backslash B_1), A \cap f(B\backslash B_1)]} \otimes B^{[A, A_2]}
  \end{align*}
  where the last equality is because the last term only survives when $A_3$ and $A_4$ agree on $f(B_1)$ and $f(B\backslash B_1)$ in which case we can take $A=A_3\cup A_4$ and obtain
  \[
  m(B_1^{[A_3, A_2 \cap f(B_1)]}, (B\backslash B_1)^{[A_4, A_2 \cap f(B\backslash B_1)]})
  = (B_1\cup (B\backslash B_1))^{[A_3\cup A_4, A_2]} = B^{[A, A_2]}
  \]
  since $A_2\subseteq f(B)$.

  For the third point      
  consider $w_1 = B_1^{[A_1, A_2]}$ and $w_2 = B_2^{[A_3, A_4]}$, by linearity proving the result for these suffices.  Now calculate.
  \[
  \rho(w_1)\rho(w_2)
  = \sum_{\substack{A_1\subseteq A_5\subseteq A_2\\A_3\subseteq A_6\subseteq A_4 \\\nexists e \in E_m\cap B_1, f(e)\subseteq f(B_1)\backslash A_5 \\ \nexists e\in E_m\cap B_2, f(e)\subseteq f(B_2)\backslash A_6}} B_1^{[A_1, A_5]}B_2^{[A_3, A_6]} \otimes B_1^{[A_5, A_2]}B_2^{[A_6, A_4]}
  \]
  If the conditions are not satisfied so that the product $w_1w_2$ is nonzero then at least one of the products in each term above is also 0, so the result is 0 on both side.  If the conditions are satisfied so that $w_1w_2$ is nonzero then
  \begin{align*}
    \rho(w_1w_2) & = \rho(B_1\cup B_2)^{[A_1\cup A_3, A_2\cup A_4]}\\
    & = \sum_{\substack{A_1\cup A_3\subseteq A\subseteq A_2\cup A_4\\\nexists e\in E_m\cap(B_1\cup B_2), f(e)\subseteq f(B_1\cup B_2)\backslash A}} (B_1\cup B_2)^{[A_1\cup A_3, A]} \otimes (B_1\cup B_2)^{[A, A_2\cup A_4]}
  \end{align*}
  while
  \begin{align*}
    \rho(w_1)\rho(w_2)
    & = \sum_{\substack{A_1\cup A_3\subseteq A\subseteq A_2\cup A_4 \\ \nexists e\in E_m\cap (B_1\cup B_2), f(e)\subseteq f(B_1\cup B_2)\backslash A}} (B_1\cup B_2)^{[A_1\cup A_3, A]} \otimes (B_1\cup B_2)^{[A, A_2\cup A_4]}
  \end{align*}
  where $A = A_5\cup A_6$ when $A_5\cap f(B_1)\cap f(B_2) = A_6 \cap f(B_1)\cap f(B_2)$ and all other terms vanish.
 
  For the fourth point, consider $w = B^{[A_1, A_2]}$.  All terms in $\rho(w)$ have $B$ as the set on both sides of the tensor, but $\hat{\One}_{\Delta_c}$ is only nonzero if $B=\emptyset$, so the left hand side is $0$ if $B\neq\emptyset$ and is $\One$ otherwise.  The same is true of the right hand side, immediately from the definitions.  The fourth point then follows by linearity of the maps.

  The cointeraction is precisely the four bulleted properties when considered in the appropriate range.
\end{proof}

Finally note that if we imposed an order on the edges of $\Gamma$, then this order could be carried through everything done above, hence the information needed for the sector decompositions (see Section~\ref{SecDec}) can be added to this set up if desired.

\subsection{Cointeraction via generators}
For any index set $J$ consider the commutative $\mathbb{Q}$-algebra $F_J:=\mathbb{Q}[x_i]_{i\in J}$
of polynomials $f\in F_J$ in the variables $x_i,i\in J$.

The constant polynomial $F_J\ni 1_F=1$  is the unit with regard to the product.

Split $J=J_1\dot{\cup} J_2$ and consider the subalgebra 
$F_1:=\mathbb{Q}[x_i]_{i\in J_1}$. 

It is a cocommutative bialgebra by setting on a single generator $x_i$, 
\[
\Delta_1(x_i)=x_i\otimes 1_F+1_F\otimes x_i,
\] 
for all $i\in J_1$ and we extend to products consistently.

The augmentation ideal $A$ is given by those functions $p\in \mathbb{Q}[x_i]_{i\in J}$
such that $p(0)=0$.

The counit is given by $\hat{\One}_{\Delta_c}(p)=0$, $p\in A$ and $\hat{\One}_{\Delta_c}(1_F)=1$.

$F_1$ coacts on $F_J$. We  denote this coaction by  $\Delta_c$ and define it on  $F_J$ by setting for a single generator 
\[
\Delta_c(x_i)=1_F\otimes x_i,\,i\in J_2,\,\Delta_c(x_i)=x_i\otimes 1_F+1_F\otimes x_i,\,i\in J_1,
\]
and again extend to products consistently.

If we define a monomial for a set $m\subseteq J$ by $x_m:=\prod_{i\in m} x_i$ we have
\be\label{setdone}
\Delta_c (x_m)=\sum_{\emptyset\subseteq p\subseteq (m\cap J_1)} x_p\otimes x_{m\setminus p}.\ee

Consider two  sets $M$ and $N$ and the incidence coalgebra structure 
$I_N$ as in Sec.(\ref{incidence}).
We also choose a map $f$ which assigns to each $e\in M$ and interval $u=[u_i,u_f]\in I_N$, $u_i\subseteq u_f$ a set $f(e,u)\subseteq u_f$.

For fixed chosen interval $u$, define 
\[
M_u=\{e\in M|f(e,u)\cap u_i=\emptyset\}.
\]

Now take as index set $J=M\times I_n$, $n\subseteq N$ so that $I_{n}\subseteq I_N$  and split $J=J_1\dot{\cup} J_2$, $J_1=M_u\times I_{n}$, $J_2=J\setminus J_1$.

We define a second coproduct $\rho: F_J\to F_J\otimes F_J$ on the algebra $F_J= F^M_n$ by
\[
\rho(x_{m,[v_i,v_f]})=\sum_{\substack{v_i\subseteq c\subseteq v_f \\ c\neq v_i \text{ if } m\in E_m}}x_{m,[v_i,c]}\otimes x_{m,[c,v_f]}, \,\forall x_{m,[v_i,v_f]}\in A_1.
\]
Breaking up by generators makes the tadpole condition on $\rho$ nice to state as we see from the relative simplicity of the previous formula.

Identifying $B^{[A,B]}$
with $\prod_{e\in B}x_{e,[A,B]}$, we have almost obtained $\mathcal{A}$ by generators.  The only difference is that the product in $\mathcal{A}$ is not the free commutative product.  Modding out by the monomials which are sent to $0$ in the product of $\mathcal{A}$ we thus obtain, by construction, a description of $\mathcal{A}$ in terms of generators.

The generators $x_{e,[A,B]}$ are particularly nice since they correspond to the fundamental cycles and cut fundamental cycles, potentally with further edges marked for contraction.  Because of their close connection to the fundamental cycles, the $x_{e,[A,B]}$ are the language we will primarily use when working with these cointeracting bialgebras in the main body of this paper.

\end{document}